\documentclass[times]{oupau}
\usepackage{graphicx,txfonts,bm,units,latexsym,yfonts}
\usepackage[table]{xcolor}

\begin{document}

\runningheads{Prodan}{Quantum transport: A study based on operator algebras}
\title{Quantum transport in disordered systems under magnetic fields: A study based on operator algebras}
\author{Emil Prodan}
\address{ Physics Department \\
              Yeshiva University \\
              245 Lexington Av.\\
              New York, NY 10016, USA}

\begin{abstract}
The linear conductivity tensor for generic homogeneous, microscopic quantum models was formulated as a noncommutative Kubo formula in Refs.~\cite{BELLISSARD:1994xj,Schulz-Baldes:1998vm,Schulz-Baldes:1998oq}. This formula was derived directly in the thermodynamic limit, within the framework of $C^*$-algebras and noncommutative calculi defined over infinite spaces. As such, the numerical implementation of the formalism encountered fundamental obstacles. The present work defines a $C^*$-algebra and an approximate noncommutative calculus over a finite real-space torus, which naturally leads to an approximate finite-volume noncommutative Kubo formula, amenable on a computer. For finite temperatures and dissipation, it is shown that this approximate formula converges exponentially fast to its thermodynamic limit, which is the exact noncommutative Kubo formula. The approximate noncommutative Kubo formula is then deconstructed to a form that is implementable on a computer and simulations of the quantum transport in a 2-dimensional disordered lattice gas in a magnetic field are presented. 
\end{abstract}
\received{XXX}
\maketitle

\section{Introduction}\label{intro}

The purpose of this work is to demonstrate that the noncommutative Kubo formula \cite{BELLISSARD:1994xj,Schulz-Baldes:1998vm,Schulz-Baldes:1998oq,Bellissard:2000lj,Spehner:2001xr,Bellissard:2002zd,BellissardLectNotesPhys2003cy,BoucletaJFuncAnalysis2005vi} for the transport coefficients of aperiodic solids in magnetic fields can be efficiently evaluated on a computer. Our  main message is that the noncommutative Kubo formula, which is now well known for its pivotal role in the theoretical foundation of charge transport in aperiodic systems \cite{Schulz-Baldes:1998vm,BarbarouxHenriPoinacre1999mx,Bellissard:2000lj,Bellissard:2000hk,Bellissard:2002zd,BellissardLectNotesPhys2003cy,NakanoRevMathPhys2002bn}, can also provide an effective numerical solution to the difficult problem of computing the transport coefficients for realistic quantum models of aperiodic systems. For disordered quantum lattice systems in magnetic fields, we derive a ``canonical" finite-volume approximation to the noncommutative Kubo formula that can be efficiently evaluated on a computer. The advantage of this approximate formula is that it does not require twisted boundary conditions and integration over a Brillouin torus but only one (periodic-) boundary condition, and yet the convergence to its thermodynamic limit, which is the exact noncommutative Kubo formula, happens exponentially fast.

The transport measurements have always been regarded as fine tools for probing and understanding the physical properties of solids but there have been particular times in solid state research when the transport measurements have taken the central stage. This was certainly the case when the Integer \cite{Klitzing:1980vh} and Fractional \cite{TSUI:1982mx} Quantum Hall Effects were discovered, or when the quasicrystals were discovered \cite{ShechtmanPRL1984vb}. The recent discovery of the topological insulators \cite{HALDANE:1988rh,Kane:2005np,Kane:2005zw,Bernevig:2006hl,Koenig:2007ko,Moore:2007ew,Fu:2007vs,Hsieh:2008vm} has brought again the transport measurements to the central stage. These materials are postulated to have non-vanishing transport coefficients even in the presence of strong disorder. A concentrated experimental effort towards proving this principle has generated an incredible amount of fine transport data for these materials. However, the goal has not been achieved yet, despite of several years of progress in sample preparation and characterization. As such, there is an urgent need to ``reverse engineer" the available experimental transport data to get a better understanding of the microscopic structure of the samples. This is contingent on our ability to simulate the quantum transport for disordered materials in the presence of magnetic fields, based on microscopic quantum models. Even though we do not present any simulations for topological insulators, we wanted to bring the subject into the reader's attention because this was the primary motivation behind our effort. Simulations for topological insulators, using the formalism developed in this work, are reported in Ref.~\cite{Xue2012fh}.

The noncommutative Kubo formula for aperiodic systems appeared in Ref.~\cite{BELLISSARD:1994xj}, where the context was that of the Integer Quantum Hall Effect. It played a crucial role in explaining this effect. In the same reference, the authors derived the zero temperature limit of the Hall conductivity and linked it to the noncommutative Chern invariant. In the past several years, we have evaluated the noncommutative Chern invariant on a computer for several disordered quantum lattice models \cite{Prodan2010ew,Shulman2010cy,Prodan2011vy,ProdanJPhysA2011xk,XuPRB2012vu}. These simulations revealed to us the extraordinary efficiency of this formula, which enabled us to compute this invariant for extremely large system sizes. We were also impressed with the accuracy of the approach for the quantization of the Chern invariant was observed to occur with better than five digits of precision (in the presence of strong disorder!), whenever the simulation box exceeded the localization length of the system. 

The results mentioned above stimulated our interest in the finite temperature noncommutative Kubo formula for charge transport. As it is always the case, on a computer we can only simulate finite systems. The noncommutative Kubo formula has earned many accolades, among which the fact that it explicitly tells us what the thermodynamic limit of the linear conductivity is. However, from a computational point of view, this formulation in the thermodynamic limit presents a major difficulty because many elements appearing in this formula do not have a canonical equivalent at finite volumes. This brings up a more general question, namely, how to accurately compute correlation functions for aperiodic systems on a finite volume? By correlation function we mean the expected values of products of observables.\footnote{Since we are considering only the finite temperature regime, the observables in these products can be assumed to come from the analytic functional calculus with the Hamiltonian, which is assumed short-range.} 

The solution to the above question can be constructed in a $C^*$-algebraic setting. For this, we developed the concept of approximating $C^*$-algebras, which are simplified versions of some ``exact" $C^*$-algebras, with the essential requirement that approximate morphisms can be established between the exact and the approximating structures. These approximate morphisms can be used to transfer the correlation functions from the exact algebra to the approximate one and to ultimately compute the error bounds for this approximation process. Our strategy will be to apply the  concept of approximating algebras in two steps as illustrated in Fig.~\ref{ApproxDiagram}. The exact $C^*$-algebra is that of the observables defined in the presence of a uniform magnetic field (described by the antisymmetric tensor ${\bm F}$) on the entire $\mathbb{Z}^d$ lattice and on the entire disorder configuration space $\Omega$: $C^*(\Omega\times \mathbb{Z}^d,{\bm F})$. The first approximating algebra is that of the observables defined on the entire lattice but only on those disorder configurations leading to periodic potentials with a repeating cell $\Lambda$: $C^*(\Omega^\Lambda_{\mathrm{per}} \times \mathbb{Z}^d,{\bm F})$. The differential calculus transfers automatically over $C^*(\Omega^\Lambda_{\mathrm{per}} \times \mathbb{Z}^d,{\bm F})$ but a new trace is required. For the second step, $C^*(\Omega^\Lambda_{\mathrm{per}} \times \mathbb{Z}^d,{\bm F})$ becomes the exact algebra and the approximating algebra is that of the observables defined on a real-space torus of volume $\Lambda$ and on a specific disorder configuration space $\Omega_\Lambda$: $C^*(\Omega_\Lambda \times \mathbb{T}_d,{\bm F})$. For the latter algebra we need to define an approximate differential calculus and a new trace. The approximate morphisms provide a canonical way to approximate the correlation functions, in particular, the non-commutative Kubo formula. For both approximating steps, we derive explicit error bounds on the finite-volume approximate Kubo formula and establish that the errors vanish exponentially fast as $\Lambda$ is taken to infinity. This is the main result of our work.

\begin{figure}
\center
  \includegraphics[width=10cm]{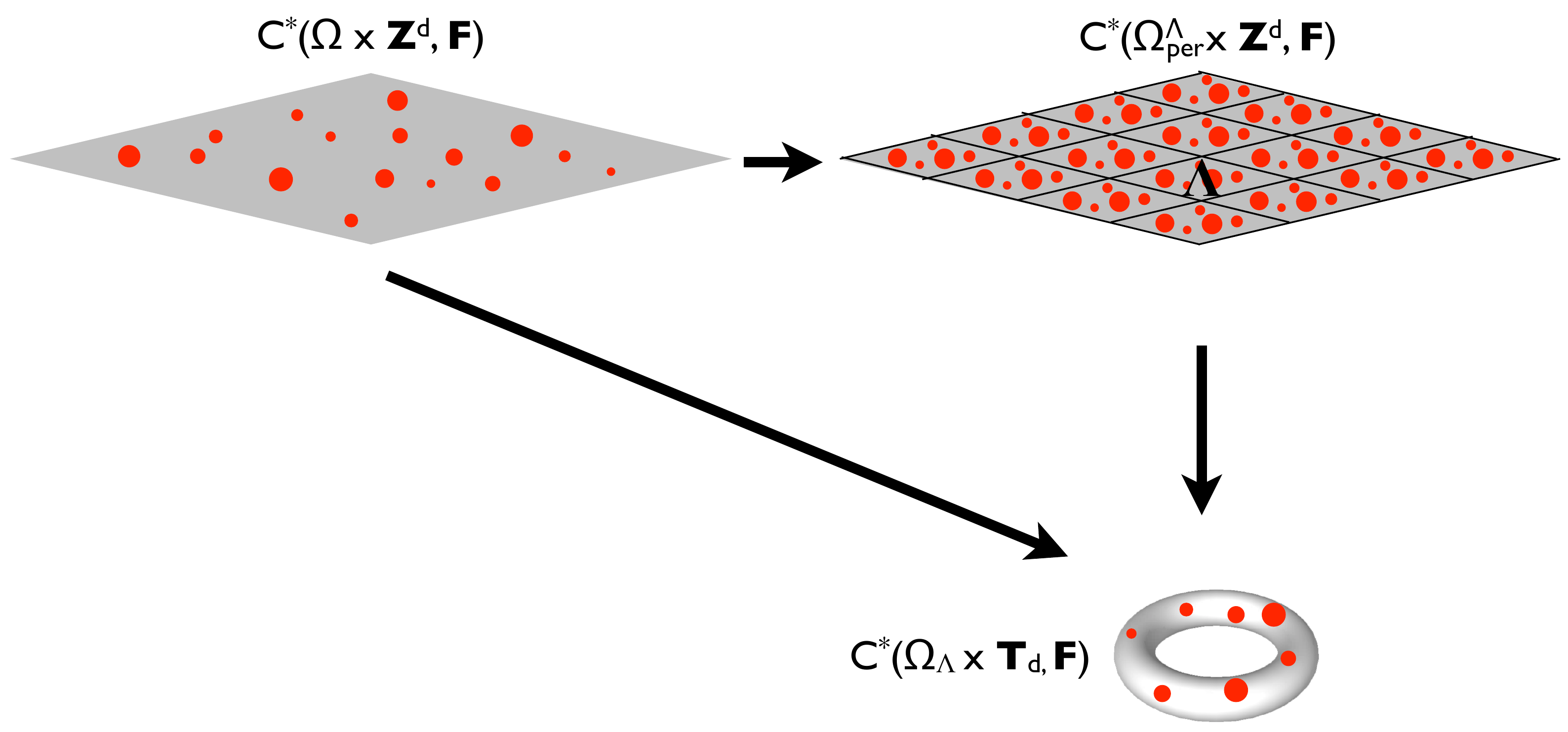}\\
  \caption{The diagram summarizes the sequence of approximations. The random potential, represented here by the size-varying dots, is approximated first by a periodically repeating potential, and then the whole system is considered on a real-space torus. Comparative estimates are derived for the pair of algebras $C^*(\Omega\times \mathbb{Z}^d,{\bm F})$ and $C^*(\Omega^\Lambda_{\mathrm{per}} \times \mathbb{Z}^d,{\bm F})$, and for the pair of algebras $C^*(\Omega^\Lambda_{\mathrm{per}} \times \mathbb{Z}^d,{\bm F})$ and $C^*(\Omega_\Lambda \times \mathbb{T}_d,{\bm F})$.}
 \label{ApproxDiagram}
\end{figure} 

Since the results could be of interest to the condensed matter physics community and to the pure, applied and computational mathematics communities, we decided to make the presentation self-contained. The paper is organized as follows. In Section 2 we discuss the quantum lattice models for electrons in solids and we argue that, as long as one is interested in the linear response regime, these models should not be taken as some crude approximation of the physical reality, but rather as accurate representations of it, with a systematic way of improving the accuracy. In Section 2 we also introduce an explicit model that will be used throughout the paper.  

In Sections 3 and 4 we review the noncommutative theory of charge transport in aperiodic solids developed by Bellissard, Schulz-Baldes and their collaborators \cite{BELLISSARD:1994xj,Schulz-Baldes:1998vm,Schulz-Baldes:1998oq}, touching only on the essential points. The interested reader can consult the comprehensive lecture notes by Bellissard \cite{BellissardLectNotesPhys2003cy} for a detailed exposition of the subject. Section 3 introduces the noncommutative $C^*$-algebra of the observables for a disordered electron gas in a magnetic field, and gives a short discussion of the functional calculus for $C^*$-algebras. It also introduces the noncommutative differential and integral calculus on this $C^*$-algebra. Together, these structures extend the notion of Brillouin torus to that of a noncommutative manifold, the famous noncommutative Brillouin torus \cite{BELLISSARD:1994xj}. Section 4 reviews the basics of the noncommutative theory of electron transport and gives the derivation of the noncommutative Kubo formula. In both Section 3 and Section 4 we follow closely the work of Bellissard, Schulz-Baldes and their collaborators \cite{BELLISSARD:1994xj,Schulz-Baldes:1998vm,Schulz-Baldes:1998oq,Bellissard:2000lj,Spehner:2001xr,Bellissard:2002zd,BellissardLectNotesPhys2003cy}.

Section 5 develops the Combes-Thomas technique \cite{Combes:1973nx} for the noncommutative Brillouin torus. This technique, which is well known to the analysts working with the Sch\"oedinger equation, will allow us to derive general and explicit exponential localization estimates for the analytic functional calculus with short-range observables. Section 6 introduces the periodic $C^*$-algebra $C^*(\Omega^\Lambda_{\mathrm{per}} \times \mathbb{Z}^d,{\bm F})$ and the noncommutative integro-differential calculus on it. It then presents the comparative estimates for the algebras $C^*(\Omega\times \mathbb{Z}^d,{\bm F})$ and $C^*(\Omega^\Lambda_{\mathrm{per}} \times \mathbb{Z}^d,{\bm F})$. 

Section 7 introduces the $C^*$-algebra $C^*(\Omega_\Lambda \times \mathbb{T}_d,{\bm F})$ over the torus. It also defines the noncommutative integral calculus and an approximate noncommutative differential calculus for this algebra. It then presents the comparative estimates for the algebras $C^*(\Omega^\Lambda_{\mathrm{per}} \times \mathbb{Z}^d,{\bm F})$ and $C^*(\Omega_\Lambda \times \mathbb{T}_d,{\bm F})$. Section 8 introduces the approximate noncommutative Kubo formula on the torus and derives the error estimates, which are shown to vanish exponentially fast in the thermodynamic limit. Section 8 also discusses the numerical implementation of the approximate noncommutative Kubo formula. Section 9 presents a numerical application to the Integer Quantum Hall Effect and Section 10 summarizes the main results of our work. 

\section{Lattice models for solids}\label{Settings}

The Hilbert space of a generic quantum lattice model is spanned by vectors of the form $|{\bm n},\alpha \rangle $, where ${\bm n}\in \mathbb{Z}^d$ is a node of the lattice and $\alpha=1,\ldots,D$ is a label for the orbitals associated with that node. The scalar product $\langle \cdot,\cdot \rangle$ is by definition such that $ \langle {\bm n},\alpha | {\bm m},\beta \rangle = \delta_{{\bm n}{\bm m}}\delta_{\alpha \beta}$. Intuitively, one can think of $|{\bm n},\alpha \rangle $ as the atomic orbitals of the elements present in the ${\bm n}$-th repeating cell of a material. Formally, the vectors $|{\bm n},\alpha \rangle $ can be viewed as the standard basis for the $\ell^2(\mathbb{Z}^d, \mathbb{C}^D)$ space. 

In the modern electronic structure theory, the quantum lattice models emerge as ``low energy" approximations of the continuum quantum models. For perfectly periodic solids, there is a standard machinery for generating such lattice models from continuum models. In short, the construction starts from an effective Hamiltonian $H_0$ which describes the equilibrium many-body state of the electrons in a material. A relevant spectral domain of $H_0$ is identified, usually by defining an energy cut-off, and the subspace corresponding to this spectral domain is represented in the so called maximally localized Wannier basis functions $\phi_{{\bm n}\alpha}$ \cite{MarzariPRB1997}. The number of Wannier functions per repeating cell depends on the number of atoms present in the repeating cell and on the energy cut-off. The vectors $|{\bm n},\alpha \rangle $, $\alpha =1,\ldots,D$, now represent the Wannier functions centered inside the ${\bm n}$-th repeating cell and the lattice Hamiltonian is simply:
\begin{equation}\label{LatticeHam}
H_{\mathrm{lattice}}=\sum_{{\bm n},\alpha}\sum_{{\bm m},\beta} |{\bm n},\alpha \rangle ( \phi_{{\bm n},\alpha},H_0 \phi_{{\bm m},\beta}) \langle {\bm m},\beta |,
\end{equation}
where $(\cdot,\cdot)$ is the scalar product of the continuum model. The above lattice Hamiltonian, although quadratic, it already incorporates interaction effects through $H_0$, which is tuned to exactly reproduce certain many-electron properties, such as the electron density for the case of the widely used Density Functional Theory \cite{HOHENBERG:1964vy,KOHN:1965cj}. The interaction effects, not accounted for already, are incorporated in the transport theory through the dissipation operator, to be discussed in Section 4.

The individual terms inside the sums in Eq.~\ref{LatticeHam} are called hoppings, and the matrix elements $( \phi_{{\bm n},\alpha} ,H_0 \phi_{{\bm m},\beta} )$ the hopping amplitudes. Generically, the maximally localized Wannier functions have a rapid spatial decay so the lattice Hamiltonians usually contain only a finite number of hopping terms between one lattice node and the rest of the nodes. It is quite often the case that just the hoppings between the nearest neighboring cells provide a reasonably accurate representation. If this is the case, we will use the terminology ``nearest-neighbor hopping Hamiltonian." Another important observation is that the lattice Hamiltonian reproduces exactly the band spectrum of the continuum Hamiltonian $H_0$, below the energy cut-off. Furthermore, the coupling integrals with the external driving fields can be exactly computed for the Wannier representation and even a lattice representation of a realistic disorder can be generated \cite{BerlijnPRL2011vu}. The main point of these observations is that we have a standard procedure to generate lattice models for periodic or disordered solids in external fields, whose accuracies can be systematically improved. Of course, the above statement is valid in the linear response regime where the representation based on the equilibrium Wannier functions makes sense. Several versions of this procedure are already implemented in the QUANTUM-ESPRESSO \cite{GiannozziJPCM2009fd} first principle electronic structure calculations package, and applications to electron transport in molecular structures have already appeared in the literature \cite{LeePRL2005se}.

The main conclusion of the present work, namely, that there is a canonical way to construct exponentially fast converging finite-volume approximations of the finite-temperature linear response coefficients, applies to a large class of disordered quantum lattice models. The essential characteristics of these models are: 1) existence of a triplet (a classical dynamical system) $(\Omega,dP(\omega),\mathfrak{t})$, where $(\Omega,dP(\omega))$ is a compact metrizable probability space and $\mathfrak{t}$ an ergodic, probability preserving action of the $\mathbb{Z}^d$ group on $\Omega$, and 2) existence of a bounded family of lattice Hamiltonians $H_\omega$ indexed by the points of $\Omega$, with the property of being covariant: $U_a H_\omega U_a^{-1}=H_{\mathfrak{t}_a\omega}$. Here $U_a$ represents the magnetic translation by $a$ \cite{BELLISSARD:1994xj}. However, in order to build our final arguments, we will need to fully identify all the constants and the parameters appearing in our error bounds. For this reason, we decided to simplify the exposition by working with a specific lattice model. This model has just one orbital per node and, with disorder and magnetic field, the Hamiltonian of the model reads: 
\begin{equation}\label{MainModel0}
H_\omega =  \sum_{{\bm n}\in \mathbb{Z}^d}\sum_{|{\bm p}|=1} e^{i \varphi_{{\bm p} {\bm n}}} \ |{\bm n}\rangle   \langle {\bm n}+{\bm p}| +W\sum_{{\bm n}\in \mathbb{Z}^d} \omega_{\bm n} |{\bm n}\rangle \langle {\bm n}|.
\end{equation}
The phase factor $ e^{i \varphi_{{\bm p} {\bm n}}}$ encodes the effect of the magnetic field through the Peierls substitution \cite{PanatiCMP2003rh}. Its specific form will be discussed in Section 3. Note that, in a short range isotropic model, all the hopping amplitudes must be equal in the absence of the magnetic field, in which case they can be set to unity like in Eq.~\ref{MainModel0} by choosing appropriate energy units.

The last term in Eq.~\ref{MainModel0}:
\begin{equation}
V_\omega=W\sum_{{\bm n}\in \mathbb{Z}^d} \omega_{\bm n} |{\bm n}\rangle \langle {\bm n}|,
\end{equation}
represents the disorder potential. This is just a simplified but very popular representation of the real effect of disorder in a solid. We will assume the amplitudes $\omega_{\bm n}$ to be independent random variables uniformly distributed within the interval $\left[-\frac{1}{2},\frac{1}{2}\right]$. The collection of amplitudes $\{\omega_{\bm n}\}_{{\bm n} \in \mathbb{Z}^d}$ can be viewed as a point $\omega$ in the space $\Omega=[-1/2,1/2]^{\mathbb{Z}^d}$. It is now quite standard \cite{BELLISSARD:1994xj} that $\Omega$ is a compact metrizable space, which accepts the following probability measure:
\begin{equation}
dP(\omega)=\prod_{{\bm n}\in \mathbb{Z}^d} d\omega_{\bm n}.
\end{equation}
There is a natural action of the additive $\mathbb{Z}^d$ discrete group on $\Omega$:
\begin{equation}
\mathfrak{t}_{\bm m}:\Omega \rightarrow \Omega, \ (\mathfrak{t}_{\bm m}\omega)_{\bm n}=\omega_{{\bm n}+{\bm m}},
\end{equation}
which is ergodic and leaves the measure $dP(\omega)$ invariant. Hence, the model of Eq.~\ref{MainModel0} gives a concrete representation of the generic quantum systems mentioned in the lines above.

\section{The noncommutative Brillouin torus}

In the Heisenberg's picture of Quantum Mechanics, the central objects are the observables, which are represented as elements of specific algebras. The main shift from the Schroedinger picture is the focus on the relations between the observables themselves, rather than on how the observables act on the wave functions. As such, the physics can be entirely formalized at a purely algebraic level. In this section, we follow Ref.~\cite{BELLISSARD:1994xj} and introduce the $C^*$-algebra of the observables for a disordered electron gas in a uniform magnetic field, and the noncommutative differential and integral calculus over this algebra. We expedite the presentation and mention only what is strictly needed for the derivation of the noncommutative Kubo formula at finite temperatures.  

\subsection{\underline{The $C^*$-algebra for a disordered electron gas in a uniform magnetic field}}

The Hamiltonian and its spatial translates generate an algebra ${\cal A}$, which in some sense is the smallest (and therefore the most useful) space to carry the calculations in. It contains all the observables of interest and it provides just the right framework to compute the correlation functions at finite temperatures.  Our first goal is to describe this algebra. 

The effect of a uniform magnetic field in a $d$-dimensional lattice model is captured entirely by a $d \times d$ antisymmetric matrix ${\bm F}$ \cite{BELLISSARD:1994xj}.  For example, in 2-dimensions, $|F_{12}|$ is equal to the magnetic flux through the repeating cell. Our units will be chosen such that the magnetic flux is expressed in the units of magnetic flux quantum of $\phi_0=h/e$. Also, we will use the notation ${\bm x}\cdot {\bm F} \cdot {\bm y}$ as a shorthand for $\sum_{j,k} x_j F_{jk} y_k$. 

The starting point of the construction is the algebra ${\cal A}_0$, generated by the continuous functions with compact support defined on $\Omega \times {\mathbb Z}^d$ with values in ${\mathbb C}$, and with the following addition and multiplication rules:
\begin{equation}\label{AlgRules}
\begin{array}{l}
(\alpha f)(\omega,{\bm n}) = \alpha f(\omega,{\bm n}) \ \ \mbox{(multiplication by scalars),} \medskip \\
(f+g)(\omega,{\bm n})=f(\omega,{\bm n})+g(\omega,{\bm n}) \ \ \mbox{(addition),} \medskip \\
(f*_{\bm F}g)(\omega,{\bm n})=\sum_{{\bm m} \in {\mathbb Z}^d} f(\omega, {\bm m})g(\mathfrak{t}_{{\bm m}}^{-1}\omega,{\bm n}-{\bm m})e^{i \pi( {\bm n}\cdot {\bm F}\cdot {\bm m})} \ \ \mbox{(multiplication).}
\end{array}
\end{equation} 
The index ${\bm F}$ attached to the multiplication operation is there to remind us that the multiplication depends explicitly on the magnetic field. The algebra ${\cal A}_0$ has an identity, denoted by ${\bm 1}$ and defined by: 
\begin{equation}
{\bm 1}(\omega,{\bm n})=\delta_{{\bm n},{\bm 0}}.
\end{equation}

Every element $f$ from ${\cal A}_0$ generates a covariant family of bounded linear operators $\{\pi_\omega f\}_{\omega\in\Omega}$ acting on $\ell^2({\mathbb Z}^d)$:
\begin{equation}\label{OpRep}
\ell^2({\mathbb Z}^d) \ni \phi \rightarrow \left((\pi_\omega f)\phi \right)({\bm n})=\sum_{{\bm m} \in {\mathbb Z}^d} f(\mathfrak{t}^{-1}_{\bm n}\omega, {\bm m}-{\bm n})e^{ i \pi ({\bm m} \cdot {\bm F}\cdot {\bm n})}\phi({\bm m}).
\end{equation}
This is the representation of $f\in {\cal A}_0$ as a linear operator acting on the quantum states, for a specific disorder configuration $\omega$. The covariant property means:
\begin{equation}
U_{\bm a} (\pi_\omega f) U_{\bm a}^{-1} = \pi_{\mathfrak{t}_{\bm a}\omega}f,
\end{equation}
where $U_{\bm a}$ are the magnetic translations on $\ell^2({\mathbb Z}^d)$:
\begin{equation}
\left(U_{\bm a}\phi\right)({\bm n})=e^{ i \pi ({\bm a}\cdot {\bm F}\cdot {\bm n})}\phi({\bm n}-{\bm a}).
\end{equation}
The map $\pi_\omega$ is a (faithful) representation of ${\cal A}_0$ algebra: 
\begin{equation}\label{faithful}
\pi_\omega(f*_{\bm F} g)=(\pi_\omega f) (\pi_\omega g).
\end{equation} 
The representation of the identity on $\ell^2({\mathbb Z}^d)$, $\pi_\omega {\bm 1}$, is simply the identity operator on $\ell^2({\mathbb Z}^d)$. Any short-range Hamiltonian can be generated from an element residing inside this algebra. For example, the first-neighbor hopping model of Eq.~\ref{MainModel0} corresponds to the element:
\begin{equation}\label{MainModel}
h(\omega,{\bm n})= \delta_{|{\bm n}|,1}({\bm n})+W\omega_{\bm n}\delta_{{\bm n}{\bm 0}},
\end{equation}
whose representation gives the covariant family of linear (self-adjoint) operators:
\begin{equation}
\pi_\omega h =  \sum_{{\bm n}\in \mathbb{Z}^d}\sum_{|{\bm p}|=1} e^{i \pi ({\bm p}\cdot {\bm F} \cdot {\bm n})} \ |{\bm n}\rangle   \langle {\bm n}+{\bm p}| +W\sum_{{\bm n}\in \mathbb{Z}^d} \omega_{\bm n} |{\bm n}\rangle \langle {\bm n}|.
\end{equation} 
The Peierls phase factor can be identified now as $\varphi_{{\bm p}{\bm n}}=\pi ({\bm p}\cdot {\bm F} \cdot {\bm n})$.

The algebra ${\cal A}_0$ is too small to include interesting elements, such as the time evolution $\exp(-ith)$, which does not have a compact support. However, all the interesting elements can be generated from ${\cal A}_0$ via limiting procedures and to rigorously define these limiting procedures we need to introduce first a norm. The norm on ${\cal A}_0$ is defined by:
\begin{equation}\label{Norm}
\|f\|=\sup_{\omega\in \Omega} \|\pi_\omega f\|,
\end{equation}
where the norm on the right-hand side refers to the operator norm on $\ell^2({\mathbb Z}^d)$: 
\begin{equation}
\|Q\|=\sup\limits_{\| \phi \|_{\ell^2({\mathbb Z}^d)}=1} \|Q\phi\|_{\ell^2({\mathbb Z}^d)}, \ \ \|\phi\|_{\ell^2({\mathbb Z}^d)}=\left[\sum_{{\bm n}\in \mathbb{Z}^d} |\phi({\bm n})|^2 \right ]^{1/2}.
\end{equation}
The following is an upper bound on the norm of an element $f\in {\cal A}_0$, which is useful because involves only a simple estimate on the kernel of $f$:
\begin{equation}
\|f\| \leq \left [ \sup_{\omega \in \Omega}\left \{\sum_{{\bm n}\in \mathbb{Z}^d} |f(\omega,{\bm n})| \right \} \ \sup_{\omega \in \Omega} \left \{\sum_{{\bm n} \in \mathbb{Z}^d} |f(\mathfrak{t}_{\bm n}\omega,{\bm n})| \right \} \right ]^{1/2}.
\end{equation}
Quite often, we will use the following simplified version of this upper bound:
\begin{equation}\label{NormA}
\|f\| \leq \sum_{{\bm n} \in \mathbb{Z}^d} \sup_{\omega \in \Omega} |f(\omega,{\bm n})|.
\end{equation}

Given that $\pi$ is a representation, and from the standard properties of the operator norm, the norm defined in Eq.~\ref{Norm} can bee seen to have the following special property:
\begin{equation}
\|f*_{\bm F} g\| \leq \|f\| \ \|g\|,
\end{equation}
which makes $({\cal A}_0,\| \ \|)$ into a normed-algebra. Also note that $\| {\bm 1} \|=1$. The completion of ${\cal A}_0$ under this norm becomes a Banach algebra, and this Banach algebra is denoted by ${\cal A}$. To make the Banach algebra ${\cal A}$ into a $C^*$-algebra we need to define first a $*$-operation, i.e. an operation with the following required properties:
\begin{equation}\label{StarProp}
f^{**}=f, \ \ (f+\alpha g)^*=f^*+\overline{\alpha} g^*, \ \ (f*_{\bm F}g)^*=g^* *_{\bm F} f^*, \ \ \ \|f^*\|=\|f\|. 
\end{equation}
Throughout our presentation, $\bar{z}$ will denote the complex conjugation of $z\in \mathbb{C}$. The following $*$-operation defined on ${\cal A}$ satisfies all the conditions stated in Eq.~\ref{StarProp}:
\begin{equation}\label{Star}
{\cal A} \ni f \rightarrow f^*, \ \ (f^*)(\omega,{\bm n})=\overline{f(\mathfrak{t}^{-1}_{\bm n}\omega,-{\bm n})}.
\end{equation}
Furthermore, this $*$-operation satisfies the fundamental relation:
\begin{equation}
\|f *_{\bm F} f^*\|= \|f\|^2,
\end{equation}
which makes ${\cal A}$ into a $C^*$-algebra. We collect all the above into the following central definition.

\begin{definition} The $C^*$-algebra of the observables for the disordered electron gas in a uniform magnetic field (represented by the antisymmetric tensor ${\bm F}$) is defined by the quadruple $({\cal A},*_{\bm F},\| \ \|,^*)$. We will use the notation $C^*(\Omega\times \mathbb{Z}^d,{\bm F})$ for this $C^*$-algebra. When there is no danger of confusion, we will just use ${\cal A}$.
\end{definition}

The computations at finite temperatures can be carried entirely in the $C^*$-algebra defined above. For zero temperature, one needs to extend ${\cal A}_0$ algebra to a von Neumann algebra but since that is outside our scope, we stop here.

\noindent {\bf Remark.} In the following we will drop the label ${\bm F}$ from $*_{\bm F}$ and simply write the multiplication in ${\cal A}$ as $f*g$.

\subsection{\underline{Spectrum, the resolvent function, the analytic functional calculus and all that}}

This subsection reviews some basic notions in operator algebras that are absolutely needed for our discussion. A good reference for this part is the short course on Spectral Theory by Arveson \cite{ArvesonBook2002}. 

\begin{definition} The resolvent set of an element $f$ belonging to an algebra with identity is the set:
\begin{equation}
\rho(f) = \{z \in {\mathbb C}\ | \ f-z{\bm 1} \ \mbox{is invertible}\}.
\end{equation}
The spectrum of $f$ is the set:
\begin{equation}
\sigma(f) = \{z \in {\mathbb C}\ | \ f-z{\bm 1} \ \mbox{is not invertible}\}.
\end{equation}
Clearly, $\rho(f)=\mathbb{C}-\sigma(f)$. 
\end{definition}
The notion of the spectrum exists for any algebra with identity, but the characterization of the spectrum becomes more refined (and therefore more useful) if the algebra is endowed with additional structures. In particular:

\begin{proposition} For any element $f$ in a Banach algebra with unit, $\sigma(f)$ is a non-empty compact subset of the complex plane. Furthermore, $\sigma(f)$ is located inside the circle of radius $\|f\|$ and centered at the origin. If the algebra is a $C^*$-algebra and $f$ is a self-adjoint element, i.e. $f^*=f$, then $\sigma(f)$ is located on the real axis. If $f$ is a unitary element, i.e. $f*f^*=f^**f=1$, then $\sigma(f)$ is located on the unit circle.
\end{proposition}  
The spectrum depends on the element itself but also on the algebra. As such, it  is very common to use the wording ``the spectrum of $x$ in the algebra y,'' something that we will do quite often in the followings. Another central object for operator algebras is the resolvent function:
\begin{definition}
Let $f$ be an element of a Banach algebra ${\cal A}$. The following function:
\begin{equation}
\rho(f) \ni z \rightarrow (f-z{\bm 1})^{-1} \in {\cal A} 
\end{equation}
is called the resolvent function of $f$.
\end{definition}

\begin{proposition}
The resolvent function of any element $f$ of a Banach algebra ${\cal A}$ is an algebra valued analytic function of $z \in \rho(f)$, in the sense that the limit
\begin{equation}
\lim_{\zeta \rightarrow 0} \frac{(f-(z+\zeta){\bm 1})^{-1}-(f-z{\bm 1})^{-1}}{\zeta}
\end{equation}
exists in ${\cal A}$ and is independent of how $\zeta$ approaches the origin, for all $z \in \rho(f)$.
\end{proposition}
Note that the above criterion of an analytic function, which we will use several times during our presentation, simply states that the function is differentiable and that the differential with respect to $\bar{z}$ is zero. Any other definition of an algebra valued analytic function is equivalent to the one considered above \cite{GamelinBook1984rt}. For Banach algebras, the notion of complex integration over $z$ of a continuously $z$-dependent element assumes the standard meaning via passing to the limit of the Riemann sum. Then complex integration and the resolvent function allows one to define the analytic functional calculus, i.e. a rigorous way to define functions of an element from the algebra. 

\begin{theorem}{(The analytic functional calculus)} Let $f$ be an element of a Banach algebra ${\cal A}$ and consider the algebra ${\cal F}_{\sigma(f)}$ of the analytic functions $\Phi$ in an open neighborhood of $\sigma(f)$ (with the pointwise addition and multiplication).  Then the following mapping:
\begin{equation}\label{AnalyticCalc}
{\cal F}_{\sigma(f)} \ni \Phi \rightarrow \frac{i}{2\pi} \oint_{\cal C} \Phi(z) (f-z{\bm 1})^{-1}dz \in {\cal A}
\end{equation}
is a homomorphism of algebras. Above, ${\cal C}$ is a contour surrounding $\sigma(f)$ and contained in the domain of $\Phi$. The right hand side of Eq.~\ref{AnalyticCalc} is regarded as $\Phi(f)$.
\end{theorem}

The functional calculus can be pushed further to the algebra of continuous functions or even further to the algebra of bounded measurable functions. The latter will be needed if, for example, one does calculations at zero temperature where the Fermi-Dirac distribution has a jump at the Fermi energy. We however stop here, because we address only the finite temperature case.

\subsection{\underline{The noncommutative differential and integral calculus}}

In the case of translationally invariant models, the transport coefficients are computed with the help of the integro-differential calculus with the ordinary functions defined over the Brillouin torus. When disorder and magnetic fields are present, the commutative algebra of the ordinary functions over the Brillouin torus is replaced by the noncommutative $C^*$-algebra $C^*(\Omega\times \mathbb{Z}^d,{\bm F})$. Remarkably, one can define a noncommutative integro-differential calculus over $C^*(\Omega\times \mathbb{Z}^d,{\bm F})$, which transforms the $C^*$-algebra into a noncommutative manifold, the famous noncommutative Brillouin torus \cite{BELLISSARD:1994xj}.

For this, the ${\bm k}$-integration over the classical Brillouin torus is replaced by a trace over $C^*(\Omega\times \mathbb{Z}^d,{\bm F})$:
\begin{equation}\label{Trace}
{\cal T}(f)=\int_\Omega dP(\omega) \ f(\omega,{\bm 0}).
\end{equation}
By a trace, in general, is meant a linear functional with two basic properties:
\begin{enumerate}
\item Positivity: ${\cal T}(f*f^*)>0$ for all $f \neq 0$.
\item Cyclic-city: ${\cal T}(f*g)={\cal T}(g*f)$.
\end{enumerate}
The definition of the trace in Eq.~\ref{Trace} is natural since, for the translational invariant case, ${\cal T}$ reduces to the ordinary ${\bm k}$-integration over the classical Brillouin torus. Another important observation is that, likewise the ${\bm k}$-integration over the classic Brillouin torus, the trace defined in Eq.~\ref{Trace} allows one to compute the correlation functions directly in the thermodynamic limit, in the sense that:
\begin{equation}
\lim\limits_{\Lambda\rightarrow \infty}\frac{1}{|\Lambda|}\mathrm{Tr}_{\Lambda}\{(\pi_\omega f) (\pi_\omega g) \ldots \}={\cal T}(f*g\ldots),
\end{equation}
where $\mathrm{Tr}_\Lambda$ means the trace over the quantum states inside the box $\Lambda$. The above relation is an immediate consequence of the Birkhoff theorem and the ergodicity property of the translations on $\Omega$ \cite{SinaiBook1976uh}. Throughout our manuscript, $|S|$ will denote the cardinal of a set $S$, and $\chi_S$ will denote it's characteristic function.

The derivations with respect to the $k_j$'s of the ordinary functions defined over the Brillouin torus become automorphisms of algebra $C^*(\Omega\times \mathbb{Z}^d,{\bm F})$:
\begin{equation}
(\partial_j f)(\omega,{\bm n}) = i n_j f(\omega,{\bm n}), \ j=1,\ldots,d, \ (i=\sqrt{-1}).
\end{equation} 
Given a multi-index ${\bm \alpha}=(\alpha_1,\ldots,\alpha_d)$, we will use the notation $\partial_{\bm \alpha}f$ for $\partial_1^{\alpha_1}\ldots \partial_d^{\alpha_d}f$ and we will write $|{\bm \alpha}|$ for $\alpha_1+\ldots + \alpha_d$. We collect some of the fundamental properties of the derivations in the following statement.

\begin{proposition}\label{Derivation} For $j, \ k=1,\ldots,d$:
\begin{enumerate}
\item The derivations commute: 
\begin{equation}
\partial_j \partial_k=\partial_k \partial_j.
\end{equation}
\item The derivations are *-derivations:
\begin{equation}
\partial_j (f^*)=(\partial_j f)^*.
\end{equation}
\item The derivations satisfy the Leibniz rule: 
\begin{equation}
\partial_j(f*g)=(\partial_j f)*g + f*(\partial_j g).
\end{equation}
\item The operator representation is given by:
\begin{equation}
\pi_\omega (\partial_j f)= -i[x_j, \pi_\omega f],
\end{equation}
where ${\bm x}=(x_1,\ldots,x_d)$ is the position operator on $\ell^2(\mathbb{Z}^d)$.
\end{enumerate}
\end{proposition}

It is important to note that the automorphisms $\partial_j$'s cannot be extended to the whole  $C^*$-algebra $C^*(\Omega\times \mathbb{Z}^d,{\bm B})$. As such, $\partial_j$'s are not derivations in the strict $C^*$-algebra language \cite{SakaiBook1971gu}. However, it is now customary to refer to $\partial_j$'s as unbounded derivations (since they are un-bounded automorphisms in the true sense). Due to their simple action, the $\partial_j$ derivations can be integrated and one can easily find that they generate a $d$-parameter group of isometric isomorphisms:
\begin{equation}\label{Trans}
\mathfrak{u}_{\bm k}:{\cal A}\rightarrow {\cal A}, \ (\mathfrak{u}_{\bm k}f)(\omega,{\bm n})=e^{-i{\bm k}\cdot{\bm n}}f(\omega,{\bm n}).
\end{equation}
This group is not uniformly continuous or differentiable of ${\bm k}$ over the entire ${\cal A}$.  The dense subspace of ${\cal A}$ containing the elements for which ${\bm k} \rightarrow \mathfrak{u}_{\bm k}f$ is $N$-times continuously differentiable will be denoted by ${\cal C}^N({\cal A})$. An element $f$ belongs to ${\cal C}^N({\cal A})$ if and only if $\|\partial_{\bm \alpha} f\|<\infty$, for all $\alpha$ with $|{\bm \alpha}| \leq N$. Using the Leibniz rule and the fundamental property of the norm, one can easily show that if $f$ and $g$ belong to ${\cal C}^N({\cal A})$, then their product $f*g$ also belongs to ${\cal C}^N({\cal A})$. In other words, ${\cal C}^N({\cal A})$ is a dense sub-algebra of $C^*(\Omega\times \mathbb{Z}^d,{\bm B})$. Finally, for any $f \in {\cal C}^1({\cal A})$, we have:
\begin{equation}
\partial_{k_j} (\mathfrak{u}_{\bm k}f) = -\partial_j (\mathfrak{u}_{\bm k}f).
\end{equation}
These technicalities will be needed when the quantum time evolution in the presence of an electric field will be discussed.

A few important classical rules of calculus extend to the noncommutative calculus.
\begin{proposition}
The following identities hold true:
\begin{enumerate}
\item For $f\in {\cal C}^1({\cal A})$ and $f$ invertible in ${\cal A}$:
\begin{equation}\label{DiffId1}
\partial_j f^{-1}=-f^{-1} * (\partial_j f) * f^{-1}.
\end{equation} 
\item If $\Phi$ and $\Phi'$ are any differentiable functions on the complex plane and $f\in {\cal C}^1({\cal A})$, then:
\begin{equation} \label{DiffId2}
{\cal T}(\Phi(f)*\partial_j \Phi'(f) ) = 0.
\end{equation}
This tells us that the noncommutative manifold has no boundaries, like the classic Brillouin torus \cite{BELLISSARD:1994xj}.
\item (The noncommutative residue theorem \cite{Prodan:2009od}.) If $\Phi$ is an analytic function in a neighborhood of the unit circle and $u$ is any unitary element, then:
\begin{equation}\label{DiffId3}
{\cal T}(\Phi(u) \partial_j u) = a_{-1} {\cal T}(u^{-1}\partial_j u),
\end{equation}
where $a_{-1}$ is the coefficient of $z^{-1}$ in the Laurent expansion $\Phi(z)=\sum_{n=-\infty}^\infty a_n z^n$.
\end{enumerate}
\end{proposition}

It is also important for our discussion to mention that, once a trace is defined over a $C^*$-algebra, one can define the $L^p$ norms:
\begin{equation}
\|f\|_{L^p}=\left [{\cal T}(|f|^p)\right]^{1/p},
\end{equation}
and the corresponding $L^p({\cal A},{\cal T})$ Banach spaces, defined as the closure of ${\cal A}$ under the $L^p$-norms. A particularly important space is $L^2({\cal A},{\cal T})$, since it can be endowed with a scalar product
\begin{equation}
\langle f,g\rangle = {\cal T}(f^**g),
\end{equation}
which transforms $L^2({\cal A},{\cal T})$ into a Hilbert space. As we shall see, certain automorphisms of the algebra ${\cal A}$ can be extended and viewed as self-adjoint operators on the Hilbert space $L^2({\cal A},{\cal T})$. An extremely fortunate consequence of all these is that we can use the functional calculus for self-adjoint operators on Hilbert spaces to define and compute functions of automorphisms. We will use this observation in our numerical analysis.

\section{The noncommutative theory of electron transport}

Here we give a telegraphic exposition of the electron transport theory in aperiodic systems, developed by Bellissard, Schulz-Baldes and their collaborators \cite{BELLISSARD:1994xj,Schulz-Baldes:1998vm,Schulz-Baldes:1998oq}. The goal is to present an explicit derivation of the noncommutative Kubo formula written in Eq.~\ref{KuboFormula}, in the most minimalist fashion possible. For a detailed exposition of the subject, the interested reader can consult the comprehensive lecture notes \cite{BellissardLectNotesPhys2003cy} by Bellissard. Before proceeding, we need a few more standard results in the theory of $C^*$-algebras.

\subsection{\underline{Derivations and automorphisms}}

Let $g$ be a generic self-adjoint element of $C^*(\Omega\times \mathbb{Z}^d,{\bm F})$, $g^*=g$. Then $g$ defines a bounded $*$-derivation (i.e. with the properties 2 and 3 of Proposition~\ref{Derivation}) via the following action:
\begin{equation}
{\cal L}_g(f)=i(g*f-f*g).
\end{equation}
The reverse statement is also true for a simple $C^*$-algebra with unity \cite{SakaiBook1971gu}, namely, any bounded $*$-derivation is generated by a self-adjoint element of the algebra via the action written above. Any bounded derivation generates a one-parameter uniformly continuous group of isometric automorphisms \cite{SakaiBook1971gu}. To make this statement fully general, let us assume that $g$ depends smoothly on the real parameter $t$ and write $g(t)$.

\begin{proposition}\label{Evolution}
The equation for $u(t,t')$:
\begin{equation}
\partial_t u(t,t')=-{\cal L}_{g(t)} \circ u(t,t'), \ \ u(t',t')=1 \ \mbox{( = the identity automorphism)}
\end{equation}
has a unique solution, and this solution defines a one-parameter uniformly continuous group of isometric automorphisms on $C^*(\Omega\times \mathbb{Z}^d,{\bm F})$. In other words:
\begin{enumerate}
\item $u(t,t') \circ u(t',t'') = u(t,t'')$,
\item $\|u(t,t')f\|=\|f\|$,
\item $u(t,t)=1$,
\end{enumerate}
for any real $t$, $t'$ and $t''$.
\end{proposition}

Below is another generic result in operator algebras, which will be important for our discussion.

\begin{proposition}\label{Laplace} Let $u(t,t')$ be a one parameter group, invariant to the simultaneous translation of its arguments, i.e. $u(t,t')=u(t-t')$. If $\|u(t)\|\leq e^{at}$ for some $a\in {\mathbb R}$, then the Laplace transform of the one-parameter group is well defined and given by:
\begin{equation}
\begin{array}{c}
\int_0^\infty dt \ e^{-\zeta t}u(t) = (\zeta + {\cal L})^{-1},
\end{array}
\end{equation}
for all $\zeta \in {\mathbb C}$ with $\mbox{Re}(\zeta) > a$. Above, ${\cal L}$ is the generator of the one-parameter group $u$.
\end{proposition}
\begin{proof} 
\begin{equation}
\begin{split}
\int_0^\infty dt \ e^{-\zeta t}u(t) & = \lim_{\Delta t \rightarrow 0} \sum_{n=0}^\infty e^{-n \zeta \Delta t} u(n\Delta t)\Delta t \medskip \\
& =\lim_{\Delta t \rightarrow 0} \sum_{n=0}^\infty \left (e^{-\zeta \Delta t} u(\Delta t)\right )^n \Delta t \medskip \\
& = \lim_{\Delta t \rightarrow 0} \Delta t (1- e^{-\zeta \Delta t} u(\Delta t))^{-1}\medskip \\
& =(\zeta +{\cal L})^{-1}.\qed
\end{split}
\end{equation}
\end{proof}

\subsection{\underline{The quantum time evolution with and without a driving electric field}}

As we already mentioned, the effective lattice Hamiltonians $h$ considered in condensed matter physics involve a finite number of hopping terms from a site to its neighboring sites, thus $h$ is a self-adjoint element of ${\cal A}_0$. We need to clarify that $h$ describes the intrinsic properties of a material, i.e. it does not include the generators of the dynamics due to external fields. A short-range Hamiltonian defines a bounded *-derivation:
\begin{equation}
\begin{split}
& {\cal L}_h : C^*(\Omega\times \mathbb{Z}^d,{\bm F}) \rightarrow C^*(\Omega\times \mathbb{Z}^d,{\bm F}), \medskip \\
 & {\cal L}_h(f) =i(h*f - f*h).
\end{split}
\end{equation}
The one parameter group of automorphisms generated by ${\cal L}_h$ is the quantum time evolution in the absence of any external fields. It will be denoted by $u(t,t')$ (or $u(t-t')$). Explicitly, $u(t,t')$ is given by:
\begin{equation}\label{TE}
C^*(\Omega\times \mathbb{Z}^d,{\bm F}) \ni f \rightarrow u(t,t')f=e^{-i(t-t')h}*f*e^{it(t-t')h},
\end{equation}
where the elements $e^{\pm i(t-t')h} \in C^*(\Omega\times \mathbb{Z}^d,{\bm F})$ are defined through the functional calculus.

The quantum evolution in the presence of a uniform electric field ${\bm E}$ is not as straightforward. It is generated by the Hamiltonian:
\begin{equation}
h_{\bm E}=h+{\bm E}{\bm x},
\end{equation}
and it will be denoted by $u_E(t,t')$ (or $u_E(t-t')$). Here we set the electron charge to $-1$ (atomic units). This $h_{\bm E}$ is not an element of the algebra $C^*(\Omega\times \mathbb{Z}^d,{\bm F})$. Nevertheless, note that $h_E$ still defines a well-behaved $*$-derivation over the algebra ${\cal A}_0$:
\begin{equation}
{\cal L}_{h_{\bm E}}={\cal L}_h-{\bm E}{\bm \nabla}, \ {\bm \nabla}=(\partial_1,\ldots,\partial_d),
\end{equation}
which can be extended to the dense sub-algebra ${\cal C}^1({\cal A})$ of $C^*(\Omega\times \mathbb{Z}^d,{\bm F})$. Furthermore, ${\cal L}_{h_{\bm E}}$ generates a one parameter group of isometric automorphisms on $C^*(\Omega\times \mathbb{Z}^d,{\bm F})$. This can be established as follows.

\begin{proposition} Let $h\in {\cal A}_0$ and $\tilde{h}(t)=\mathfrak{u}_{t{\bm E}}h$, with $\mathfrak{u}$ defined in Eq.~\ref{Trans}. Then $\tilde{h}(t)$ is an element of ${\cal A}_0$ for all $t \in {\mathbb R}$. As a consequence, the equation in $v(t,t')$,
\begin{equation}
\partial_t v(t,t') =- {\cal L}_{\tilde{h}(t)} \circ v(t,t'), \ \ v(t',t')=1,
\end{equation}
defines a uniformly continuous group of isometric automorphisms $v(t,t')$ on $C^*(\Omega\times \mathbb{Z}^d,{\bm F})$ (cf. Proposition~\ref{Evolution}). Then:
\begin{equation}\label{uE}
\begin{split}
& u_E(t,t') : C^*(\Omega\times \mathbb{Z}^d,{\bm F}) \rightarrow C^*(\Omega\times \mathbb{Z}^d,{\bm F}) \medskip  \\
& u_E(t,t')=\mathfrak{u}_{t{\bm E}}^{-1} \circ v(t,t') \circ \mathfrak{u}_{t'{\bm E}},
\end{split}
\end{equation}
defines a one parameter group of isometric automorphisms on $C^*(\Omega\times \mathbb{Z}^d,{\bm F})$, whose generator is precisely ${\cal L}_{h_{\bm E}}$.
\end{proposition}

\proof We start by showing that:
\begin{equation}\label{Id1}
{\cal L}_{\tilde{h}(t)} = \mathfrak{u}_{t{\bm E}} \circ {\cal L}_h \circ \mathfrak{u}_{t{\bm E}}^{-1}.
\end{equation}
Indeed:
\begin{equation}
\begin{split}
\left  (\mathfrak{u}_{t{\bm E}} \circ {\cal L}_h \circ \mathfrak{u}_{t{\bm E}}^{-1}\right )(f) & =\mathfrak{u}_{t{\bm E}}\left ( {\cal L}_h(\mathfrak{u}_{t{\bm E}}^{-1} f) \right)\medskip \\
& =i \mathfrak{u}_{t{\bm E}}\left (h*(\mathfrak{u}_{t{\bm E}}^{-1} f) - (\mathfrak{u}_{t{\bm E}}^{-1} f) *h\right)\medskip \\
&=i (\mathfrak{u}_{t{\bm E}}h)*(\mathfrak{u}_{t{\bm E}}(\mathfrak{u}_{t{\bm E}}^{-1} f))-i(\mathfrak{u}_{t{\bm E}} (\mathfrak{u}_{t{\bm E}}^{-1}f) )*(\mathfrak{u}_{t{\bm E}}h) \medskip \\
& =i((\mathfrak{u}_{t{\bm E}}h)* f-f *(\mathfrak{u}_{t{\bm E}}h)),
\end{split}
\end{equation}
and the last line is precisely ${\cal L}_{\tilde{h}(t)}f$. Given the group properties of $v(t,t')$, it is straightforward to verify that $u_E$ defined in Eq.~\ref{uE} has all the three properties listed in Proposition~\ref{Evolution}. Furthermore, we have:
\begin{equation}
u_E(t,t')-1=\mathfrak{u}_{t{\bm E}}^{-1} \circ v(t,t') \circ (\mathfrak{u}_{t'{\bm E}}-\mathfrak{u}_{t{\bm E}}) + \mathfrak{u}_{t{\bm E}}^{-1} \circ (v(t,t')-1) \circ \mathfrak{u}_{t{\bm E}},
\end{equation} 
so the following equation holds true on ${\cal C}^1({\cal A})$:
\begin{equation}
\lim_{t' \rightarrow t} \frac{u_E(t,t')-1}{t-t'} = \mathfrak{u}_{t{\bm E}}^{-1} \circ ({\bm E}{\bm \nabla})  \circ \mathfrak{u}_{t{\bm E}}  - \mathfrak{u}_{t{\bm E}}^{-1} \circ {\cal L}_{\tilde{h}(t)} \circ \mathfrak{u}_{t{\bm E}}.
\end{equation}
If we use the identity from Eq.~\ref{Id1} and the fact that $\mathfrak{u}$ commutes with ${\bm \nabla}$, we can see that the right hand side is exactly $-({\cal L}_h-{\bm E}{\bm \nabla})$.\qed

\subsection{\underline{The collision processes}}

Dissipation is a result of a perpetual sequence of electron scattering events. A scattering event can be described by a scattering potential $w \in C^*(\Omega\times \mathbb{Z}^d,{\bm F})$, whose exact form depends on the dissipation mechanism that is being considered. The scattering potential will, in general, fluctuate from one scattering event to another but for simplicity we will neglect such fluctuations. As discussed and exemplified in the lecture notes of Ref.~\cite{BellissardLectNotesPhys2003cy}, $w$ can be derived from microscopic many-electron quantum models describing the electron system, the environment (such as a phonon bath) and the coupling between the two system, by integrating out the degrees of freedom of the environment. General arguments show that $w$ must commute with the Hamiltonian $h$.

Following Refs.~\cite{BELLISSARD:1994xj,Schulz-Baldes:1998vm,Schulz-Baldes:1998oq,BellissardLectNotesPhys2003cy},  we introduce the collision processes via a time-dependent potential:
\begin{equation}\label{ScattPot}
w(t)=\sum_{j \in \mathbb{Z}} \delta(t-t_j)w. 
\end{equation} 
A sequence of collision events can be indexed by a countable subset ${\cal E}$ of ${\mathbb R}$, where each point of ${\cal E}$ represents the collision time of a particular scattering event from the sequence. The lengths $s_j$ of the time intervals between two consecutive collision events $t_j$ and $t_{j+1}$ are assumed to be independent random variables taking values in ${\mathbb R}_+=(0,\infty)$ and identically distributed according to the Poisson law $e^{-s/\tau}ds/\tau$.  $\tau$ is called the collision time and it gives the average time between two consecutive collision events. Thus, the collision processes can be parametrized by a sequence:
\begin{equation}
{\bm s} = (\ldots, s_{-1},s_0,s_1 \ldots).
\end{equation}
The label $0$ will be attached to the time interval that contains the origin. To complete the parametrization, we need one more parameter, which is taken to be:
\begin{equation}
q=\mbox{dist}\{0,{\cal E}\cap {\mathbb R}_+\}.
\end{equation}
In other words, $q$ is the time coordinate of the collision event occurring right after $t=0$. As such, $q$ takes values in the interval $[0,s_0]$. With this parametrization, the collision times are given by $t_j=q+s_1 + \ldots + s_{j-1}$.

To summarize, the collision processes are parametrized by a point $\eta=({\bm s},q)$ of the space
\begin{equation}
\Xi = ({\mathbb R}_+)^{\mathbb Z}\times {\mathbb R}_+,
\end{equation}
which is endowed with the probability measure:
\begin{equation}
d\mu(\eta)= \left (\prod\limits_{j\in{\mathbb Z}} e^{-s_j/\tau}ds_j/\tau \right ) \left (\chi_{[0,s_0]}(q)d q/\tau \right),
\end{equation}
where $\chi_{[0,s_0]}(q)$ is the characteristic function of the interval $[0,s_0]$. We will use the shorthand $d\mu(s)$ for $e^{-s/\tau}ds/\tau$ and $d\mu({\bm s})$ for $\prod_{j\in{\mathbb Z}} d\mu(s_j)$. A crucial observation is the existence of a natural action of the time translation $\alpha_{\Delta t}$ on ${\cal E}$ which simply shifts each collision time by a $\Delta t$. This action defines a one parameter group of automorphisms on the probability space $\Xi$ that are ergodic and measure preserving. As such, the time averaging and the $\eta$-averaging coincide \cite{SinaiBook1976uh} or, in other words, the collision processes are self-averaging. In the following, we will attach a label $\eta$ and write $w_\eta(t)$ for the time dependent potential of Eq.~\ref{ScattPot}.

\subsection{\underline{The effective time evolution}}

The time evolution of the system in the presence of a driving electric field and of the collisions events is generated by the Hamiltonian:
\begin{equation}
h_{\eta,E}(t)=h +E{\bm x}+w_\eta (t).
\end{equation}
Since the collisions are assumed to take place instantaneously, we can write down the explicit expression of the time evolution:
\begin{equation}\label{QEvolution}
u_{\eta,E}(t,t')=u_E(t-\max\{t_j\}) \left (\prod\limits_{t\geq t_j \geq t_{j-1}\geq t'} \hat{w} u_E(t_j-t_{j-1}) \right ) u_E(\min\{t_{j-1}\}-t'),
\end{equation}
where the max/min is taken over the $t_j$'s appearing in the product. The automorphism $\hat{w}$ is defined as:
\begin{equation}
\hat{w}f=e^{iw}*f*e^{-i w}.
\end{equation}
Another way of writing the time evolution is (from here on we choose $t'=0$):
\begin{equation}\label{GEvolution}
u_{\eta,E}(t)=\chi_{(t_1,0)}(t) u_E(t) +  \sum\limits_{j=1}^\infty  \chi_{(t_{j+1},t_j)}(t) 
u_E(t-t_j)\hat{w} \ldots u_E(t_2-t_1) \hat{w} u_E(t_1).
\end{equation}
The time evolution is a covariant automorphism:
\begin{equation}
u_{\eta,E}(t-\Delta t,t'-\Delta t)=u_{\alpha_{\Delta t}\eta,E}(t,t'),
\end{equation}
and, as such, the averaging over the collision configurations:
\begin{equation}
\bar{u}_E(t,t')=\int_\Xi d \eta \ u_{\eta,E} (t,t'),
\end{equation}
results in a one-parameter group of automorphisms that are invariant to time translations: $\bar{u}_E(t,t')=\bar{u}_E(t-t')$. 

Our goal now is to compute the averaged time evolution $\bar{u}_E(t)$. We will evaluate its Laplace transform, $L[\bar{u}_E]$, starting from Eq.~\ref{GEvolution}. Considering each term separately, we need to evaluate:
\begin{equation}
\int_0^\infty dt
    \ e^{-\zeta t}   \int d\mu({\bm s})\int_0^{s_0}dq/\tau   \  \ \chi_{(t_j,t_{j-1})}(t)u_E(t-t_j) \hat{w}  \ldots u_E(t_2-t_1) \hat{w} u_E(t_1).
 \end{equation}
 This is same as: 
\begin{equation}
\begin{split}
& \int d \mu({\bm s})\int_0^{s_0}dq/\tau \int_{t_j}^{t_{j+1}} dt \ e^{-\zeta(t-t_j)}u_E(t-t_j) \hat{w}  \medskip \\
& \ \ \  \times \ e^{-\zeta(t_j-t_{j-1})}u_E(t_j-t_{j-1}) \hat{w} \ldots e^{-\zeta(t_2-t_1)}u_E(t_2-t_1) \hat{w} e^{-\zeta t_1}u_E(t_1) \medskip \\
& = \int d\mu(s_j)\int_{0}^{s_j} dt \ e^{-\zeta t} u_E(t) \hat{w} \int d\mu(s_{j-1}) e^{-\zeta s_{j-1}} u_E(s_{j-1}) \hat{w} \ldots\medskip \\
& \ \ \  \int d\mu(s_1) e^{-\zeta s_1} u_E(s_1) \hat{w} \int d\mu(s_0)\int_0^{s_0} dq/\tau e^{-\zeta q} u_E(q).
\end{split}
\end{equation}
The emerging integrals can be evaluated explicitly, by directly applying Proposition~\ref{Laplace} or slight variants of it:
\begin{equation}
\int d\mu(s_j)\int_{0}^{s_j} dt \ e^{-\zeta t} u_E(t)=\tau \left (1+\tau(\zeta+{\cal L}_{h_E})\right)^{-1},
\end{equation}
\begin{equation}
\int d\mu(s_j) \ e^{-\zeta s_j} u_E(s_j)=(1+\tau(\zeta+{\cal L}_{h_E}))^{-1}, \ (j=i-1,\ldots,1),
\end{equation}
and
\begin{equation}
\int d\mu(s_0)\int_{0}^{s_0} dq/\tau \ e^{-\zeta q} u_E(q)= (1+\tau(\zeta+{\cal L}_{h_E}))^{-1}.
\end{equation}
Thus:
\begin{equation}
L[\bar{u}](\zeta)=\tau (1+\tau(\zeta+{\cal L}_{h_E}))^{-1} \sum_{j=0}^\infty  \left ( \hat{w} (1+\tau(\zeta+{\cal L}_{h_E}))^{-1} \right )^{j},
\end{equation}
or
\begin{equation}
L[\bar{u}](\zeta)=\left (\zeta + (1-\hat{w} )/\tau + {\cal L}_{h_E} \right )^{-1}. 
\end{equation}
Given the statement of Proposition~\ref{Laplace}, we can now draw the important conclusion that the generator of the effective time evolution $\bar{u}_E(t)$ is:
\begin{equation}
 (1-\hat{w} )/\tau +{\cal L}_{h_E}.
\end{equation}
We introduce the notation $\Gamma$ for the collision operator $(1-\hat{w} )/\tau$, and finally write: 
\begin{equation}\label{ubar}
\bar{u}_E(t)=e^{-t\left(\Gamma +{\cal L}_{h_E}\right)}.
\end{equation}
 The spectral properties of the collision operator have been discussed in Ref.~\cite{BellissardLectNotesPhys2003cy} for various dissipation mechanisms. In our numerical applications we will use a simplifying approximation for $\Gamma$, called the relaxation time approximation, where $\Gamma$ becomes proportional to the identity automorphism. 

\subsection{\underline{The noncommutative Kubo formula}}

We assume that the electric field is turned on at $t=0$, when the electron system was still in its finite temperature equilibrium state. In this state, the expected value (per unit volume) of an observable $f \in C^*(\Omega\times \mathbb{Z}^d,{\bm F})$ is given by ${\cal T}(f * \rho_0)$, where:
\begin{equation}
\rho_0=\Phi_{\mathrm{FD}}(h)=\left (1+e^{(h-E_F)/kT} \right )^{-1},
\end{equation}
with $\Phi_{\mathrm{FD}}$ being the Fermi-Dirac distribution, $T$ being the temperature, $k$ the Boltzmann constant and $E_F$ the Fermi energy. The Fermi energy is fixed by the electron density: $n_e={\cal T}(\rho_0)$. After the electric field was turned on, the expected values of the observables are given by ${\cal T}(f * \rho_t)$, where $\rho_t=u_{\eta,E}(t,0)\rho_0$. As such, the charge current density at time $t$ is:
\begin{equation}\label{CurDen}
{\bm J}_{\eta,E}(t)={\cal T}\left (\hat{{\bm j}} * u_{\eta,E}(t,0)\rho_0 \right).
\end{equation}
The charge current operator $\hat{{\bm j}}$ is explicitly given by $\hat{{\bm j}}=-{\bm \nabla}h$. Note that above we have a trace over volume, so Eq.~\ref{CurDen} gives indeed the current density. 

Now, the average over the collision events gives:
\begin{equation}
{\bm J}_{E}(t)=\int_\Xi d\mu(\eta) \ {\cal T}\left (\hat{{\bm j}} * u_{\eta,E}(t,0)\rho_0\right)={\cal T}\left (\hat{{\bm j}} * \bar{u}_{E}(t)\rho_0 \right).
\end{equation}
The time average of this effective charge current density,
\begin{equation}
\bar{{\bm J}}_E=\lim_{t' \rightarrow \infty} \frac{1}{t'}\int_0^{t'} dt \ {\bm J}_E(t),
\end{equation}
can be computed as:
\begin{equation}
\bar{{\bm J}}_E = \lim_{\delta \rightarrow 0} \delta \int_0^\infty dt \  e^{-\delta t} \ {\cal T}\left (\hat{{\bm j}} * \bar{u}_E(t)\rho_0 \right).
\end{equation}
At this step we see the Laplace transform of $\bar{u}_E(t)$ appearing:
\begin{equation}
\lim_{\delta \rightarrow 0} \delta \int_0^\infty dt \  e^{-\delta t} \ {\cal T}\left (\hat{{\bm j}} * \bar{u}_E(t)\rho_0 \right ) =\lim_{\delta \rightarrow 0} \delta \ {\cal T}\left (\hat{{\bm j}} * L[\bar{u}](\delta)\rho_0 \right ).
\end{equation}
From Eq.~\ref{ubar}, we can conclude at this step that:
\begin{equation}\label{Partx}
\bar{{\bm J}}_E =\lim_{\delta \rightarrow 0} \delta \ {\cal T}\left (\hat{{\bm j}} * (\delta+\Gamma + {\cal L}_{h_E})^{-1}\rho_0 \right ).
\end{equation}
Since in the absence of an electric field the expected current is zero:
\begin{equation}
\delta \ {\cal T}\left (\hat{{\bm j}} * (\delta+\Gamma + {\cal L}_h)^{-1}\rho_0 \right )=0,
\end{equation}
we can subtract such null contribution from Eq.~\ref{Partx} to obtain:
\begin{equation}
\bar{\bm J}_E=\lim\limits_{\delta \rightarrow 0} \delta \ {\cal T}\left (\hat{{\bm j}} * (\delta+\Gamma + {\cal L}_{h_E})^{-1}\circ
({\bm E}{\bm \nabla})\circ(\delta+\Gamma + {\cal L}_h)^{-1}\rho_0 \right ).
\end{equation}
Finally, since $\Gamma$ and ${\cal L}_h$ commute with each other annihilate $\rho_0$, we see that the action of the last automorphism on $\rho_0$ reduces to $\delta^{-1}\rho_0$. At this point we can take the limit $\delta \rightarrow 0$ to obtain:
\begin{equation}\label{Current}
\bar{{\bm J}}_E=-{\cal T}\left (({\bm \nabla}h) * (\Gamma + {\cal L}_{h_E})^{-1}(
{\bm E}{\bm \nabla}\rho_0)\right).
\end{equation}
The conductivity tensor $\sigma_{jk}$ is defined as the link between the charge current density and the electric field:
\begin{equation}
J_j = \sum_{k=1}^d \sigma_{jk} E_k,
\end{equation}
so Eq.~\ref{Current} already allows us to identify its expression:
\begin{equation}
\sigma_{jk}({\bm E})=- {\cal T}\left ((\partial_k h) * (\Gamma + {\cal L}_{h_E})^{-1}\partial_k \Phi_{\mathrm{FD}}(h)\right).
\end{equation}
This expression depends on the applied field ${\bm E}$ through ${\cal L}_{h_E}$, hence it is usually called the nonlinear conductivity tensor. The limit $E \rightarrow 0$ defines the linear conductivity tensor which has the following expression:
\begin{equation}\label{KuboFormula}
\sigma_{jk}= -{\cal T}\left ((\partial_j h) * (\Gamma + {\cal L}_h)^{-1}
\partial_k \Phi_{\mathrm{FD}}(h) \right ).
\end{equation}
This is the famous noncommutative Kubo formula \cite{BELLISSARD:1994xj,Schulz-Baldes:1998vm,Schulz-Baldes:1998oq}.

\section{Analytic Functional Calculus: Exponential Localization}

We start here the discussion concerning the numerical evaluation of Eq.~\ref{KuboFormula}. The exponential localization estimates derived in this section are essential for the numerical implementation of the noncommutative Kubo formula. We will state our main results and the main conclusions for this section here at the beginning, filling the rest of the section with detailed proofs. 

\begin{theorem} \label{TH1}
Let $f\in {\cal A}_0$ be an invertible element in ${\cal A}$. Then $\partial_{\bm \alpha} f^{-1}$ is exponentially localized for any multi-index ${\bm \alpha}$. More precisely, let
\begin{equation}
c(\xi)=\sup\limits_{|\mathrm{Im}({\bm \xi})|=\xi} \sum_{{\bm n}\in \mathbb{Z}^d}  \ |(1-e^{-\mathrm{Im}({\bm \xi} \cdot {\bm n})})| \sup\limits_{\omega\in \Omega}|f(\omega,{\bm n}) |,
\end{equation}
Definitely, $c(\xi)$ is finite for any $f\in {\cal A}_0$ and $c(\xi)$ is monotonically increasing of $\xi$. Let $\bar{\xi} >0$ be defined by the equation $c(\bar{\xi})=1/\|f^{-1}\|$. Then, as long as $\xi < \bar{\xi}$, the following relation holds:
\begin{equation}\label{ExpMain}
|\partial_{\bm \alpha} f^{-1}(\omega,{\bm n})|\leq \frac{C_{f,{\bm \alpha}}(\xi) e^{-\xi|{\bm n}|}}{\left ( \|f^{-1}\|^{-1}-c(\xi)\right)^{|{\bm \alpha}|+1}},
\end{equation}
where $C_{f,{\bm \alpha}}(\xi)$ is a non-diverging parameter, in the sense that it takes a finite value when $\xi=\bar{\xi}$. 
In particular
\begin{equation}\label{PP1}
|f^{-1}(\omega,{\bm n})|\leq \frac{e^{-\xi|{\bm n}|}}{\|f^{-1}\|^{-1}-c(\xi)},
\end{equation}
and,
\begin{equation}\label{PP2}
|\partial_j f^{-1}(\omega,{\bm n})|\leq \frac{ \sup_{|\mathrm{Im}({\bm \xi})|=\xi}\|e^{ -\mathrm{Im}({\bm \xi}\cdot{\bm x})}\partial_j f\| }{\left (\|f^{-1}\|^{-1}-c(\xi)\right)^2} \ e^{-\xi|{\bm n}|}.
\end{equation} 
\end{theorem}

\noindent {\bf Remark.} The parameter $\sup_{|\mathrm{Im}({\bm \xi})|=\xi}\|e^{ -\mathrm{Im}({\bm \xi}\cdot{\bm x})}\partial_j f\|$ appearing in the last equation can be easily evaluated once $f$ is explicitly given. For higher derivatives, explicit expressions for the coefficients $C_{f,{\bm \alpha}}(\xi)$ are also available but they are relatively complicated.

\begin{corollary}\label{EX1} Let $h$ be the nearest-neighbor hopping Hamiltonian of Eq.~\ref{MainModel}. Then, for any multi-index ${\bm \alpha}$, there exists a finite, non-diverging parameter $C_{h,{\bm \alpha}}(\xi)$, independent of $z$, such that:
\begin{equation}
|\partial_{\bm \alpha} (h-z{\bm 1})^{-1}(\omega,{\bm n})|\leq \frac{C_{h,{\bm \alpha}}(\xi) e^{-\xi |{\bm n}|}}{ (\mathrm{dist}(z,\sigma(h))-2d\sinh \xi )^{|{\bm \alpha}|+1}},
\end{equation}
whenever 
\begin{equation}\label{XI}
\xi < \bar{\xi}=\sinh^{-1} \left (\frac{\mathrm{dist}(z,\sigma(h))}{2d}\right ).
\end{equation}
In particular:
\begin{equation}
|(h-z{\bm 1})^{-1}(\omega,{\bm n})|\leq \frac{e^{-\xi |{\bm n}|}}{\mathrm{dist}(z,\sigma(h))-2d \sinh \xi},
\end{equation}
and
\begin{equation}
|\partial_j (h-z{\bm 1})^{-1}(\omega,{\bm n})|\leq \frac{2d  \cosh (\xi) \ e^{-\xi |{\bm n}|}}{ \left (\mathrm{dist}(z,\sigma(h))-2d \sinh \xi \right)^2}.
\end{equation}
\end{corollary}

\begin{corollary}\label{EX2} Let $h$ be the nearest-neighbor Hamiltonian of Eq.~\ref{MainModel} and let $\Phi(z)$ be an analytic function in a complex neighborhood of $\sigma(h)$ defined by $\mathrm{dist}(z,\sigma(h))\leq \kappa$. Then there exists a finite parameter $\bar{\Phi}$, depending entirely on $\Phi$, such that:
\begin{equation}
|\partial_{\bm \alpha} \Phi(h)(\omega,{\bm n})|\leq \frac{C_{h,{\bm \alpha}}(\xi)\bar{\Phi} \ e^{-\xi |{\bm n}|}}{ \left (\kappa- 2d \sinh \xi \right)^{|{\bm \alpha}|+1}},
\end{equation}
whenever
\begin{equation}
\xi < \bar{\xi}=\sinh^{-1} \left (\frac{\kappa}{2d}\right ).
\end{equation}
\end{corollary}

\begin{corollary}\label{EX3} Let $h$ and $\Phi(z)$ be like in the previous corollary. Then the operator representation $\pi_\omega \Phi(h)$ and its derivatives have exponentially localized kernels:
\begin{equation}
|\langle {\bm n}| \pi_\omega (\partial_{\bm \alpha} \Phi(h))|{\bm m}\rangle |\leq \frac{C_{h,{\bm \alpha}}(\xi)\bar{\Phi} \ e^{-\xi |{\bm n}-{\bm m}|}}{ \left (\kappa- 2d\sinh \xi \right )^{|{\bm \alpha}|+1}},
\end{equation}
for all $\xi < \bar{\xi}$.
\end{corollary}

\subsection{\underline{Thomas-Combes technique for the noncommutative Brillouin torus}}

The main goal here is to prove Theorem~\ref{TC}, which is the key for exponential localization estimates stated above. Let ${\bm \xi}$ be a $d$-dimensional vector with complex components. We define the following map:
\begin{equation}
e^{i{\bm \xi}\cdot{\bm x}}:{\cal A}_0 \rightarrow {\cal A}_0, \ \ (e^{i{\bm \xi}\cdot{\bm x}}f)(\omega,{\bm n})=e^{i{\bm \xi}\cdot{\bm n}}f(\omega,{\bm n}).
\end{equation}

\begin{proposition}
The map $e^{i{\bm \xi}\cdot{\bm x}}$ is an automorphism on ${\cal A}_0$, which extends to the entire ${\cal A}$ algebra when $\mbox{Im}({\bm \xi})=0$. This automorphism sends ${\bm 1}$ into ${\bm 1}$.
\end{proposition}  

\begin{proof} We have successively:
\begin{equation}
\begin{split}
(e^{i{\bm \xi}\cdot{\bm x}}(f*g))(\omega,{\bm n})&=e^{i{\bm \xi}\cdot{\bm n}}\sum_{\bm m} f(\omega,{\bm m})g(\mathfrak{t}^{-1}_{\bm m}\omega,{\bm n}-{\bm m})e^{i\pi ({\bm n}\cdot {\bm F}\cdot {\bm m})}\medskip \\
&=\sum_{\bm m} e^{i{\bm \xi}\cdot{\bm m}}f(\omega,{\bm m}) e^{i{\bm \xi}\cdot({\bm n}-{\bm m})}g(\mathfrak{t}^{-1}_{\bm m}\omega,{\bm n}-{\bm m})e^{i\pi ({\bm n}\cdot {\bm F}\cdot {\bm m})} \medskip \\
&=\sum_{\bm m} (e^{i{\bm \xi}\cdot{\bm x}}f)(\omega,{\bm m}) (e^{i{\bm \xi}\cdot{\bm x}}g)(\mathfrak{t}^{-1}_{\bm m}\omega,{\bm n}-{\bm m})e^{i\pi ({\bm n}\cdot {\bm F}\cdot {\bm m})} \medskip \\
&=((e^{i{\bm \xi}\cdot{\bm x}}f)*(e^{i{\bm \xi}\cdot{\bm x}}g))(\omega,{\bm n}).
\end{split}
\end{equation}
This confirms that:
\begin{equation}
e^{i{\bm \xi}\cdot{\bm x}}(f*g)=(e^{i{\bm \xi}\cdot{\bm x}}f) *(e^{i{\bm \xi}\cdot{\bm x}}g).
\end{equation}
Clearly
\begin{equation}
(e^{i{\bm \xi}\cdot{\bm x}}{\bm 1})(\omega,{\bm n})=\delta_{{\bm n},{\bm 0}}={\bm 1}(\omega,{\bm n}).
\end{equation}
The map is invertible on ${\cal A}_0$ since $e^{i{\bm \xi}\cdot{\bm x}}\circ e^{-i{\bm \xi}\cdot{\bm x}}=1$. Furthermore, if $\mbox{Im}({\bm \xi}) = 0$, then $\|e^{i{\bm \xi}\cdot{\bm x}}f\|=\|f\|$ for all $f \in {\cal A}_0$, so $e^{i{\bm \xi}\cdot{\bm x}}$ can be extended by continuity over the entire ${\cal A}$, as a norm-preserving map.\qed
\end{proof}

\begin{proposition}\label{PX} Fix $f \in {\cal A}_0$ and consider $e^{i{\bm \xi}\cdot{\bm x}}f$ as a function from ${\mathbb C}^d$ with values in a normed space (the algebra ${\cal A}$). Then $e^{i{\bm \xi} \cdot {\bm x}}f$ is analytic of ${\bm \xi}$ on the entire ${\mathbb C}^d$.
\end{proposition}

\begin{proof} Consider a complex variation $\Delta \xi_j$ of the $j$-th component of ${\bm \xi}$. This variation induces a variation of the vector ${\bm \xi}$ that will be denoted by $\Delta {\bm \xi}_j$. We need to show that the limit
\begin{equation}
\lim_{\Delta {\bm \xi}_j\rightarrow 0} \frac{e^{i({\bm \xi}+\Delta {\bm \xi}_j)\cdot{\bm x}}f-e^{i{\bm \xi}\cdot{\bm x}}f}{\Delta  \xi_j}
\end{equation}
exists and does not depend on how $\Delta \xi_j$ approaches the origin in the complex plane, for all ${\bm \xi}\in {\mathbb C}^d$. This criterion for ``analyticity" is equivalent to all the other existing criteria. We first guess that the limit is equal to $i x_j e^{i{\bm \xi}\cdot{\bm x}}f$ and then we prove that this is indeed the case. Note that $i x_j e^{i{\bm \xi}\cdot{\bm x}}f$ exists as an element of ${\cal A}_0$. We have successively:
\begin{equation}
\begin{split}
& \left [ (\Delta \xi_j)^{-1}(e^{i({\bm \xi}+\Delta {\bm \xi}_j)\cdot{\bm x}}f-e^{i{\bm \xi}\cdot{\bm x}}f)-i x_j e^{i{\bm \xi}\cdot{\bm x}}f \right ](\omega,{\bm n}) \medskip \\
& =\left [ (\Delta \xi_j)^{-1}(e^{in_j\Delta \xi_j }-1)-i n_j \right ] e^{i{\bm \xi}\cdot{\bm n}}f (\omega,{\bm n}) \medskip \\
&= O(\Delta \xi_j) e^{i{\bm \xi}\cdot{\bm n}}f (\omega,{\bm n}). 
\end{split}
\end{equation} 
Since $f$ has a compact support, $O(\Delta \xi_j)$ can be bounded by an $\bar O(\Delta \xi_j)$ that is independent of ${\bm n}$. Then
\begin{equation}
\|(\Delta \xi_j)^{-1}(e^{i({\bm \xi}+\Delta {\bm \xi}_j)\cdot{\bm x}}f-e^{i{\bm \xi}\cdot {\bm x}}f)-i x_j e^{i{\bm \xi}\cdot{\bm x}}f \| \leq \bar O(\Delta \xi_j) \|e^{i{\bm \xi}\cdot{\bm n}}f \| \rightarrow 0,
\end{equation}
whenever $|\Delta \xi_j|$ is sent to zero.\qed
\end{proof}
 
 \begin{proposition}\label{X} If $f \in {\cal A}_0$ is an invertible element in ${\cal A}$, then, for any real ${\bm \xi}$, so is $e^{i{\bm \xi}\cdot{\bm x}}f$. Furthermore, $(e^{i{\bm \xi}\cdot{\bm x}}f)^{-1}=e^{i{\bm \xi}\cdot{\bm x}}f^{-1}$.
 \end{proposition}

\begin{proof} Even though $f\in {\cal A}_0$, $f^{-1}$ will not be in ${\cal A}_0$. However, since $\mbox{Im}({\bm \xi})= 0$, the map $e^{i{\bm \xi}\cdot{\bm x}}$ extendeds to the entire algebra ${\cal A}$, where it continues to be an automorphism. We are then allowed to consider the element $e^{i{\bm \xi}\cdot{\bm x}}f^{-1}$. Note that with the knowledge we have so far about this map, we will not be able to make sense of $e^{i{\bm \xi}\cdot{\bm x}}f^{-1}$ if $\mbox{Im}({\bm \xi}) \neq 0$. Finally, we have:
\begin{equation}
(e^{i{\bm \xi}\cdot{\bm x}}f)*(e^{i{\bm \xi}\cdot{\bm x}}f^{-1})=e^{i{\bm \xi}\cdot{\bm x}}(f*f^{-1})=e^{i{\bm \xi}\cdot{\bm x}}{\bm 1}={\bm 1}.\qed
\end{equation}
\end{proof}

 \begin{proposition}\label{Y} If $f \in {\cal A}_0$ is an invertible element in ${\cal A}$, then there exists $\bar{\xi}>0$ (to be given during the proof) such that $e^{i{\bm \xi}\cdot{\bm x}}f$ is also invertible as long as $|\mathrm{Im}({\bm \xi})|<\bar{\xi}$.
 \end{proposition} 
 \begin{proof}
 Since
 \begin{equation}
 e^{i{\bm \xi}\cdot{\bm x}}f = e^{i\mathrm{Re}( {\bm \xi}\cdot{\bm x})}(e^{-\mathrm{Im}( {\bm \xi}\cdot{\bm x})}f),
 \end{equation}
 and given the previous result, $e^{i{\bm \xi}\cdot{\bm x}}f $ is invertible if $e^{-\mathrm{Im}( {\bm \xi}\cdot{\bm x})}f$ is invertible. We have:
 \begin{equation}
 e^{-\mathrm{Im}( {\bm \xi}\cdot{\bm x})}f=(1-(f-e^{-\mathrm{Im}( {\bm \xi}\cdot{\bm x})}f)*f^{-1})*f,
 \end{equation}
 and, using Eq.~\ref{NormA},
 \begin{equation}\label{V}
 \|(f-e^{-\mathrm{Im} ({\bm \xi}\cdot{\bm x})}f)*f^{-1}\|  \leq   \|f^{-1}\|\sup\limits_{|\mathrm{Im}({\bm \xi})|=\xi} \sum_{{\bm n}\in \mathbb{Z}^d}  \ \left |\left (1-e^{-\mathrm{Im}({\bm \xi} \cdot {\bm n})}\right )\right | \sup\limits_{\omega\in \Omega}|f(\omega,{\bm n}) |.
 \end{equation}
Since the support of $f(\omega,\cdot)$ is contained in a bounded subset of $\mathbb{Z}^d$ for all $\omega$'s, the last line can be clearly seen to go to zero when $|\mathrm{Im}( {\bm \xi})|\rightarrow 0$. As such, there is a $\bar{\xi}>0$ such that 
 \begin{equation}
\|(f-e^{-\mathrm{Im}( {\bm \xi}\cdot{\bm x})}f)*f^{-1}\| < 1,
\end{equation}
for all ${\bm \xi}$'s with  $|\mathrm{Im} ({\bm \xi})| < \bar{\xi}$, in which case the sequence
\begin{equation}
\sum_{j=0}^\infty \left ((f-e^{-\mathrm{Im} ({\bm \xi}\cdot{\bm x})}f)*f^{-1}\right )^{*j}
\end{equation}
converges in the alegebra ${\cal A}$ to an element that is precisely 
\begin{equation}
\left (1-(f-e^{-\mathrm{Im}( {\bm \xi}\cdot{\bm x})}f)*f^{-1}\right)^{-1},
\end{equation}
 enabling us to conclude that $e^{i{\bm \xi}\cdot{\bm x}}f$ is invertible in these conditions. Furthermore:
\begin{equation}
\|(e^{i{\bm \xi}\cdot{\bm x}}f)^{-1}\| \leq \frac{1}{\|f^{-1}\|^{-1}- c(\xi)}
\end{equation}
where 
\begin{equation}
c(\xi)=\sup\limits_{|\mathrm{Im}({\bm \xi})|=\xi} \sum_{{\bm n}\in \mathbb{Z}^d}  \ |(1-e^{-\mathrm{Im}({\bm \xi} \cdot {\bm n})})| \sup\limits_{\omega\in \Omega}|f(\omega,{\bm n}) |, 
\end{equation}
is the constant appearing in Eq.~\ref{V}. The upper limit $\bar{\xi}$ can be identified as the unique solution to the equation: $c(\bar{\xi})=\|f^{-1}\|^{-1}$.\qed
 \end{proof}
 
 \begin{proposition}\label{TT} Let us fix an invertible (in ${\cal A}$) element $f\in {\cal A}_0$. Then the map:
 \begin{equation}
 {\bm \xi} \rightarrow \left (e^{i{\bm \xi}\cdot{\bm x}}f \right )^{-1}
 \end{equation}
 is analytic for all ${\bm \xi}\in {\mathbb C}^d$ with $|\mathrm{Im} {\bm \xi}| < \bar{\xi}$ (with $\bar{\xi}$ defined in the previous Proposition).
 \end{proposition}
 
 \begin{proof} Considering again a variation $\Delta \xi_j$ of the $j$-th component of ${\bm \xi}$, we need to show that the limit
\begin{equation}\label{UU}
\lim_{\Delta \xi_j \rightarrow 0} \frac{(e^{i({\bm \xi}+\Delta {\bm \xi}_j)\cdot{\bm x}}f)^{-1}-(e^{i{\bm \xi}\cdot{\bm x}}f)^{-1}}{\Delta  \xi_j}
\end{equation}
exists and does not depend on how $\Delta \xi_j$ approaches the origin, for all ${\bm \xi}$ with $|\mathrm{Im} {\bm \xi}| < \bar{\xi}$. We first guess that this limit is $-(e^{i{\bm \xi}\cdot{\bm x}}f)^{-1}(ix_je^{i{\bm \xi}\cdot{\bm x}}f)(e^{i{\bm \xi}\cdot{\bm x}}f)^{-1}$ and then we prove that this is indeed the case. The inverse $(e^{i{\bm \xi}\cdot{\bm x}}f)^{-1}$ exists for all ${\bm \xi}$ with $|\mathrm{Im} {\bm \xi}| < \bar{\xi}$. We have:
\begin{equation}
\begin{split}
& \|(\Delta  \xi_j)^{-1}[(e^{i({\bm \xi}+\Delta {\bm \xi}_j)\cdot{\bm x}}f)^{-1}-(e^{i{\bm \xi}\cdot{\bm x}}f)^{-1}]+(e^{i{\bm \xi}\cdot{\bm x}}f)^{-1}(ix_je^{i{\bm \xi}\cdot{\bm x}}f)(e^{i{\bm \xi}\cdot{\bm x}}f)^{-1}\| \medskip \\
& \ \ \ =\|(e^{i({\bm \xi}+\Delta {\bm \xi}_j)\cdot{\bm x}}f)^{-1}e^{i{\bm \xi}\cdot{\bm x}}[(\Delta  \xi_j)^{-1}(f-e^{i(\Delta {\bm \xi}_j)\cdot{\bm x}}f)+ix_jf](e^{i{\bm \xi}\cdot{\bm x}}f)^{-1}\medskip \\
& \ \ \ -[(e^{i({\bm \xi}+\Delta {\bm \xi}_j)\cdot{\bm x}}f)^{-1}-(e^{i{\bm \xi}\cdot{\bm x}}f)^{-1}](ix_je^{i{\bm \xi}\cdot{\bm x}}f)(e^{i{\bm \xi}\cdot{\bm x}}f)^{-1}\|=O(|\Delta \xi_j|),
\end{split}
\end{equation}
as we've seen in our previous estimates in Proposition~\ref{PX}. Thus, the limit in Eq.~\ref{UU} exists and is equal to 
\begin{equation}
-(e^{i{\bm \xi}\cdot{\bm x}}f)^{-1}(ix_je^{i{\bm \xi}\cdot{\bm x}}f)(e^{i{\bm \xi}\cdot{\bm x}}f)^{-1}.\qed
\end{equation}
 \end{proof} 

\begin{theorem}[Thomas-Combes \cite{Combes:1973nx} for algebra ${\cal A}$]\label{TC}
Let $f\in {\cal A}_0$ be an invertible element in ${\cal A}$. Then there exists a constant $\bar{\xi}$, determined entirely by $f$, such that, for all ${\bm \xi}$'s with $|\mathrm{Im}({\bm \xi})|<\bar{\xi}$, $e^{i{\bm \xi}\cdot{\bm x}}f$ is invertible and we have the equality:
\begin{equation}
e^{i{\bm \xi}\cdot{\bm x}}f^{-1} = (e^{i{\bm \xi}\cdot{\bm x}}f)^{-1}.
\end{equation}
\end{theorem}

\begin{proof} We have already learned that $e^{i{\bm \xi}\cdot{\bm x}}$ is an analytic automorphism on ${\cal A}_0$. Unfortunately $f^{-1}$ is generally not in ${\cal A}_0$. So we will force $f^{-1}$ to be in ${\cal A}_0$ by applying the map $\chi_R:{\cal A}\rightarrow {\cal A}_0$ defined as $(\chi_R f)(\omega,{\bm n})=f(\omega,{\bm n})$ if $|{\bm n}|<R$ and zero otherwise. Then:
\begin{equation}
{\bm \xi} \rightarrow e^{i{\bm \xi}\cdot{\bm x}}(\chi_R f^{-1}),
\end{equation}
is an analytic map on the entire complex plane. Note that $\chi_R$ commutes with $e^{i{\bm \xi}{\bm x}}$: $e^{i{\bm \xi}\cdot{\bm x}} \circ \chi_R=\chi_R \circ e^{i{\bm \xi}\cdot{\bm x}}$ on ${\cal A}_0$, and this property can be extended to the whole ${\cal A}$ by continuity. Note that we can always make sense of $e^{i{\bm \xi}\cdot{\bm x}}f^{-1}$ as a complex valued function defined on $\Omega \times \mathbb{Z}^d$. As such, we can also say that:
\begin{equation}
{\bm \xi} \rightarrow \chi_R (e^{i{\bm \xi}\cdot{\bm x}}f^{-1}),
\end{equation}
is an analytic map on the entire complex plane. Now, we have already seen in Propostion~\ref{TT} that:
\begin{equation}
 {\bm \xi} \rightarrow (e^{i{\bm \xi}\cdot{\bm x}}f )^{-1}
 \end{equation}
 is analytic for all ${\bm \xi}\in {\mathbb C}^d$ with $|\mathrm{Im}({\bm \xi})| < \bar{\xi}$. Obviously the map
\begin{equation}
 {\bm \xi} \rightarrow \chi_R(e^{i{\bm \xi}\cdot{\bm x}}f )^{-1}
 \end{equation} 
shares the same property. So at this point, we have two maps, 
\begin{equation}
\chi_R (e^{i{\bm \xi}\cdot{\bm x}}f^{-1}) \ \ \mbox{and} \ \ \chi_R(e^{i{\bm \xi}\cdot{\bm x}}f )^{-1},
\end{equation}
which are analytic on the cylinder $|\mathrm{Im}({\bm \xi})| < \bar{\xi}$ and, according to Proposition~\ref{X}, they coincide on the real axis. Therefore the two maps are identical on the entire cylinder $|\mathrm{Im} {\bm \xi}| < \bar{\xi}$. Note that this statement is independent of how large $R$ is. We can take the limit $R\rightarrow \infty$ for $\chi_R(e^{i{\bm \xi}{\bm x}}f )^{-1}$ because we already know (see Proposition~\ref{Y}) that $(e^{i{\bm \xi}{\bm x}}f )^{-1}$ exists as an element of ${\cal A}$. In the light of the equality we just established, it follows that $\chi_R (e^{i{\bm \xi}{\bm x}}f^{-1})$ also has a limit in ${\cal A}$ as $R$ is taken to infinity, and this limit must be equal to $(e^{i{\bm \xi}{\bm x}}f )^{-1}$.\qed  
\end{proof}

\subsection{\underline{Proof of Theorem~\ref{TH1} and its corollaries}}

From the upper bounds on $(e^{i{\bm \xi}{\bm x}}f)^{-1}$ established in Proposition~\ref{Y}, we have
\begin{equation}
\left ( \|f^{-1}\|^{-1}-c(\xi)\right )^{-1} \geq \|(e^{i{\bm \xi}\cdot{\bm x}}f)^{-1}\|=\|e^{i{\bm \xi}\cdot{\bm x}}f^{-1}\| \geq |e^{i{\bm \xi}\cdot{\bm n}}f^{-1}(\omega,{\bm n})|,
\end{equation}
or:
\begin{equation}\label{XX}
|f^{-1}(\omega,{\bm n})|\leq \frac{e^{\mathrm{Im}({\bm \xi}\cdot {\bm n})}}{\|f^{-1}\|^{-1}-c(\xi)},
\end{equation}
independently of the direction of vector ${\bm \xi}$. Above we used the simple fact that $|f(\omega,{\bm n})|$ is always bounded by $\|f\|$. Then the upper bound of Eq.~\ref{PP1} follows by taking $\mathrm{Im}({\bm \xi})$ as $-\xi\frac{{\bm n}}{|{\bm n}|}$, with $\xi<\bar{\xi}$. 

To demonstrate the upper bound of Eq.~\ref{PP2}, we make use of the identity from Eq.~\ref{DiffId1}:
\begin{equation}\label{sh}
\partial_j f^{-1} = - f^{-1}*(\partial_j f)*f^{-1}.
\end{equation}
Since $e^{i{\bm \xi}\cdot {\bm x}}$ is a morphism, we can write:
\begin{equation}
e^{i{\bm \xi}\cdot {\bm x}} (\partial_j f^{-1}) = - (e^{i{\bm \xi}\cdot {\bm x}}f^{-1})*(e^{i{\bm \xi}\cdot {\bm x}}\partial_j f)*(e^{i{\bm \xi}\cdot {\bm x}}f^{-1}),
\end{equation}
to conclude:
\begin{equation}
|e^{i{\bm \xi}\cdot {\bm n}} (\partial_j f^{-1})(\omega,{\bm n})|\leq \|e^{i{\bm \xi}\cdot {\bm x}} (\partial_j f^{-1})\| \leq \|e^{i{\bm \xi}\cdot {\bm x}}\partial_j f\| \ \|e^{i{\bm \xi}\cdot {\bm x}}f^{-1}\|^2.
\end{equation}
Using again the upper bound established in Proposition~\ref{Y}, we can conclude:
\begin{equation}
| (\partial_j f^{-1})(\omega,{\bm n})|\leq \frac{e^{\mathrm{Im}({\bm \xi}\cdot {\bm n})}\sup_{|\mathrm{Im}({\bm \xi})|=\xi}\|e^{ -\mathrm{Im}({\bm \xi}\cdot{\bm x})}\partial_j f\| }{\left( \|f^{-1}\|^{-1}-c(\xi) \right )^2}
\end{equation}
and Eq.~\ref{PP2} follows by taking $\mathrm{Im}({\bm \xi})$ as $-\xi\frac{{\bm n}}{|{\bm n}|}$, with $\xi<\bar{\xi}$.

Estimates on the higher order derivatives can be obtained in a similar way. For example, taking another derivative on Eq.~\ref{sh},
\begin{equation}
\begin{split}
\partial_k\partial_j f^{-1} &=  f^{-1}*(\partial_k f)*f^{-1}*(\partial_j f)*f^{-1} \medskip \\
&\ \ \ - f^{-1}*(\partial_k\partial_j f)*f^{-1} +  f^{-1}*(\partial_j f)* f^{-1}*(\partial_k f)*f^{-1},
\end{split}
\end{equation}
and following the same steps as before, we obtain:
\begin{equation}
|e^{i{\bm \xi} \cdot {\bm n}}(\partial_k\partial_j f^{-1})(\omega,{\bm n})| 
\leq \frac{2\|e^{-\mathrm{Im}({\bm \xi}\cdot{\bm x})}(\partial_k f)\| \ \|e^{-\mathrm{Im}({\bm \xi}\cdot{\bm x})}(\partial_j f)\|}{\left( \|f^{-1}\|^{-1}-c(\xi) \right )^3} +   \frac{\|e^{-\mathrm{Im}({\bm \xi}\cdot{\bm x})}(\partial_k \partial_j f)\| }{\left( \|f^{-1}\|^{-1}-c(\xi) \right )^2}.
\end{equation}
This is the same as:
\begin{equation}
|(\partial_k\partial_j f^{-1})(\omega,{\bm n})| \leq \frac{C_{kj}(\xi) \ e^{\mathrm{Im}({\bm \xi}\cdot {\bm n})}}{\left( \|f^{-1}\|^{-1}-c(\xi) \right )^3},
\end{equation}
with
\begin{equation}
\begin{split}
C_{kj}(\xi) & =\sup\limits_{|\mathrm{Im}({\bm \xi})|=\xi}\left ( 2\|e^{-\mathrm{Im}({\bm \xi}\cdot{\bm x})}(\partial_k f)\| \ \|e^{-\mathrm{Im}({\bm \xi}\cdot{\bm x})}(\partial_j f)\| \right .\medskip \\
& \ \ \ \left . +\|e^{-\mathrm{Im}({\bm \xi}\cdot{\bm x})}(\partial_k \partial_j f)\|  ( \|f^{-1}\|^{-1}-c(\xi) ) \right ).
\end{split}
\end{equation}
At this point we can take $\mathrm{Im}({\bm \xi})$ as $-\xi\frac{{\bm n}}{|{\bm n}|}$, with $\xi<\bar{\xi}$ and the desired estimate follows. The above parameter can be easily computed once the element $f$ is explicitly given.  Higher order derivatives can be treated in the same way.

Corollary~\ref{EX1} is just a particular case of the main statement. For the case when $f=h-z{\bm 1}$, with $h$ the nearest-neighbor (self-adjoint) Hamiltonian of Eq.~\ref{MainModel}, we can compute the parameters explicitly. Indeed $c(\xi)$ can be computed from its very definition, to be:
\begin{equation}
c(\xi)=2 d \sinh {\xi}.
\end{equation}
It is important to remark that $c(\xi)$ is entirely determine by $h$ (so it is independent of $z$). Also $\|(h-z{\bm 1})^{-1}\|=\mathrm{dist}(z,\sigma(h))$, a well known fact in the spectral theory. It is also straightforward to show that 
\begin{equation}
\sup_{|\mathrm{Im}({\bm \xi})|=\xi}\|e^{\mathrm{Im}({\bm \xi}\cdot{\bm x})}\partial_j (h-z{\bm 1})\|\leq 2d\cosh \xi,
\end{equation}
 which is again entirely determined by $h$. Corollary~\ref{EX1} then follows.

For Corollary~\ref{EX2}, we have:
\begin{equation}
\Phi(h)=\frac{i}{2\pi}\int_{{\cal C}_\kappa}(h-z{\bm 1})^{-1} \Phi(z) dz,
\end{equation}
with ${\cal C}_\kappa$ the contour defined by $\mathrm{dist}(z,\sigma(h))=\kappa$. Using the estimates already established for the resolvent, we can automatically write:
\begin{equation}
\left \| e^{i{\bm \xi}{\bm x}}\partial_{\bm \alpha}\Phi(h) \right \| \leq  \frac{C_{h,{\bm \alpha}}(\xi)}{\left (\kappa- 2d\sinh \xi \right )^{|{\bm \alpha}|+1}} \frac{1}{2\pi}\int_{{\cal C}_\kappa} |\Phi(z)||dz|,
\end{equation}
and the statement follows, with:
\begin{equation}\label{PhiBar}
\bar{\Phi} = \frac{1}{2\pi}\int_{{\cal C}_\kappa} |\Phi(z)||dz|.
\end{equation}
Throughout the following sections, we will consistently use the notation from Eq.~\ref{PhiBar}. 

For Collorary~\ref{EX3}, it is enough to observe that:
\begin{equation}
|\langle {\bm n}|\pi_\omega(\partial_{\bm \alpha} \Phi(h))|{\bm m}\rangle | = \partial_{\bm \alpha} \Phi(h)(\mathfrak{t}_{\bm n}^{-1}\omega,{\bm m}-{\bm n}).
\end{equation}
Then the statement follow from the previous localization estimates.

\section{Approximating algebras: First Round}

In this section we construct the algebra of observables for the more computationally manageable situations when the on-site potential is random only inside a large super-cell $\Lambda$ and this super-cell is periodically repeated until the whole $\mathbb{Z}^d$ is filled. We call this algebra the periodic algebra. Note that no condition is imposed on the magnetic field, so the whole system is not periodic. While the periodic algebra is not the end of our string of approximation, we want to point out that, especially in 2-dimensions, explicit analytic and computer assisted analysis is still possible at the level of this algebra. This is because one can chose a special (asymmetric-) gauge to represent the matrix ${\bm F}$, so that the phases appearing in the multiplication law affect only one spatial direction. In this case the system becomes periodic in the remaining spatial direction and one can use the Bloch decomposition with respect to that direction. Hence, the 2-dimensional Schr\"odinger equation can be reduced to a 1-dimensional Harper's like equation, which has been studied computationally by various techniques \cite{CzychollSolStComm1988ds,CzychollZPhysB1988jd,MandalPhysB1998as,TanJPhysCondMatt1994rt,RochePRB1999te,SteffenPRB2004bn,ShengPRL1997cy,ShengPRB2001se,KoshinoPRB2006gh,MaitiPhysLettA2012gh,Dutta2012cv}. As such, the rigorous error bounds reported in this section (see Theorem~\ref{ThR1}) can still be of interest to numerical analysts.  

\subsection{\underline{The periodic $C^*$-algebra}}

Let $\mathfrak{t}_{\bm e_j}$ be the translation by one unit in the $j$-th direction and let ${\cal N}$ be a positive integer. We define the subset $\Omega_{\mathrm{per}}^{\cal N} \subset \Omega$ of periodic $\omega$'s as: 
\begin{equation}
\begin{array}{c}
\Omega_{\mathrm{per}}^{\cal N}=\{\omega \in \Omega \ | \ \mathfrak{t}_{(2{\cal N}+1){\bm e}_j}\omega = \omega, \ j=1,\ldots,d\}.
\end{array}
\end{equation}
The set $\Omega_{\mathrm{per}}^{\cal N}$ is closed (hence compact) and translationally invariant, in the sense that $\mathfrak{t}_{\bm n}\omega \in \Omega_{\mathrm{per}}^{\cal N}$, for all $\omega \in \Omega_{\mathrm{per}}^{\cal N}$ and ${\bm n}\in \mathbb{Z}^d$. An $\omega$ from $\Omega_{\mathrm{per}}^{\cal N}$ can be constructed by periodically repeating its components $\omega_{\bm n}$ for ${\bm n}$ in the unit cell defined as:
\begin{equation}
\Lambda_{\cal N}=\{-{\cal N},\ldots,{\cal N}\}^{d}.
\end{equation} 

We define the periodic $C^*$-algebra $C^*(\Omega_{\mathrm{per}}^{\cal N} \times \mathbb{Z}^d,{\bm F})$ in the following way. The elements are continuous functions $f(\omega,{\bm n})$ from $\Omega_{\mathrm{per}}^{\cal N} \times \mathbb{Z}^d$ to $\mathbb{C}$. Since $\Omega_{\mathrm{per}}^{\cal N}$ is invariant to translations, the multiplication law in Eq.~\ref{AlgRules}, the $*$-operation in Eq.~\ref{Star} and the operator representations in Eq.~\ref{OpRep}, defined for algebra ${\cal A}$, make sense when we restrict the $f$'s to $\Omega_{\mathrm{per}}^{\cal N}$. As such, we will keep $*_{\bm F}$, the operation in Eq.~\ref{Star} and the representation from Eq.~\ref{OpRep} as the multiplication, $*$-operation and operator representation  for the $C^*(\Omega_{\mathrm{per}}^{\cal N} \times \mathbb{Z}^d,{\bm F})$, respectively. The norm is defined by:
\begin{equation}
\|f\|=\sup_{\omega \in \Omega^{\cal N}_{\mathrm{per}}}\|\pi_\omega f\|.
\end{equation}
Obviously, $\|f^**f\|=\|f\|^2$, so $C^*(\Omega_{\mathrm{per}}^{\cal N} \times \mathbb{Z}^d,{\bm F})$ is a $C^*$-algebra.

The differential calculus can be defined as before because the automorphisms $\partial_j$ remain well defined when the functions are restricted to $\Omega_{\mathrm{per}}^{\cal N}$. However, the the trace defined in Eq.~\ref{Trace} for algebra $C^*(\Omega \times \mathbb{Z}^d,{\bm F})$ is of no use here because the set $\Omega_{\mathrm{per}}^{\cal N}$ has zero measure relative to $dP(\omega)$. We define a new probability measure (only on $\Omega_{\mathrm{per}}^{\cal N}$), through the following natural formula:
\begin{equation}\label{ML}
dP_{\mathrm{per}}(\omega)=\prod_{{\bm n}\in \Lambda_{\cal N}} d \omega_{\bm n}.
\end{equation}
Clearly this measure is invariant to the translations $\mathfrak{t}$ because the $\omega$'s are periodic, a property that allows us to define a trace over $C^*(\Omega_{\mathrm{per}}^{\cal N} \times \mathbb{Z}^d,{\bm F})$.

\begin{proposition} The following linear functional over $C^*(\Omega_{\mathrm{per}}^{\cal N} \times \mathbb{Z}^d,{\bm F})$:
\begin{equation}\label{TraceP}
{\cal T}_{\mathrm{per}}(f)= \int_{\Omega^{\cal N}_{\mathrm{per}}} dP_{\mathrm{per}}(\omega) \ f(\omega,0)
\end{equation}
has the following properties:
\begin{equation}
{\cal T}_{\mathrm{per}}(f * g)={\cal T}_{\mathrm{per}}(g * f), {\cal T}_{\mathrm{per}}({\bm 1})=1, \ \mathrm{and} \ {\cal T}_{\mathrm{per}}(f * f^*)>0 \ \mathrm{for} \ f \neq 0.
\end{equation}
As such, ${\cal T}_{\mathrm{per}}$ defines a trace over $C^*(\Omega_{\mathrm{per}}^{\cal N} \times \mathbb{Z}^d,{\bm F})$.
\end{proposition}
\begin{proof} The statement follows via a direct computation, using the fact that $dP_{\mathrm{per}}(\omega)$ is translational invariant.\qed
\end{proof}
\begin{proposition}\label{QQ} The following is an alternative, useful way to compute the trace:
\begin{equation}
{\cal T}_{\mathrm{per}}(f)=\frac{1}{|\Lambda_{\cal N}|} \int_{\Omega^{\cal N}_{\mathrm{per}}} dP_{\mathrm{per}}(\omega) \ \mathrm{Tr}\{(\pi_\omega f)\chi_{\Lambda_{\cal N}}\},
\end{equation}
where $\chi_{\Lambda_{\cal N}}$ is the characteristic function of $\Lambda_{\cal N}$.
\end{proposition}
\begin{proof} The statement follows from the definitions and from the fact that $dP_{\mathrm{per}}(\omega)$ is translationally invariant.\qed
\end{proof}

\subsection{\underline{Analytic Functional Calculus: Comparative estimates}}

Let $f_j \in C^*(\Omega \times \mathbb{Z}^d,{\bm F})$, $j=1,\ldots,N$, be defined via the analytic calculus with the Hamiltonian:
\begin{equation}\label{fis}
f_j=\partial_{{\bm \alpha}_j}\Phi_j(h),
\end{equation}
with $\Phi_j$'s analytic functions in a closed neighborhood of $\sigma(h)$ defined by $\mathrm{dist}(z,\sigma(h))\leq \kappa$. Generally, we are interested in computing expected values of the form:
\begin{equation}
{\cal T}\left (\prod_{j=1}^N f_j\right),
\end{equation}
because most, if not all, the linear physical response coefficients of a system to external fields can be written in this way at finite temperatures. 

Now, any $f$ from algebra ${\cal A}$ can be associated with an element from algebra $C^*(\Omega_{\mathrm{per}}^{\cal N} \times \mathbb{Z}^d,{\bm B})$ by just restricting its domain of definition from $\Omega \times \mathbb{Z}^d$ to $\Omega_{\mathrm{per}}^{\cal N} \times \mathbb{Z}^d$. Let us denote this map by $\mathfrak{P}_{\Lambda_{\cal N}}$. Since the restriction effectively occurs in the $\Omega$ space, the analytic calculus enjoys the following special property:
\begin{equation}\label{SP}
\Phi(\mathfrak{P}_{\Lambda_{\cal N}} f)=\mathfrak{P}_{\Lambda_{\cal N}} \Phi(f).
\end{equation}
As already mentioned in the short introduction for the present section, the correlation functions in the periodic algebra may be evaluated on a computer. Then a natural question arises, namely, what errors should one expect if the elements from $C^*(\Omega \times \mathbb{Z}^d,{\bm F})$ are mapped into elements of $C^*(\Omega_{\mathrm{per}}^{\cal N} \times \mathbb{Z}^d,{\bm F})$ via $\mathfrak{P}_{\Lambda_{\cal N}}$, and the analytic calculus and the expected evaluation are performed inside the algebra $C^*(\Omega_{\mathrm{per}}^{\cal N} \times \mathbb{Z}^d,{\bm F})$ instead of $C^*(\Omega \times \mathbb{Z}^d,{\bm F})$? This is the main question for the present section and the answer to this question is given below.

\begin{theorem}\label{ThR1} Let $f_j$'s be defined as above and assume $h$ is the nearest-neighbor Hamiltonian of Eq.~\ref{MainModel}. Then, for any $0<\xi<\sinh^{-1}(\kappa/2d)$, there exists a finite, completely identifiable (at least in the asymptotic limit of large ${\cal N}$'s) factor $\mathfrak{B}_1(\xi,\{{\bm \alpha}\})$ that is independent of $\Phi_j$'s, such that:
\begin{equation}
\left |{\cal T}\left(\prod_{j=1}^N f_j \right )-{\cal T}_{\mathrm{per}}\left (\prod_{j=1}^N \mathfrak{P}_{\Lambda_{\cal N}} f_j \right) \right | 
\leq  \mathfrak{B}_1(\xi,\{{\bm \alpha}\}) \left (\prod_{j=1}^N \bar{\Phi}_j \right ) {\cal N}^{-1} e^{-\frac{2\sqrt{2}}{3}\xi{\cal N}}.
\end{equation}
\end{theorem}

\begin{proof} Since all $\Phi_j(h)$'s can be generates as 
\begin{equation}
\Phi_j(h)=\frac{i}{2\pi }\int_{{\cal C}_\kappa}\Phi_j(z) (h-z{\bm 1})^{-1}dz,
\end{equation} 
where ${\cal C}_\kappa$ is defined by $\mathrm{dist}(z,\sigma(h))=\kappa$, we can focus first on $f$'s of the form:
\begin{equation}\label{Fs}
f_j=\partial_{{\bm \alpha}_j}(h-z_j{\bm 1})^{-1},
\end{equation}
with $z_i \in {\cal C}_\kappa$. In the following, we will use $h-z$ instead of $h-z{\bm 1}$ in order to simplify the notation.

We partition the box $\Lambda_{\cal N}$ in $3^d$ smaller boxes (assuming $(2{\cal N}+1)\mathrm{mod} \ 3=0$), in which case, at the center of $\Lambda_{\cal N}$, there is a smaller box $\tilde{\Lambda}_{\cal N}$ surrounded from all sides by the remaining $3^d-1$ boxes of the same size. For any $\omega \in \Omega$, let $\mathfrak{p}_{\Lambda_{\cal N}}$ be the projector whose action is to take the restriction of $\omega$ to $\Lambda_{\cal N}$ and to repeat it periodically until a periodic $\omega$ is defined over the entire ${\mathbb Z}^d$. Clearly $\mathfrak{p}_{\Lambda_{\cal N}} \circ \mathfrak{p}_{\Lambda_{\cal N}} =\mathfrak{p}_{\Lambda_{\cal N}}$, so $\mathfrak{p}_{\Lambda_{\cal N}}$ is indeed a projector over $\Omega$

Our starting point is the following difference between the expected values computed on the inner box $\tilde{\Lambda}_{\cal N}$ and for the operator representation:
\begin{equation}\label{Difference}
\frac{1}{|\tilde{\Lambda}_{\cal N}|}\mathrm{Tr}\left \{\pi_\omega \left (\prod_{j=1}^N f_j \right ) \chi_{\tilde{\Lambda}_{\cal N}} \right \}-\frac{1}{|\tilde{\Lambda}_{\cal N}|}\mathrm{Tr}\left \{\pi_{\mathfrak{p}_{\Lambda_{\cal N}}\omega} \left (\prod_{j=1}^N f_j \right ) \chi_{\tilde{\Lambda}_{\cal N}} \right \}. 
\end{equation}
The first task will be to derive a bound on this difference. After simple manipulations, this difference takes the following equivalent expression:
\begin{equation}
\frac{1}{|\tilde{\Lambda}_{\cal N}|}\sum_{k=1}^N \mathrm{Tr}\left \{\pi_\omega \left (\prod_{j < k}  f_j \right ) (\pi_\omega f_k -  \pi_{\mathfrak{p}_{\Lambda_{\cal N}}\omega} f_k)
\pi_{\mathfrak{p}_{\Lambda_{\cal N}}\omega} \left (\prod_{j > k}  f_j \right )\chi_{\tilde{\Lambda}_{\cal N}} \right  \}.
\end{equation}
A basic identity for the calculus with the resolvent gives:
\begin{equation}
\begin{split}
&\partial_{\bm \alpha}(\pi_\omega h-z)^{-1}-\partial_{\bm \alpha}(\pi_{\mathfrak{p}_{\Lambda_{\cal N}}\omega} h-z)^{-1} \medskip \\
& = \partial_{\bm \alpha}\left ((\pi_\omega h-z)^{-1}(\pi_{\mathfrak{p}_{\Lambda_{\cal N}} \omega} h-\pi_\omega h)(\pi_{\mathfrak{p}_{\Lambda_{\cal N}}\omega} h-z)^{-1} \right ).
\end{split}
\end{equation}
A key observation here is that the difference $\pi_{\mathfrak{p}_{\Lambda_{\cal N}} \omega} h-\pi_\omega h$ is equal to the difference between the onsite random potentials, so the previous equation reduces to:
\begin{equation}
\partial_{\bm \alpha}\left ((\pi_\omega h-z)^{-1}(V_{\mathfrak{p}_{\Lambda_{\cal N}}\omega}-V_\omega)(\pi_{\mathfrak{p}_{\Lambda_{\cal N}}\omega} h-z)^{-1} \right )
\end{equation}
Note that the support of $ V_{\mathfrak{p}_{\Lambda_{\cal N}} \omega}-V_\omega$ is contained in the complement $\Lambda_{\cal N}^c=\mathbb{Z}^d-\Lambda_{\cal N}$ of the box $\Lambda_{\cal N}$. Furthermore, since we are considering only onsite disorder potentials, all derivatives of $V_{\mathfrak{p}_{\Lambda_{\cal N}} \omega} - V_\omega$ are identically zero so, by applying the Leibniz rule, the previous equation reduces to:
\begin{equation}
\ldots =\sum_{\bm \beta} \partial_{\bm \beta}(\pi_\omega h-z)^{-1}(V_{\mathfrak{p}_{\Lambda_{\cal N}}\omega}-V_\omega)\partial_{{\bm \beta}'}(\pi_{\mathfrak{p}_{\Lambda_{\cal N}} \omega} h-z)^{-1}.
\end{equation}
Thus the difference written in Eq.~\ref{Difference} is a finite sum of terms of the following generic form:
\begin{equation}\label{Generic}
\frac{1}{|\tilde{\Lambda}_{\cal N}|} \mathrm{Tr}\left \{\pi_\omega\left (\left(\prod_{j < k} f_j \right) f^-_k\right ) (V_{\mathfrak{p}_{\Lambda_{\cal N}} \omega}-V_\omega) \pi_{\mathfrak{p}_\Lambda \omega} \left (f^+_k\prod_{j > k}  f_j \right )\chi_{\tilde{\Lambda}_{\cal N}} \right   \},
\end{equation}
where $k=1,\ldots,N,$ and: 
\begin{equation}
f^\pm_k=\partial_{{\bm \alpha}_k^\pm}(h-z)^{-1},
\end{equation}
with ${\bm \alpha}_k^- + {\bm \alpha}_k^+={\bm \alpha}_k$.

We exploit now the exponential localization of the analytic calculus. For some arbitrary product of $f$'s, we have:
\begin{equation}
\left \|e^{i{\bm \xi}{\bm x}}\prod_{j} f_j \right \|=\left \| \prod_{j} (e^{i{\bm \xi}{\bm x}} f_j) \right \|\leq \prod_{j} \|e^{i{\bm \xi}{\bm x}}f_j\|\leq  \prod_{j}\frac{ C_{h,{\bm \alpha}_j}(\xi)}{\left (\kappa-2d\sinh \xi \right)^{|{\bm \alpha}_j|+1}},
\end{equation}
with $C_{h,{\bm \alpha}_j}(\xi)$ defined in Corollary~\ref{EX1}. Hence we can automatically conclude that, for an arbitrary $\omega \in \Omega$:
\begin{equation}
\begin{split}
&\left [\left(\prod_{j < k} f_j \right )f^-_k \right ](\omega,{\bm n})\leq \frac{C_{h,{\bm \alpha}_k^-}(\xi)}{\left (\kappa-2d \sinh \xi \right)^{|{\bm \alpha}_k^{-}|+1}}\prod\limits_{j < k} \frac{ C_{h,{\bm \alpha}_j}(\xi)}{\left (\kappa-2d\sinh \xi \right)^{|{\bm \alpha}_j|+1}}e^{-\xi |{\bm n}|}, \medskip  \\ 
&\left [ f_k^+  \left ( \prod_{j > k} f_j  \right ) \right ]  (\omega,{\bm n})\leq \frac{C_{h,{\bm \alpha}_k^+}(\xi) }{\left (\kappa-2d\sinh \xi \right)^{|{\bm \alpha}_k^+|+1}} \prod\limits_{j > k}\frac{C_{h,{\bm \alpha}_j}(\xi)}{\left (\kappa-2d \sinh \xi \right)^{|{\bm \alpha}_j|+1}} e^{-\xi |{\bm n}|},
\end{split}
\end{equation}
whenever $\xi<\sinh^{-1}  (\kappa/2d)$. 
In the following, we will denote the product of the two coefficients appearing in front of $e^{-\xi |{\bm n}|}$ by $C(\xi)$:
\begin{equation}
C(\xi)=\frac{C_{h,{\bm \alpha}_k^-}(\xi)C_{h,{\bm \alpha}_k^+}(\xi)\prod_{j \neq k} C_{h,{\bm \alpha}_j}(\xi)}{\left (\kappa-2d\sinh \xi \right)^{N+1+\sum_{j = 1}^N|{\bm \alpha}_j|}}.
\end{equation}
Since $C_{h,{\bm \alpha}_k^-}(\xi)C_{h,{\bm \alpha}_k^+}(\xi) \leq C_{h,{\bm \alpha}_k}(\xi)$, we can actually simplify the expression of $C(\xi)$ to:
\begin{equation}
C(\xi)=\frac{\prod_{j=1}^N C_{h,{\bm \alpha}_j}(\xi)}{\left (\kappa-2d\sinh \xi \right)^{N+1+\sum_{j = 1}^N|{\bm \alpha}_j|}},
\end{equation}
which is important because now $C(\xi)$ is independent of $k$ and is determined only by the initial data.

Now, as we have already seen in Corollary~\ref{EX3}, the exponential localization of the elements implies the exponential localization of the matrix elements of their operator representation. Since $\langle {\bm m}|V_\omega -V_{\mathfrak{p}_{\Lambda_{\cal N}}\omega}|{\bm m}\rangle$ is non-zero only if ${\bm m}$ is located in the complement $\Lambda_{\cal N}^c$ of the large box $\Lambda_{\cal N}$ and $|\langle {\bm m}|V_\omega -V_{\mathfrak{p}_{\Lambda_{\cal N}}\omega}|{\bm m}\rangle|\leq W$, we can conclude:
\begin{equation}
\begin{split}
&\left | \frac{1}{|\tilde{\Lambda}_{\cal N}|}\mathrm{Tr}\left \{\pi_\omega \left (\prod_{j=1}^N f_j \right ) \chi_{\tilde{\Lambda}_{\cal N}} \right \}-\frac{1}{|\tilde{\Lambda}_{\cal N}|}\mathrm{Tr}\left \{\pi_{\mathfrak{p}_{\Lambda_{\cal N}}\omega} \left (\prod_{j=1}^N f_j \right ) \chi_{\tilde{\Lambda}_{\cal N}} \right \} \right | \medskip \\
& \ \ \ \leq \left (\sum_{j=1}^N  2^{|{\bm \alpha}_j|} \right ) C(\xi) \frac{1}{|\tilde{\Lambda}_{\cal N}|}\sum_{{\bm n}\in \tilde{\Lambda}_{\cal N}}\sum_{{\bm  m}\in \Lambda^c} e^{-2\xi|{\bm n}-{\bm m}|}|\langle {\bm m}|V_\omega -V_{\mathfrak{p}_{\Lambda_{\cal N}}\omega}|{\bm m}\rangle| \medskip \\
& \ \ \ \leq  W \left (\sum_{j=1}^N  2^{|{\bm \alpha}_j|} \right ) C(\xi) \frac{1}{|\tilde{\Lambda}_{\cal N}|} \sum_{{\bm n}\in \tilde{\Lambda}_{\cal N}}\sum_{{\bm  m}\in \Lambda^c_{\cal N}} e^{-2\xi|{\bm n}-{\bm m}|},
\end{split}
\end{equation}
where $\sum_{j=1}^N  2^{|{\bm \alpha}_j|}$ counts the number of generic terms, like those in Eq.~\ref{Generic}, generated by the expansion of the original difference.
Following the notation from Fig.~\ref{SumDiagram}:
\begin{equation}\label{Suma}
\frac{1}{|\tilde{\Lambda}_{\cal N}|}\sum_{{\bm  n}\in \tilde{\Lambda}_{\cal N}}\sum_{{\bm  m}\in \Lambda_{\cal N}^c} e^{-2\xi|{\bm n}-{\bm m}|} = \frac{1}{|\tilde{\Lambda}_{\cal N}|}\sum_{j=1}^{2d} \sum_{{\bm  n}\in \tilde{\Lambda}_{\cal N}} \sum_{k=1}^\infty \sum_{{\bm \gamma}} e^{-2\xi\sqrt{({\cal N}-n_j+k)^2+|{\bm \gamma}|^2 } }.
 \end{equation}
 Due to the symmetry, each term in the sum over $j$ will contribute with the same amount. Also, we will let ${\bm \gamma}$ take values in the whole $\mathbb{Z}^{d-1}$ and use the inequality $\sqrt{2(a^2+b^2)}\geq a +b$, to continue:
\begin{equation}
\ldots \leq \frac{2d}{|\tilde{\Lambda}_{\cal N}|} \sum_{{\bm  n}\in \tilde{\Lambda}_{\cal N}} \sum_{k=1}^\infty \sum_{{\bm \gamma} \in \mathbb{Z}^{d-1}} e^{-\sqrt{2}\xi ({\cal N}-n_1+k+|{\bm \gamma}|) }.
\end{equation} 
At this point we combine the sums over $k$ and over ${\bm \gamma}$ into one sum over a ${\bm \gamma}$ in $\mathbb{Z}^d$, to continue as:
 \begin{equation}
 \begin{split}
\ldots  & \leq \frac{d}{|\tilde{\Lambda}_{\cal N}|} \sum_{{\bm  n}\in \tilde{\Lambda}_{\cal N}}  \sum_{{\bm \gamma} \in \mathbb{Z}^d} e^{-\sqrt{2}\xi ({\cal N}-n_1+|{\bm \gamma}|) } \medskip \\
 & = \frac{6d}{2{\cal N}+1} \sum_{n_1 =-{\cal N}/3}^{{\cal N}/3}e^{-\sqrt{2}\xi ({\cal N}-n_1)} \sum_{{\bm \gamma} \in \mathbb{Z}^{d}} e^{-\sqrt{2}\xi |{\bm \gamma}| } \medskip \\
 & \sim \frac{3d}{2} A_d(\sqrt{2}\xi) {\cal N}^{-1} e^{-\frac{2\sqrt{2}}{3}\xi {\cal N}}/(1-e^{-\sqrt{2}\xi}),
\end{split}
\end{equation}
where the last estimate holds in the asymptotic limit of large ${\cal N}$'s. Here $A_d(\xi)$ is defined by:
\begin{equation}\label{TheA}
A_d(\xi)=\sum_{{\bm \gamma} \in \mathbb{Z}^d} e^{-\xi |{\bm \gamma}| }.
\end{equation}  
Using a basic property of $A_d(\xi)$, we can actually simplify the above upper bound, to conclude:
\begin{equation}\label{TT1}
\frac{1}{|\tilde{\Lambda}_{\cal N}|}\sum_{{\bm  n}\in \tilde{\Lambda}_{\cal N}}\sum_{{\bm  m}\in \Lambda_{\cal N}^c} e^{-2\xi|{\bm n}-{\bm m}|} \leq  \frac{3d}{4} A_{d+1}(\sqrt{2}\xi) {\cal N}^{-1} e^{-\frac{2\sqrt{2}}{3}\xi {\cal N}}.
\end{equation} 
Collecting everything together:
\begin{equation}
\begin{split}
&\left | \frac{1}{|\tilde{\Lambda}_{\cal N}|}\mathrm{Tr}\left \{\pi_\omega \left (\prod_{j=1}^N f_j \right ) \chi_{\tilde{\Lambda}_{\cal N}} \right \}-\frac{1}{|\tilde{\Lambda}_{\cal N}|}\mathrm{Tr}\left \{\pi_{\mathfrak{p}_{\Lambda_{\cal N}}\omega} \left (\prod_{j=1}^N f_j \right ) \chi_{\tilde{\Lambda}_{\cal N}} \right \} \right | \medskip \\
& \leq  \frac{3d}{4}W \left ( \sum_{j=1}^N  2^{|{\bm \alpha}_j|} \right ) A_{d+1}(\xi)C(\xi) {\cal N}^{-1} e^{-\frac{2\sqrt{2}}{3}{\cal N}}.
\end{split}
\end{equation}

\begin{figure}
\center
  \includegraphics[width=5cm]{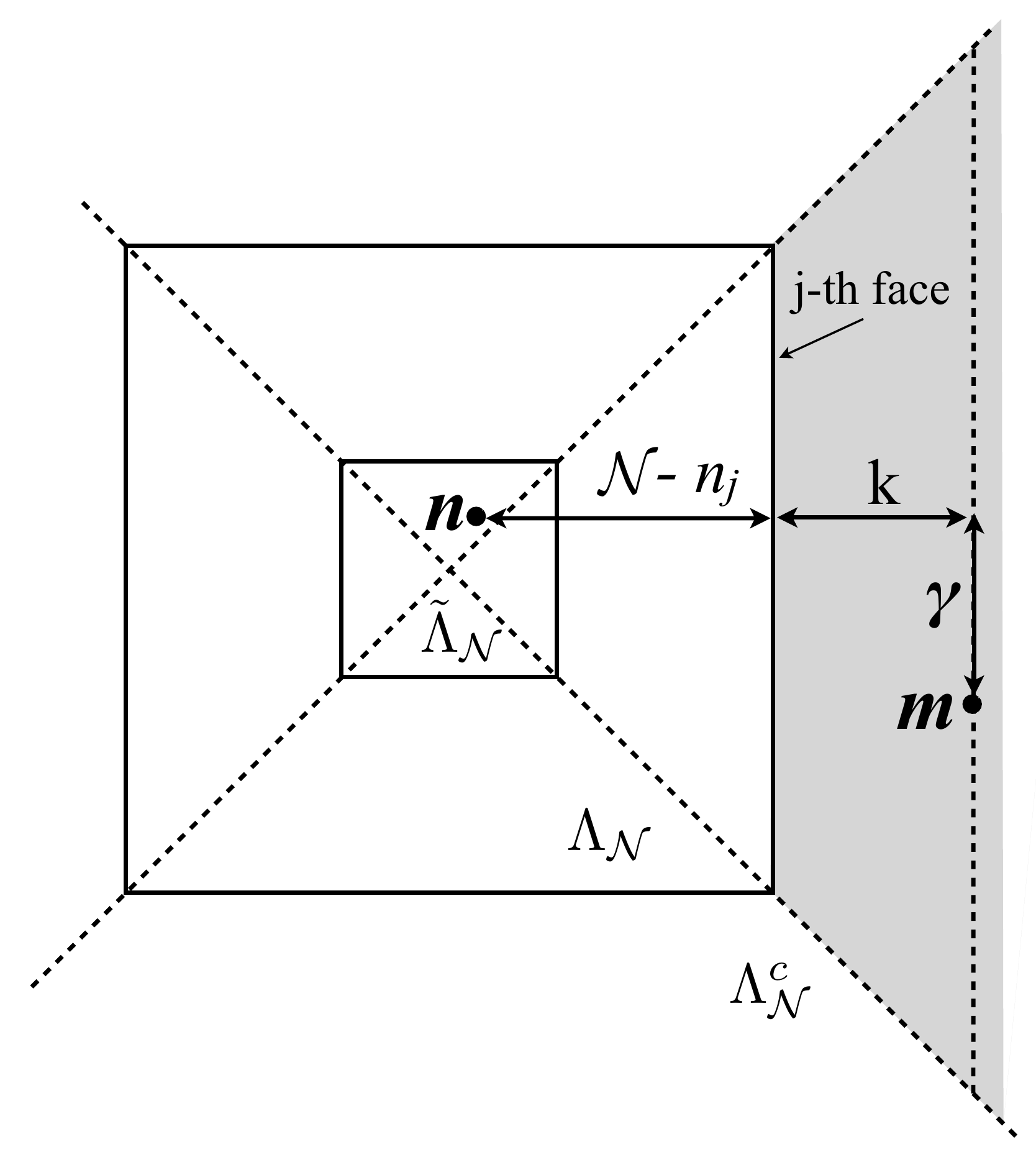}\\
  \caption{The summation over ${\bm m}$ in Eq.~\ref{Suma} is done by partitioning $\Lambda_{\cal N}^c$ is domains. There is one domain for each face of $\Lambda_{\cal N}$. The domain for the $j$-th face is shown as a shaded region. Then one fixes a distance $k$ from ${\bm m}$ to $\Lambda_{\cal N}$ and let ${\bm m}$ walk parallel to the $j$-th face, by varying the coordinate $\gamma$. The step is then repeated for all allowed $k$'s.}
 \label{SumDiagram}
\end{figure}

At this point we have an exponentially small upper bound on the difference written in Eq.~\ref{Difference}, but there is one unsatisfying aspect, related to the fact that the periodically repeating cell is $\Lambda_{\cal N}$ while the trace in Eq.~\ref{Difference} is taken only over the smaller box $\tilde{\Lambda}_{\cal N}$. We will correct this by using a special shuffling map $\mathfrak{s}$ which, in each repeating cell, switches the inner box with one of the remaining $3^d-1$ boxes inside that cell. Such shuffling map induces a natural map on $\Omega$ that leaves invariant the probability measure $dP(\omega)$. We will use the same symbol $\mathfrak{s}$ for this map. It is important to note that $\mathfrak{s}$ commutes with the projector $\mathfrak{p}_{\Lambda_{\cal N}}$, which is true because the shuffling occurs only between the small boxes from the same repeating cell. The shuffling map also induces a canonical map on $\ell^2(\mathbb{Z}^d)$, to be denoted by $\hat{\mathfrak{s}}$. Then we have successively:
\begin{equation}
\begin{split}
&\int dP(\omega) \ \mathrm{Tr}\left \{(\pi_{\mathfrak{p}_{\Lambda_{\cal N}}\omega} f)\chi_{\tilde {\Lambda}_{\cal N}} \right \}= \int dP(\omega) \ \mathrm{Tr}\left \{(\pi_{(\mathfrak{p}_{\Lambda_{\cal N}} \circ \mathfrak{s})\omega} f)\chi_{{\tilde {\Lambda}_{\cal N}}}\right \} \medskip \\
&=\int dP(\omega) \ \mathrm{Tr}\left \{(\pi_{(\mathfrak{s}\circ \mathfrak{p}_{\Lambda_{\cal N}})\omega} f)\chi_{{\tilde {\Lambda}_{\cal N}}}\right \}=\int dP(\omega) \ \mathrm{Tr}\left \{\hat{\mathfrak{s}}(\pi_{\mathfrak{p}_{\Lambda_{\cal N}} \omega} f)\hat{\mathfrak{s}}^{-1}\chi_{\tilde {\Lambda}_{\cal N}}\right \} \medskip \\
&=\int dP(\omega) \ \mathrm{Tr}\left \{(\pi_{\mathfrak{p}_{\Lambda_{\cal N}} \omega} f)\hat{\mathfrak{s}}^{-1}\chi_{\tilde {\Lambda}_{\cal N}}\hat{\mathfrak{s}}\right \}=\int dP(\omega) \ \mathrm{Tr}\left \{(\pi_{\mathfrak{p}_{\Lambda_{\cal N}} \omega} f)\chi_{\mathfrak{s}{\tilde {\Lambda}_{\cal N}}} \right \}.
\end{split}
 \end{equation}
 The conclusion is that, when averaging over $\omega$, the above trace does not change its value when $\tilde {\Lambda}_{\cal N}$ is shuffled for another small box, so we can use $3^d-1$ shuffling operation to patch the entire $\Lambda_{\cal N}$ and to conclude:
\begin{equation}
\frac{1}{|{\tilde {\Lambda}_{\cal N}}|}\int dP(\omega) \ \mathrm{Tr}\left \{(\pi_{\mathfrak{p}_{\Lambda_{\cal N}} \omega} f)\chi_{\tilde {\Lambda}} \right \} = \frac{1}{|\Lambda_{\cal N}|}\int dP(\omega) \ \mathrm{Tr}\left \{(\pi_{\mathfrak{p}_{\Lambda_{\cal N}} \omega} f ) \chi_{\Lambda_{\cal N}} \right \}.
\end{equation} 
Therefore:
\begin{equation}\label{Rez1}
\begin{split}
&\frac{3d}{4} W \left ( \sum_{j=1}^N  2^{|{\bm \alpha}_j|} \right ) A_{d+1}(\xi)C(\xi) {\cal N}^{-1} e^{-\frac{2\sqrt{2}}{3}{\cal N}} \geq \medskip \\
&\left | \int dP(\omega) \left [ \frac{1}{|{\tilde {\Lambda}_{\cal N}}|} \mathrm{Tr}\left \{\pi_\omega \left (\prod_{j=1}^N f_j \right ) \chi_{\tilde {\Lambda}_{\cal N}} \right \}-\frac{1}{|\Lambda_{\cal N}|} \mathrm{Tr}\left \{\pi_{\mathfrak{p}_{\Lambda_{\cal N}} \omega} \left (\prod_{j=1}^N f_j \right ) \chi_{\Lambda_{\cal N}} \right \} \right ] \right | \medskip \\
&=\left |{\cal T}\left(\prod_{j=1}^N f_j\right)-\frac{1}{|\Lambda_{\cal N}|}\int dP(\omega) \mathrm{Tr}\left \{\pi_{\mathfrak{p}_{\Lambda_{\cal N}}\omega} \left (\prod_{j=1}^N f_j \right ) \chi_{\Lambda_{\cal N}} \right \} \right |.
\end{split}
\end{equation}
Finally, the integrand of the remaining integral over $\omega$ is independent of the values of $\omega$ outside the box $\Lambda_{\cal N}$, so those components of $\omega$ can be integrated out and what is left is just an integral over $[-\frac{1}{2},\frac{1}{2}]^{|\Lambda_{\cal N}|}$ with the measure written in Eq.~\ref{ML}. Then, from Proposition~\ref{QQ}, it follows that the remaining integral is precisely ${\cal T}_{\mathrm{per}}(\prod_{j=1}^N \mathfrak{P}_{\Lambda_{\cal N}} f_j)$.

If $f_j$'s are of the form given in Eq.~\ref{fis}, we have to perform one additional operation, namely, we have to apply the following operation:
\begin{equation}
\left(\frac{i}{2\pi}\right)^N \int_{{\cal C}_\kappa} dz_1 \ \Phi_1(z_1)  \ldots \int_{{\cal C}_\kappa} dz_N \  \Phi_N(z_N) \ \big \{\ldots  \big \},
\end{equation}
on the difference we just estimated in Eq.~\ref{Rez1}. As before, ${\cal C}_\kappa$ is the contour defined by $\mathrm{dist}(z,\sigma(h))=\kappa$. The upper bound derived in Eq.~\ref{Rez1} is independent of the actual position of the $z_j$'s on the contour ${\cal C}_\kappa$, so we can automatically conclude:
\begin{equation}\label{Final23}
\begin{split}
&\left |{\cal T}\left (\prod_{j=1}^N \partial_{{\bm \alpha}_j}\Phi_j(h)\right )-{\cal T}_{\mathrm{per}} \left (\prod_{j=1}^N \partial_{{\bm \alpha}_j}\Phi_j(\mathfrak{P}_\Lambda h)\right) \right | \medskip \\ 
&\leq   \frac{3d}{4} W\left ( \sum_{j=1}^N  2^{|{\bm \alpha}_j|} \right ) A_{d+1}(\xi) C(\xi) \left (\prod\limits_{j=1}^N \bar{\Phi}_j \right ) {\cal N}^{-1} e^{-\frac{2\sqrt{2}}{3}{\cal N}},
\end{split}
\end{equation}
and this concludes the proof. The expression of the factor $\mathfrak{B}_1(\xi,\{{\bm \alpha}\})$, in the asymptotic limit of large ${\cal N}$, can be easily read from Eq.~\ref{Final23}.\qed
\end{proof}

\section{Approximating algebras: Second Round}

Here we go one step further and wrap the repeating cell $\Lambda_{\cal N}$ of the periodic algebra into a torus. We define a new $C^*$-algebra over a $d$-dimensional discrete torus and establish approximate morphisms between the periodic algebra and the torus algebra. The main result of the Section is that the analytic functional calculus with short-range elements from the periodic algebra is well approximated by the analytic functional calculus with elements from the algebra over the discrete torus. 

\subsection{\underline{A $C^*$-algebra over the discrete torus}}

We call a discrete circle ${\cal S}^1_D$ a discrete sequence of equally spaced points on a circle in a 2-dimensional plane. The $d$-dimensional discrete torus ${\mathbb T}_d$ will be $({\cal S}^1_D)^{\times d}$. We will use letters like ${\bm p}$, ${\bm q}$, etc., to specify points on the discrete torus, while reserving ${\bm n}$ and ${\bm m}$ exclusively for points of $\mathbb{Z}^d$. On this torus we pick an arbitrary point and call it the origin ${\bm o}$. Given a point ${\bm p}$ of the torus, we can use a succession of rigid rotations along the ${\cal S}^1_D$ circles, to rotate the torus  until the origin reaches the position of the point ${\bm p}$ before the rotations. After such action, each original point of the torus has been replaced by a different point, and this substitution defines a map $\mathfrak{r}_{\bm p}$ on the torus. It is clear that $\mathfrak{r}_{\bm p}({\bm o})={\bm p}$, and that $\mathfrak{r}_{\bm p}$ is independent of the specific sequence (which is not unique) of ${\cal S}^1_D$ rotations used to bring the origin at position ${\bm p}$. Furthermore, the rotations satisfy the following commutative group relation:
\begin{equation}
\mathfrak{r}_{\bm q}\circ \mathfrak{r}_{\bm p}=\mathfrak{r}_{\mathfrak{r}_{\bm q}{\bm p}}=\mathfrak{r}_{\bm p}\circ \mathfrak{r}_{\bm q}=\mathfrak{r}_{\mathfrak{r}_{\bm p}{\bm q}}.
\end{equation}

It is convenient to require that each ${\cal S}^1_d$ circles contain $2{\cal N}+1$ points, though this requirement is not essential. We introduce a coordinate system on ${\mathbb T}_d$ by naturally un-rolling the torus onto the points $\{-{\cal N},\ldots,0,\ldots,{\cal N}\}^d$ of $\mathbb{Z}^d$, with the origin pinned at ${\bm 0}$. The coordinates of a point ${\bm p} \in \mathbb{T}_d$ will be denoted by ${\bm n}_{\bm p}$. The set $\{-{\cal N},\ldots,0,\ldots,{\cal N}\}^d \subset \mathbb{Z}^d$ will be denoted by $\Lambda_{\cal N}$. This $\Lambda_{\cal N}$ will later be connected with the repeating cell of the periodic algebra, hence the notation.

We let $\omega=\{\omega_{\bm p}\}_{{\bm p}\in \mathbb{T}_d}$ be a sequence of identical independent random variables with a uniform distribution in $[-1/2,1/2]$. Then $\omega$ can be viewed as a point of the probability space $\Omega_{\cal N}=[-1/2,1/2]^{ |\mathbb{T}_d|}$ endowed with the probability measure: 
\begin{equation}
dP_{\mathbb{T}}(\omega)=\prod_{{\bm p}\in \mathbb{T}_d} d\omega_{\bm p}.
\end{equation}
The rotations $\mathfrak{r}$ introduced in the previous paragraph induce a group of automorphisms on $\Omega_{\cal N}$, which leave the measure $dP_{\mathbb{T}}(\omega)$ invariant. We will use the same notation $\mathfrak{r}$ for the elements of this group.

The $C^*$-algebra $C^*(\Omega_{\cal N} \times \mathbb{T}_d,{\bm F})$ over the torus is defined as follows. The elements are continous functions $\tilde{f}: \Omega_{\cal N}\times \mathbb{T}_d \rightarrow \mathbb{C}$ and the law of composition is:
\begin{equation}
(\tilde{f}*\tilde{g})(\omega,{\bm p}) = \sum_{{\bm q}\in \mathbb{T}_d} \tilde{f}(\omega,{\bm q})\tilde{g}(\mathfrak{r}_{\bm q}^{-1}\omega,\mathfrak{r}_{\bm q}^{-1} {\bm p}) e^{i \pi ({\bm n}_{\bm p}\cdot {\bm F} \cdot {\bm n}_{\bm q})}.
\end{equation}
The operator representations on $\ell^2(\mathbb{T}_d) $ can be defined as: 
\begin{equation}\label{TRep}
\left((\tilde{\pi}_\omega \tilde{f})\phi\right)({\bm p})=\sum_{{\bm q}\in \mathbb{T}_d} \tilde{f} (\mathfrak{r}_{\bm p}^{-1}\omega,\mathfrak{r}_{\bm p}^{-1}{\bm q}) e^{i \pi ({\bm n}_{\bm q}\cdot {\bm F} \cdot {\bm n}_{\bm p})}\phi({\bm q}).
\end{equation}

\begin{proposition}\label{Mo} The operator representation $\tilde{\pi}$ defines an injective algebra morphism:
\begin{equation}
\tilde{\pi}_\omega (\tilde{f} * \tilde{g}) = (\tilde{\pi}\tilde{f}) (\tilde{\pi}\tilde{g})
\end{equation}
if and only if all the entries of ${\bm F}$ are quantized as:
\begin{equation}\label{Q}
F_{ij} = \frac{2}{2{\cal N}+1} \times \mathrm{integer}.
\end{equation}
\end{proposition}

\noindent {\bf Remark.} This is in line with Zak's finding that the magnetic translations accept finite representations only if the above quantization is satisfied \cite{ZakPR1964vy}.

\begin{proof} A direct calculation will show that:
\begin{equation}
\left((\tilde{\pi}_\omega (\tilde{f} * \tilde{g}) )\phi\right)({\bm p})= 
\sum\limits_{{\bm q},{\bm q}'\in \mathbb{T}_d}\tilde{f}(\mathfrak{r}_{\bm p}^{-1}\omega,\mathfrak{r}_{\bm p}^{-1}{\bm q})\tilde{g}(\mathfrak{r}_{\bm q}^{-1}\omega,\mathfrak{r}_{\bm q}^{-1}{\bm q}') e^{i \pi ({\bm n}_{\mathfrak{r}_{\bm p}^{-1}{\bm q}'} \cdot {\bm F} \cdot  {\bm n}_{\mathfrak{r}_{\bm p}^{-1}{\bm q}}    +{\bm n}_{{\bm q}'}\cdot {\bm F} \cdot {\bm n}_{\bm p})}\phi({\bm q}')
\end{equation}
and
\begin{equation}
\left((\tilde{\pi}_\omega \tilde{f}) (\tilde{\pi}_\omega \tilde{g}) )\phi\right)({\bm p}) = 
\sum\limits_{{\bm q},{\bm q}'\in \mathbb{T}_d}\tilde{f}(\mathfrak{r}_{\bm p}^{-1}\omega,\mathfrak{r}_{\bm p}^{-1}{\bm q})\tilde{g}(\mathfrak{r}_{\bm q}^{-1}\omega,\mathfrak{r}_{\bm q}^{-1}{\bm q}') e^{i \pi ({\bm n}_{\bm q} \cdot {\bm F} \cdot {\bm n}_{\bm p}+{\bm n}_{{\bm q}'}\cdot {\bm F} \cdot {\bm n}_{\bm q})}\phi({\bm q}').
\end{equation}
The equality between the two expressions will hold if and only if:
\begin{equation}\label{C}
{\bm n}_{\mathfrak{r}_{\bm p}^{-1}{\bm q}'} \cdot {\bm F} \cdot {\bm n}_{\mathfrak{r}_{\bm p}^{-1}{\bm q}}   +{\bm n}_{{\bm q}'}\cdot {\bm F} \cdot {\bm n}_{\bm p}={\bm n}_{\bm q} \cdot {\bm F} \cdot {\bm n}_{\bm p}+{\bm n}_{{\bm q}'}\cdot {\bm F} \cdot {\bm n}_{\bm q} + 2 \times \mathrm{integer}.
\end{equation}
We have:
\begin{equation}
{\bm n}_{\mathfrak{r}_{\bm p}^{-1}{\bm q}}={\bm n}_{\bm q}-{\bm n}_{\bm p}+{\bm \sigma}, \ \ {\bm n}_{\mathfrak{r}_{\bm p}^{-1}{\bm q}'}={\bm n}_{{\bm q}'}-{\bm n}_{\bm p}+{\bm \sigma}',
\end{equation}
where ${\bm \sigma}$ and ${\bm \sigma}'$ are vectors with entries equal to $0$ or $\pm (2{\cal N}+1)$. Eq.~\ref{C} then becomes:
\begin{equation}
{\bm \sigma}\cdot {\bm F}\cdot ({\bm n}_{\bm p}-{\bm n}_{{\bm q}'})+{\bm \sigma}' \cdot {\bm F} \cdot ({\bm n}_{\bm q} - {\bm n}_{\bm p})-{\bm \sigma}\cdot {\bm F} \cdot {\bm \sigma}' = 2 \times \mathrm{integer},
\end{equation}
which is clearly satisfied, for all allowed ${\bm n}_{\bm p}$, ${\bm n}_{\bm q}$ and ${\bm n}_{{\bm q}'}$, only if $F_{ij}$'s are quantized as in Eq.~\ref{Q}. The injective property follows directly from the definition of $\tilde \pi$. \qed
\end{proof}

From now on, we will assume the quantization of ${\bm F}$ stated in Eq.~\ref{Q}. We define the norm on $C^*(\Omega_{\cal N} \times \mathbb{T}_d,{\bm F})$ as:
\begin{equation}\label{TorusNorm}
\|\tilde{f}\|=\sup_{\omega \in \Omega_{\cal N}} \|\tilde{\pi}_\omega \tilde{f}\|,
\end{equation}
and, given the result proven in Proposition~\ref{Mo}, this norm has the fundamental property:
\begin{equation}
\|\tilde{f}*\tilde{g}\| \leq \|\tilde{f}\| \ \|\tilde{g}\|,
\end{equation}
transforming $C^*(\Omega_{\cal N} \times \mathbb{T}_d,{\bm F})$ into a Banach algebra. Furthermore, a $*$-operation can be defined as:
\begin{equation}\label{cs}
\tilde{f}^*(\omega,{\bm p})=\overline{\tilde{f}(\mathfrak{r}_{\bm p}^{-1}\omega,\mathfrak{r}_{\bm p}^{-1}{\bm o})},
\end{equation}
and the $\tilde{\pi}$ representation and the $*$-operation defined above satisfy the essential relation:
\begin{equation}
\tilde{\pi}_\omega \tilde{f}^* = (\tilde{\pi}_\omega \tilde{f})^*.
\end{equation}
As such, the norm defined in Eq.~\ref{TorusNorm} has the fundamental property:
\begin{equation}
\|\tilde{f}*\tilde{f}^*\|=\| \tilde{f}\|^2,
\end{equation}
which makes $C^*(\Omega_{\cal N} \times \mathbb{T}_d,{\bm F})$ into a $C^*$-algebra.

The magnetic rotations on $\ell^2({\mathbb{T}_d})$ can be defined as:
\begin{equation}
(U_{\bm q}\phi)({\bm p})=e^{i \pi ({\bm n}_{\bm q}\cdot F \cdot {\bm n}_{\bm p})}\phi(\mathfrak{r}_{\bm q}^{-1}{\bm p}).
\end{equation}
With the quantization condition of Eq.~\ref{Q}, the magnetic rotations give a projective representation of the rotation group $\mathfrak{r}$. Furthermore, the following covariance property holds:
\begin{equation}
U_{\bm q} (\tilde{\pi}_\omega \tilde{f})U_{\bm q}^{-1} = \tilde{\pi}_{\mathfrak{r}_{\bm q}\omega} \tilde{f}.
\end{equation}

Lastly, we introduce a trace on $C^*(\Omega_{\cal N} \times \mathbb{T}_d,{\bm F})$ by:
\begin{equation}\label{TorusTrace}
{\cal T}_{\mathbb{T}}(\tilde{f})= \int_{\Omega_{\cal N}}  dP_{\mathbb{T}}(\omega) \ \tilde{f}(\omega,{\bm o}).
\end{equation}

\begin{proposition} The linear functional defined in Eq.~\ref{TorusTrace} satisfies all the required properties of a trace.
\end{proposition}
\begin{proof} It follows immediately from the invariance of the measure $dP_{\mathbb{T}}(\omega)$ relative to the rotations $\mathfrak{r}$.\qed
\end{proof} 

\begin{proposition}\label{TraceTorus} The trace can be computed via the equivalent formula:
\begin{equation}\label{AltTorusTrace}
{\cal T}_{\mathbb{T}}(\tilde{f})= \frac{1}{|\mathbb{T}_d|}\int_{\Omega_{\cal N}}  dP_{\mathbb{T}}(\omega) \ \mathrm{Tr}\{\tilde{\pi}_\omega \tilde{f} \}.
\end{equation}
\end{proposition}

\begin{proof} It follows immediately from the invariance of the measure $dP_{\mathbb{T}}(\omega)$ relative to the rotations $\mathfrak{r}$.\qed
\end{proof}

\subsection{\underline{A bridge between algebras}}

We will define two natural maps which will be uses to navigate between the algebras $C^*(\Omega_{\mathrm{per}}^{\cal N} \times \mathbb{Z}^d,{\bm F})$ and $C^*(\Omega_{\cal N} \times \mathbb{T}_d,{\bm F})$. First, let us note that $\Omega_{\mathrm{per}}^{\cal N}$ and $\Omega_{\cal N}$ can be identified in a canonical way. From now on, they will be considered one and same. The maps are:
\begin{equation}\label{P}
\begin{split}
&P:C^*(\Omega_{\mathrm{per}}^{\cal N} \times \mathbb{Z}^d,{\bm F})\rightarrow C^*(\Omega_{\cal N} \times \mathbb{T}_d,{\bm F}), \medskip \\
& (Pf)(\omega,{\bm p}) = f(\omega,{\bm n}_{\bm p})
 \end{split}
\end{equation}
and 
\begin{equation}\label{J}
\begin{split}
& J: C^*(\Omega_{\cal N} \times \mathbb{T}_d,{\bm F})\rightarrow C^*(\Omega_{\mathrm{per}}^{\cal N} \times \mathbb{Z}^d,{\bm F}), \medskip \\
&  (J\tilde{f})(\omega,{\bm n}) = \sum_{{\bm q}\in \mathbb{T}^d} \delta_{{\bm n},{\bm n}_{\bm q}} \tilde{f}(\omega,{\bm q}).
 \end{split}
\end{equation}
It is important to note that $(J\tilde{f})(\omega,{\bm n})$ is null for ${\bm n}$ outside of $\Lambda_{\cal N}$. These maps have several other  important properties, which we state in the following two Propositions. 

\begin{figure}
\center
  \includegraphics[width=6cm]{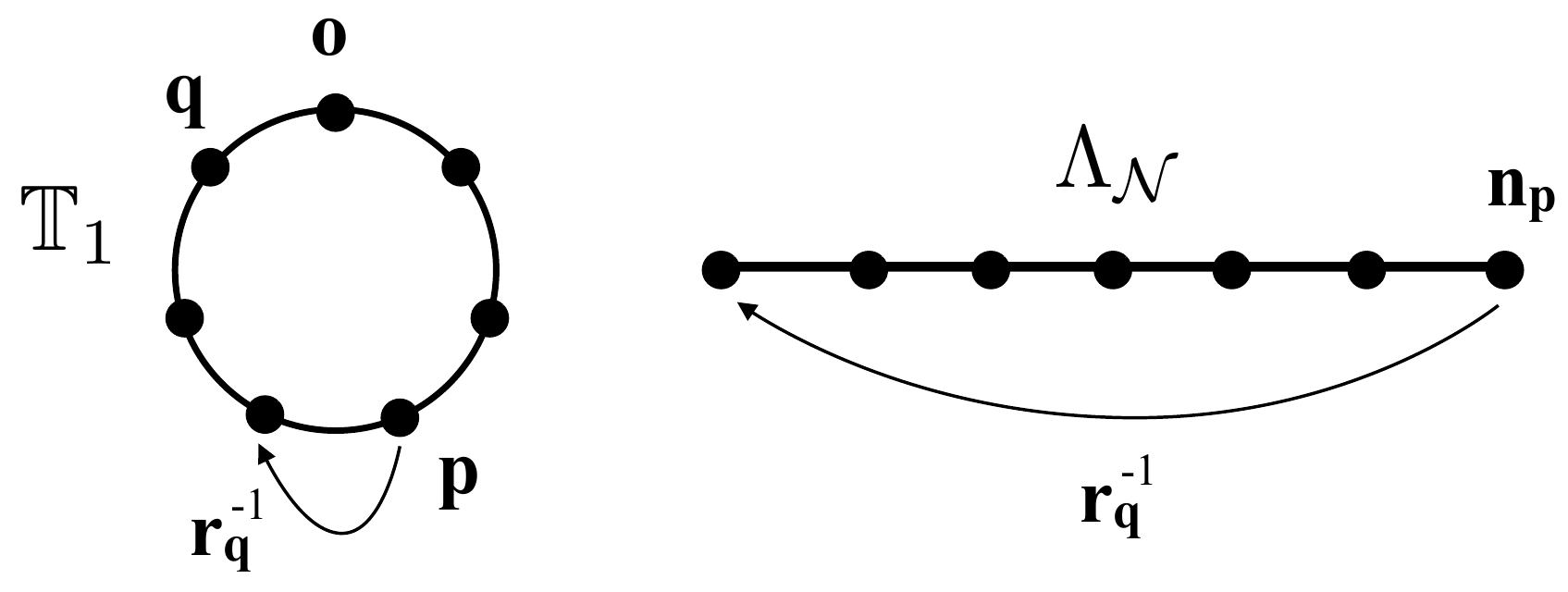}\\
  \caption{A point ${\bm p}$ with ${\bm n}_{\bm p}\in \partial \Lambda_{\cal N}$ and the set $\Sigma_{\bm p}$ for $\mathbb{T}_1$. In this simple case, $\Sigma_{\bm p}$ consists of only one point, the point ${\bm q}$. The actions of $\mathfrak{r}_{\bm q}^{-1}$ on ${\bm p}$ and ${\bm n}_{\bm p}$ are illustrated.}
 \label{Action}
\end{figure}

\begin{proposition} The maps $P$ and $J$ satisfy the relations:
\begin{equation}
\begin{split}
&P\circ J = 1_{C^*(\Omega_{\cal N} \times \mathbb{T}_d,{\bm F})} \medskip \\
&J \circ P = \chi_{\Lambda_{\cal N}}.
\end{split}
\end{equation}
\end{proposition}
\begin{proof} Indeed:
\begin{equation}
((P \circ J)\tilde{f})(\omega,{\bm p})=(J\tilde{f})(\omega,{\bm n}_{\bm p})
=\sum_{{\bm q}\in \mathbb{Z}^d} \delta_{{\bm n}_{\bm p},{\bm n}_{\bm q}} \tilde{f}(\omega,{\bm q})=\tilde{f}(\omega,{\bm p}).
\end{equation}
For the second identity:
\begin{equation}
((J \circ P)f)(\omega,{\bm n})=\sum_{{\bm q}\in \mathbb{Z}^d} \delta_{{\bm n},{\bm n}_{\bm q}}(Jf)(\omega,{\bm q}) = \left \{ 
\begin{array}{c}
f(\omega,{\bm n}) \ \ \mathrm{if} \ {\bm n} \in \Lambda_{\cal N} \medskip \\
0, \ \ \mathrm{otherwise}.
\end{array}
\right .
\end{equation}\qed
\end{proof}
The following rule of calculus will prove to be extremely useful.
\begin{proposition}\label{Bhu} Let $\tilde{f} \in C^*(\Omega_{\cal N} \times \mathbb{T}_d,{\bm B})$ such that $\tilde{f}(\omega,{\bm p})=0$ if ${\bm n}_{\bm p}>1$ (i.e. if ${\bm p}$ is not a first neighbor of the origin). Then:
\begin{equation}
(\tilde{f}*Pg)(\omega,{\bm p})=((J\tilde{f})*g)(\omega,{\bm n}_{\bm p}) +{\cal R},
\end{equation}
where
\begin{equation}
{\cal R}= \left \{
\begin{array}{l}
0 \ \mathrm{if} \ {\bm n}_{\bm p} \notin \partial \Lambda_{\cal N} \medskip \\
\sum\limits_{{\bm q}\in \Sigma_{\bm p}} \tilde{f}(\omega,{\bm q})[g(\mathfrak{t}^{-{\bm n}_{\bm q}} \omega,{\bm n}_{\mathfrak{r}^{-1}_{\bm q}{\bm p}})-g(\mathfrak{t}^{-{\bm n}_{\bm q}} \omega,{\bm n}_{\bm p}-{\bm n}_{\bm q})]e^{i \pi ({\bm n}_{\bm p} \cdot {\bm F} \cdot {\bm n}_{\bm q})}, \ \mathrm{if} \ {\bm n}_{\bm p} \in \partial \Lambda_{\cal N}.
\end{array}
\right .
\end{equation}
Here $\Sigma_{\bm p}$ is the set of those ${\bm q}$'s who are first neighbors of the origin and their corresponding rotations ${\mathfrak{r}^{-1}_{\bm q}{\bm p}}$ send ${\bm n}_{\bm p}$, which is necessarily located on one or more faces of $\Lambda_{\cal N}$, to an opposing face of $\Lambda_{\cal N}$ (see Fig.~\ref{Action}). 
\end{proposition}

\begin{proof} 
\begin{equation}
\begin{split}
(\tilde{f}*Pg)(\omega,{\bm p})&=\sum_{{\bm q}\in \mathbb{T}_d} \tilde{f}(\omega,{\bm q})(Pg)(\mathfrak{r}^{-1}_{\bm q} \omega,\mathfrak{r}^{-1}_{\bm q}{\bm p})e^{i \pi ({\bm n}_{\bm p} \cdot {\bm F} \cdot {\bm n}_{\bm q})} \medskip \\
&=\sum_{{\bm q}\in \mathbb{T}_d} \tilde{f}(\omega,{\bm q})g(\mathfrak{t}^{-{\bm n}_{\bm q}} \omega,{\bm n}_{\mathfrak{r}^{-1}_{\bm q}{\bm p}})e^{i \pi ({\bm n}_{\bm p} \cdot {\bm F} \cdot {\bm n}_{\bm q})} \medskip \\
&=\sum_{{\bm q}\in \mathbb{T}_d} ( (P\circ J)\tilde{f})(\omega,{\bm q})g(\mathfrak{t}^{-{\bm n}_{\bm q}} \omega,{\bm n}_{\mathfrak{r}^{-1}_{\bm q}{\bm p}})e^{i \pi ({\bm n}_{\bm p} \cdot {\bm F} \cdot {\bm n}_{\bm q})} \medskip \\
&=\sum_{{\bm q}\in \mathbb{T}_d} ( J\tilde{f})(\omega,{\bm n}_{\bm q})g(\mathfrak{t}^{-{\bm n}_{\bm q}} \omega,{\bm n}_{\mathfrak{r}^{-1}_{\bm q}{\bm p}})e^{i \pi ({\bm n}_{\bm p} \cdot {\bm F} \cdot {\bm n}_{\bm q})} \medskip \\
\end{split}
\end{equation}
If ${\bm n}_{\bm p} \notin \partial \Lambda_{\cal N}$, then ${\bm n}_{\mathfrak{r}^{-1}_{\bm q}{\bm p}}={\bm n}_{\bm p}-{\bm n}_{\bm q}$ because ${\bm n}_{\bm q}$ is a first neighbor of ${\bm 0}$ and we can continue as:
\begin{equation}
\begin{split}
\ldots &= \sum_{\bm q} ( J\tilde{f})(\omega,{\bm n}_{\bm q})g(\mathfrak{t}^{-{\bm n}_{\bm q}} \omega,{\bm n}_{\bm p}-{\bm n}_{\bm q})e^{i \pi ({\bm n}_{\bm p} \cdot {\bm F} \cdot {\bm n}_{\bm q})} \medskip \\
& = \sum_{{\bm m}\in \mathbb{Z}^d} ( J\tilde{f})(\omega,{\bm m})g(\mathfrak{t}^{-{\bm m}} \omega,{\bm n}_{\bm p}-{\bm m})e^{i \pi ({\bm n}_{\bm p} \cdot {\bm F} \cdot {\bm m})} \medskip \\
&=\left((J\tilde{f})*g\right)(\omega,{\bm n}_{\bm p}).
\end{split}
\end{equation}
If ${\bm n}_{\bm p} \in \partial \Lambda_{\cal N}$, then the rotation $\mathfrak{r}^{-1}_{\bm q}$ can send ${\bm n}_{\mathfrak{r}^{-1}_{\bm q}{\bm p}}$ to the opposite facet of $\Lambda_{\cal N}$, relative to ${\bm n}_{\bm p}$. Thus, the product is no longer equal to $((J\tilde{f})*g)(\omega,{\bm n}_{\bm p})$, but instead:
\begin{equation}
(\tilde{f}*Pg)(\omega,{\bm p})=\left((J\tilde{f})*g\right)(\omega,{\bm n}_{\bm p})  + \sum_{{\bm q}\in \Sigma_{\bm p}} ( J\tilde{f})(\omega,{\bm n}_{\bm q})[g(\mathfrak{t}^{-{\bm n}_{\bm q}} \omega,{\bm n}_{\mathfrak{r}^{-1}_{\bm q}{\bm p}})-g(\mathfrak{t}^{-{\bm n}_{\bm q}} \omega,{\bm n}_{\bm p}-{\bm n}_{\bm q})]e^{i \pi ({\bm n}_{\bm p} \cdot {\bm F} \cdot {\bm n}_{\bm q})}. \qed
\end{equation}
\end{proof}

\noindent {\bf Remark.} The above result can be used to demonstrate that $J$ and $P$ are approximate homomorphisms.

\subsection{\underline{The analytic functional calculus: Comparative estimates}}

Let $h \in C^*(\Omega \times \mathbb{Z}^d,{\bm F})$ be the Hamiltonian of a quantum system. Our string of thoughts was to restrict $h$ to periodic $\omega$'s via the map $\mathfrak{P}_{\Lambda_{\cal N}}$, and then to approximate $\mathfrak{P}_{\Lambda_{\cal N}}h$ on the torus by $\tilde{h}=P(\mathfrak{P}_{\Lambda_{\cal N}}h)$. If there are only a finite number of hoppings from one site to the neighboring sites, then $\mathfrak{P}_{\Lambda_{\cal N}}h$ can be recovered from the Hamiltonian on the torus via $\mathfrak{P}_{\Lambda_{\cal N}}h=J\tilde{h}$, so no truncation or loss of information occurs for the Hamiltonian itself during the last step of approximation. However, the functional calculus with the Hamiltonian will be different for the two algebras. Indeed, although $\mathfrak{P}_{\Lambda_{\cal N}}h$ and $\tilde{h}$ are similar as complex functions, the inverse of $\tilde{h} - z {\bm 1}$, which generates the analytic functional calculus on $C^*(\Omega_{\cal N} \times \mathbb{T}_d,{\bm F})$, will be sought inside $C^*(\Omega_{\cal N} \times \mathbb{T}_d,{\bm F})$, while the inverse of $\mathfrak{P}_{\Lambda_{\cal N}}h - z {\bm 1}$, which generates the analytic functional calculus on $C^*(\Omega_{\mathrm{per}}^{\cal N} \times \mathbb{Z}^d,{\bm F})$, will be sought inside $C^*(\Omega_{\mathrm{per}}^{\cal N} \times \mathbb{Z}^d,{\bm F})$. 

Our task now is to demonstrate that the two generators of the analytic functional calculi converge to each other exponentially fast when the size of the torus is increased to infinity. To keep the technical complications to the minimum, we will consider only  nearest-neighbor hopping Hamiltonian of Eq.~\ref{MainModel}. However, the main conclusion of this subsection, contained in Theorem~\ref{Comp2}, can be easily extended to any model with short-range Hamiltonian.

\begin{proposition} The following relations between the spectra always holds: 
\begin{equation}
\sigma({\tilde h}) \subset \sigma(J{\tilde h}).
\end{equation}
\end{proposition}
\begin{proof} With the quantization of Eq.~\ref{Q}, $J{\tilde h}$ becomes a periodic Hamiltonian on $\mathbb{Z}^d$. Then one can decompose $J{\tilde h}$ according to the irreducible representations of the discrete translations, which are indexed by a ${\bm k}$ point of the $d$-dimensional torus: $h=\int^{\oplus} J{\tilde h}({\bm k})d{\bm k}$. It is straightforward to show that $\tilde{h}$ is equal to $J{\tilde h}(0)$. The statement follow because $\sigma(J{\tilde h})=\cup_{\bm k} \sigma\left(J{\tilde h}({\bm k})\right)$.\qed
\end{proof}

\begin{lemma}\label{CT1} Let ${\tilde h}$ be the Hamiltonian on the torus corresponding to the nearest-neighbor hopping Hamiltonian of Eq.~\ref{MainModel}, and let $z \in \mathbb{C}$ be outside the spectrum of $J{\tilde h}$. Then:
\begin{equation}
\|(\tilde{h} - z {\bm 1})^{-1} - P( J\tilde{h}-z{\bm 1})^{-1}\|\leq \frac{4d^2 A_{d-1}(\xi/\sqrt{2})e^{-\frac{\xi}{\sqrt{2}}{\cal N}}}{\mathrm{dist}(z,\sigma(\tilde{h}))\left ( \mathrm{dist}(z,\sigma(J\tilde{h}))-2d\sinh \xi \right )},
\end{equation}
whenever $\xi<\sinh^{-1} (\mathrm{dist}(z,\sigma(J\tilde{h}))/2d)$.
\end{lemma}

\begin{proof} Let $\tilde{r}\in  C^*(\Omega_{\cal N} \times \mathbb{T}_d,{\bm F})$ be defined by
\begin{equation}\label{S}
{\bm 1}+\tilde{r} = (\tilde{h}-z {\bm 1})*P(J\tilde{h}-z{\bm 1})^{-1}
\end{equation}
According to Proposition~\ref{Bhu}:
\begin{equation}\label{W}
\tilde{r}(\omega,{\bm p})=\left \{
\begin{array}{l}
0 \ \ \mathrm{if} \ {\bm n}_{\bm p} \notin \partial \Lambda_{\cal N} \medskip \\
\begin{split}
&-\sum_{{\bm q} \in \Sigma_{\bm p}} \tilde{h}(\omega,{\bm n}_{\bm q})[(J\tilde{h}-z{\bm 1})^{-1}(\mathfrak{t}^{-{\bm n}_{\bm q}} \omega,{\bm n}_{\mathfrak{r}^{-1}_{\bm q}{\bm p}}) \medskip \\
 &\ \ \ -  (J\tilde{h}-z{\bm 1})^{-1}(\mathfrak{t}^{-{\bm n}_{\bm q}} \omega,{\bm n}_{\bm p} -{\bm n}_{\bm q})]e^{i \pi ({\bm n}_{\bm p} \cdot {\bm F} \cdot {\bm n}_{\bm q})} \ \ \mathrm{if} \ {\bm n}_{\bm p} \in \partial \Lambda_{\cal N}.
\end{split}
\end{array}
\right.
\end{equation}
The number of points in $\Sigma_{\bm p}$ is less or equal to $d$, the equality occuring when ${\bm n}_{\bm p}$ is located at a corner of $\Lambda_{\cal N}$. Note that the second argument of $(J\tilde{h}-z{\bm 1})^{-1}$ in Eq.~\ref{W} is on the boundary or outside $\Lambda_{\cal N}$ and the key here is that $(J\tilde{h}-z{\bm 1})^{-1}$ is exponentially localized. As such:
\begin{equation}
|\tilde{r}(\omega,{\bm p})| \left \{
\begin{split}
&=0 \ \ \mathrm{if} \ {\bm n}_{\bm p} \notin \partial \Lambda_{\cal N} \medskip \\
&\leq  \frac{2d e^{-\xi|{\bm n}_{\bm p}|}}{\mathrm{dist}(z,\sigma(J\tilde{h}))-2d\sinh \xi} \ \ \mathrm{if} \ {\bm n}_{\bm p} \in \partial \Lambda_{\cal N}.
\end{split}
\right .
\end{equation}
To compute the norm, we proceed as follows:
\begin{equation}\label{B1}
\begin{split}
\|\tilde{r}\| & \leq \sum_{ {\bm n}_{\bm p} \in \partial \Lambda_{\cal N}} \sup_\omega\{|r(\omega,{\bm p})|\} \medskip \\
&\leq \frac{2d }{\mathrm{dist}(z,\sigma(J\tilde{h}))-2d\sinh \xi}\sum_{ {\bm n}_{\bm p} \in \partial \Lambda_{\cal N}}e^{-\xi|{\bm n}_{\bm p}|} \medskip \\
&\leq \frac{2d}{\mathrm{dist}(z,\sigma(J\tilde{h}))-2d\sinh \xi}2d\sum_{ {\bm m} \in \mathbb{Z}^{d-1}}e^{-\xi\sqrt{{\cal N}^2 +|{\bm m}|^2}}\medskip \\
&\leq \frac{4d^2 }{\mathrm{dist}(z,\sigma(J\tilde{h}))-2d\sinh \xi}\sum_{ {\bm m} \in \mathbb{Z}^{d-1}}e^{-\frac{\xi}{\sqrt{2}}({\cal N} +|{\bm m}|)},
\end{split}
\end{equation}
to conclude:
\begin{equation}
\|\tilde{r}\|\leq \frac{4d^2 A_{d-1}(\xi/\sqrt{2})e^{-\frac{\xi}{\sqrt{2}}{\cal N}}}{\mathrm{dist}(z,\sigma(J\tilde{h}))-2d\sinh \xi}.
\end{equation}
Since $\sigma(\tilde{h}) \subset \sigma(J\tilde{h})$, we can multiply Eq.~\ref{S} with $(\tilde{h}-z{\bm 1})^{-1}$, to obtain:   
\begin{equation}
(\tilde{h}-z {\bm 1})^{-1} - P(J\tilde{h}-z{\bm 1})^{-1}=-(\tilde{h}-z {\bm 1})^{-1}*\tilde{r}
\end{equation}
and the statement follows.\qed
\end{proof}

\begin{theorem}\label{Comp2} The analytic functional calculi with the Hamiltonians corresponding to the nearest-neighbor hopping Hamiltonian of Eq.~\ref{MainModel}, in the $C^*$-algebras $C^*(\Omega_{\mathrm{per}}^{\cal N} \times \mathbb{Z}^d,{\bm F})$ and $C^*(\Omega_{\cal N} \times \mathbb{T}_d,{\bm F})$, become exponentially close as ${\cal N}$ is taken to infinity. Specifically:
\begin{equation}
\|\Phi(\tilde{h}) - P\Phi(J\tilde{h})\|\leq \frac{4d^2 A_{d-1}(\xi/\sqrt{2})\bar{\Phi}e^{-\frac{\xi}{\sqrt{2}}{\cal N}}}{\kappa\left (\kappa-2d\sinh \xi\right )},
\end{equation}
for any function $\Phi(z)$ analytic in the neighborhood $\mathrm{dist}(z,\sigma(J\tilde{h}))\leq \kappa$ of $\sigma(J\tilde{h})$ and for $\xi < \sinh^{-1} (\kappa/2d)$.
\end{theorem}

\begin{proof} The statement follows from the representation:
\begin{equation}
\Phi(\tilde{h}) - P\Phi(J\tilde{h})=\frac{i}{2\pi}\int_{{\cal C}_\kappa} \Phi(z) [(\tilde{h} - z {\bm 1})^{-1} - P(J\tilde{h} - z{\bm 1})^{-1}]dz,
\end{equation}
and from the upper bound of Lemma~\ref{CT1}.\qed
\end{proof}

\subsection{\underline{An approximate differential calculus on the torus}}

Our goal here is to demonstrate that an approximate differential calculus can be defined over the torus algebra such that, for elements generated by the analytic functional calculus with a short-range Hamiltonian, the errors introduced by the approximate derivations vanish exponentially fast as the size of the discrete torus is taken to infinity.

\begin{theorem} \label{DiffCalcTorus} Let $\mathfrak{x}: {\cal S}^1_D \rightarrow \mathbb{R}$ be a continuous function such that ($n_p$ is the coordinate of the point $p\in {\cal S}^1_D$):
\begin{equation}
|\mathfrak{x}(p)-n_p|=0 \ \mathrm{if} \ |n_p| < {\cal N}/\sqrt{2}
\end{equation}
and, 
\begin{equation}
|\mathfrak{x}(p) |\leq |n_p| \  \mathrm{for \ all} \ p \in {\cal S}^1_D.
\end{equation} 
We define an approximate differential calculus on the torus by:
\begin{equation}
(\tilde{\partial}_j \tilde{f})({\bm p}) = i\mathfrak{x}(p_j) \tilde{f}({\bm p}), \ j=1,\ldots,d.
\end{equation}
Here, $p_j$ denotes the $j$-th component of ${\bm p}$. Let $\tilde{h}$ be the Hamiltonian on the torus corresponding to the nearest-neighbor hopping Hamiltonian of Eq.~\ref{MainModel}, and let $\Phi(z)$ be an analytic function  in the neighborhood $\mathrm{dist}(z,\sigma(J\tilde{h}))\leq \kappa$ of $\sigma(J\tilde{h})$. Then for any $0<\xi<\sinh^{-1} (\kappa/2d)$, there exists a finite and fully identifiable parameter $\mathfrak{B}_2(\xi,{\bm \alpha})$ such that:
\begin{equation}
\|\tilde{\partial}_{\bm \alpha} \Phi(\tilde{h}) - P \partial_{\bm \alpha} \Phi(J\tilde{h})\| \leq \mathfrak{B}_2(\xi,{\bm \alpha})\bar{\Phi}{\cal N}^{|{\bm \alpha}|+d}  \ e^{-\frac{\xi}{\sqrt{2}} {\cal N}},
\end{equation}
for any multi-index ${\bm \alpha}=\{\alpha_1,\alpha_2, \ldots, \alpha_d\}$.
\end{theorem}

\noindent {\bf Remark.} It is very likely that the above estimates can be improved and the power law factor can be entirely eliminated.

\begin{proof} We will use the notations:
\begin{equation}
\mathfrak{x}_{\bm \alpha}({\bm p}) :=\mathfrak{x}(p_1)^{\alpha_1} \ldots \mathfrak{x}(p_d)^{\alpha_d}, \ \mathrm{and} \ n_{\bm \alpha}({\bm p}):=n_{p_1}^{\alpha_1} \ldots n_{p_d}^{\alpha_d}.
\end{equation}
We have:
\begin{equation}
\begin{split}
&[\tilde{\partial}_{\bm \alpha} \Phi(\tilde{h}) - P \partial_{\bm \alpha} \Phi(J\tilde{h})](\omega,{\bm p}) \medskip \\
&=i^{|{\bm \alpha}|}\mathfrak{x}_{\bm \alpha}({\bm p}) \Phi(\tilde{h})(\omega,{\bm p})-i^{|{\bm \alpha}|}n_{\bm \alpha}({\bm p}) (P\Phi(J\tilde{h}))(\omega,{\bm p}) \medskip \\
&=i^{|{\bm \alpha}|}\mathfrak{x}_{\bm \alpha}({\bm p}) [\Phi(\tilde{h})-P\Phi(J\tilde{h})](\omega,{\bm p})
+i^{|{\bm \alpha}|} [\mathfrak{x}_{\bm \alpha}({\bm p})- n_{\bm \alpha}({\bm p})] \Phi(J\tilde{h})(\omega,{\bm n}_{\bm p}).
\end{split}
\end{equation}
From Theorem~\ref{Comp2} and the conditions assumed on function $\mathfrak{x}$:
\begin{equation}
\left | \mathfrak{x}_{\bm \alpha}({\bm p}) [\Phi(\tilde{h})-P\Phi(J\tilde{h})](\omega,{\bm p}) \right |
\leq \frac{4d^2 A_{d-1}(\xi/\sqrt{2})\bar{\Phi}|n_{\bm \alpha}({\bm p})|e^{-\frac{\xi}{\sqrt{2}}{\cal N}}}{\kappa\left (\kappa-2d\sinh \xi \right )}.
\end{equation}
Furthermore, from Corollary~\ref{EX2}:
\begin{equation}
\left | [\mathfrak{x}_{\bm \alpha}({\bm p})- n_{\bm \alpha}({\bm p})] \Phi(J\tilde{h})(\omega,{\bm n}_{\bm p}) \right | 
\leq \frac{  \chi_{\Lambda^c_{{\cal N}/\sqrt{2}}}({\bm n}_{\bm p}) \bar{\Phi} \ |n_{\bm \alpha}({\bm p})| \ e^{-\xi |{\bm n}_{\bm p}|}}{\kappa- 2d\sinh \xi  }.
\end{equation}
These estimates and the following coarse upper bound on the operator norm:
\begin{equation}
\|\tilde{\partial}_{\bm \alpha} \Phi(\tilde{h}) - P \partial_{\bm \alpha} \Phi(J\tilde{h})\| 
\leq (2{\cal N}+1)^d\sup\limits_{{\omega},{\bm p}} \left | [\tilde{\partial}_{\bm \alpha} \Phi(\tilde{h}) - P \partial_{\bm \alpha} \Phi(J\tilde{h})](\omega,{\bm p})\right | ,
\end{equation}
allow us to conclude that, in the asymptotic limit of large ${\cal N}$'s: 
\begin{equation}
\|\tilde{\partial}_{\bm \alpha} \Phi(\tilde{h}) - P \partial_{\bm \alpha} \Phi(J\tilde{h})\|  
\leq  2^d \left [4d^2 A_{d-1}(\xi/\sqrt{2})/\kappa + 1 \right]\frac{\bar{\Phi}{\cal N}^{|{\bm \alpha}|+d}  \ e^{-\frac{\xi}{\sqrt{2}} {\cal N}}}{\kappa- 2d\sinh \xi},
\end{equation}
and the statement follows. The parameter $\mathfrak{B}_2(\xi,{\bm \alpha})$ can be identified from the last inequality.\qed 
\end{proof}

It remains to determine how the approximate differential calculus works in the operator representation and to explicitly construct and test a practical function $\mathfrak{x}$.

\begin{proposition}\label{AppDiffCalc} Let ${\cal O}$ be the discrete unit circle in the complex plane defined by the solutions of the equation $z^{2{\cal N}+1}=1$, and let 
\begin{equation}
\mathfrak{x}(p) = \sum_{\lambda \in {\cal O}} b_\lambda \lambda^{n_p}
\end{equation}
be the discrete Fourier decomposition of the function $\mathfrak{x}$. Then:
\begin{equation}
\tilde{\pi}_\omega (\tilde{\partial}_j f) = i\sum_{\lambda\in {\cal O}} b_\lambda \lambda^{-\tilde{x}_j}(\tilde{\pi}_\omega f) \lambda^{\tilde{x}_j},
\end{equation}
where $\tilde{x}_j$ is the operator $(\tilde{x}_j \phi)({\bm p})=n_{p_j} \phi({\bm p})$ on $\ell^2(\mathbb{T}_d)$.
\end{proposition} 

\begin{proof} Using the explicit expression of the operator representation given in Eq.~\ref{TRep}:
\begin{equation}
\begin{split}
[(\tilde{\pi}_\omega (\tilde{\partial}_j f))\phi]({\bm p})&=\sum_{{\bm q}\in \mathbb{T}_d} (\tilde{\partial}_j \tilde{f}) (\mathfrak{r}_{\bm p}^{-1}\omega,\mathfrak{r}_{\bm p}^{-1}{\bm q}) e^{i \pi ({\bm n}_{\bm q}\cdot {\bm F} \cdot {\bm n}_{\bm p})}\phi({\bm q}) \medskip \\
&=i\sum_{{\bm q}\in \mathbb{T}_d} \mathfrak{x}(({\mathfrak{r}_{\bm p}^{-1}{\bm q}})_j) \tilde{f}  (\mathfrak{r}_{\bm p}^{-1}\omega,\mathfrak{r}_{\bm p}^{-1}{\bm q}) e^{i \pi ({\bm n}_{\bm q}\cdot {\bm F} \cdot {\bm n}_{\bm p})}\phi({\bm q}).
\end{split}
\end{equation}
A key observation is that ${\bm n}_{{\mathfrak{r}_{\bm p}^{-1}{\bm q}}}={\bm n}_{\bm q}-{\bm n}_{\bm p}+{\bm \sigma}$ where the components of ${\bm \sigma}$ are either 0 or $\pm (2{\cal N}+1)$. As such, for all $\lambda \in {\cal O}$ we can write:
\begin{equation}
\lambda^{n_{\left(\mathfrak{r}_{\bm p}^{-1}{\bm q}\right)_j}}=\lambda^{n_{q_j}-n_{p_j}}.
\end{equation}
Then we can continue:
\begin{equation}
\begin{split}
[(\tilde{\pi}_\omega (\tilde{\partial}_j f))\phi]({\bm p})&=i\sum_{{\bm q}\in \mathbb{T}_d} \sum_{\lambda \in {\cal O}} b_\lambda \lambda^{n_{q_j}-n_{p_j}} \tilde{f}  (\mathfrak{r}_{\bm p}^{-1}\omega,\mathfrak{r}_{\bm p}^{-1}{\bm q}) e^{i \pi ({\bm n}_{\bm q}\cdot {\bm F} \cdot {\bm n}_{\bm p})}\phi({\bm q}) \medskip \\
&=i\sum_{\lambda \in {\cal O}} b_\lambda \lambda^{-n_{p_j}} \sum_{{\bm q}\in \mathbb{T}_d} \tilde{f}  (\mathfrak{r}_{\bm p}^{-1}\omega,\mathfrak{r}_{\bm p}^{-1}{\bm q}) e^{i \pi ({\bm n}_{\bm q}\cdot {\bm F} \cdot {\bm n}_{\bm p})} \lambda^{n_{q_j}} \phi({\bm q}) \medskip \\
&= i\sum_{\lambda \in {\cal O}} b_\lambda \lambda^{-n_{p_j}} \sum_{{\bm q}\in \mathbb{T}_d} \tilde{f}  (\mathfrak{r}_{\bm p}^{-1}\omega,\mathfrak{r}_{\bm p}^{-1}{\bm q}) e^{i \pi ({\bm n}_{\bm q}\cdot {\bm F} \cdot {\bm n}_{\bm p})} (\lambda^{\tilde{x}_j} \phi)({\bm q}) \medskip \\
&= i\sum_{\lambda \in {\cal O}} b_\lambda \lambda^{-n_{p_j}} [(\tilde{\pi}_\omega \tilde{f}) (\lambda^{\tilde{x}_j} \phi)]({\bm p}) \medskip \\
&= i\sum_{\lambda \in {\cal O}} b_\lambda  [\lambda^{-\tilde{x}_j}(\tilde{\pi}_\omega \tilde{f}) (\lambda^{\tilde{x}_j} \phi)]({\bm p}). \qed
\end{split}
\end{equation}
\end{proof}

To construct a practical solution, it is convenient to perform a rescaling of the coordinate system on ${\cal S}^1_D$ (hence on $\mathbb{T}_d$):
\begin{equation}
n_p \rightarrow x_p:=n_p/(2{\cal N}+1).
\end{equation}
Note that the new variable $x_p$ always take values in the interval $[-\frac{1}{2},\frac{1}{2}]$. With this rescaling,
\begin{equation}
\lambda^{n_p} \rightarrow \lambda^{(2{\cal N}+1)x_p}=\tilde{\lambda}^{x_p},
\end{equation}
where $\tilde{\lambda}$ takes the values $e^{2k\pi i}$, with $k=-{\cal N},\ldots,{\cal N}$. The function $\mathfrak{x}$ needs to be rescaled too:
\begin{equation}
\mathfrak{x} \rightarrow \tilde{\mathfrak{x}}(p)=\mathfrak{x}(p)/(2{\cal N}+1).
\end{equation}
The Fourier series of the function $\tilde{\mathfrak{x}}$ takes the form:
\begin{equation}
\tilde{\mathfrak{x}}(p)=\sum_{k=-{\cal N}}^{\cal N} b_k \tilde{\lambda}_k^{x_p},
\end{equation}
and since $\mathfrak{x}(p)$ is real and can be chosen odd relative to the reflection relative to the origin, the Fourier series reduces to:
\begin{equation}
\tilde{\mathfrak{x}}(p)= \frac{1}{i} \sum_{k=1}^{\cal N} b_k (\tilde{\lambda}_k^{x_p}-\tilde{\lambda}_k^{-x_p}), \ \mathrm{Im}(b_k) = 0.
\end{equation}
A key observation about the above rescaling is that the $\tilde{\lambda}_k$'s do not depend on ${\cal N}$ anymore (only the upper and lower limits of $k$ depend on ${\cal N}$). 
Furthermore, since the elements obtained from the analytic functional calculus are exponentially localized near the origin, the rescaling of the coordinate system concentrates these elements more and more towards the origin, as ${\cal N}$ is increased. Thus, a good approximation for the differential calculus is obtained by requiring $\tilde{\mathfrak{x}}$ to generate a good representation of the function $ix$ near the origin. In the same time, we need to make sure that the function closes continuously on the discrete circle. These two requirements are fulfilled by a Fourier series $i\sum_{k=1}^{\cal Q} b_k (\tilde{\lambda}_k^{x}-\tilde{\lambda}_k^{-x})$, where the coefficients $b_k$ are chosen such that the difference:
\begin{equation}
x-\frac{1}{i}\sum_{k=1}^{\cal Q} b_k (\tilde{\lambda}_k^{x}-\tilde{\lambda}_k^{-x})\sim o(x^{{\cal Q}+1}),
\end{equation}
asymptotically for $x$ near the origin. Here, ${\cal Q}$ is taken to be a fraction of ${\cal N}$. This condition is fulfilled if $b_k$'s are solutions to the following system of linear equations:
\begin{equation}\label{BLinSys}
\sum_{k=1}^{\cal Q}k^{2j-1}b_k = \frac{1}{4\pi} \delta_{j,1}, \ j\in \{1,\ldots,{\cal Q}\}.
\end{equation}
The matrix on the left side of the above linear system of equations is a Vandermonde matrix and, as such, its explicit inverse is known \cite{MaconAMM1958ty} and the linear system of equations can be easily solved.

To summarize, we propose:
\begin{equation}
\mathfrak{x}(p)= i (2{\cal N}+1)\sum_{k=1}^{\cal Q} b_k (\lambda_k^{n_p}-\lambda_k^{-n_p}),
\end{equation}  
with $b_k$'s satisfying the linear system of equation from Eq.~\ref{BLinSys}. With this choice:
\begin{equation}
\mathfrak{x}(p)-n_p \sim o(n_p/{\cal N})^Q,
\end{equation}
in the asymptotic limit $n_p <<{\cal N}$, which is relevant for us. Furthermore $\mathfrak{x}(p)$ closes smoothly on the discrete circle. Fig.~\ref{ApproxDiff} illustrates the behavior of $\mathfrak{x}(p)$ as ${\cal Q}$ is increased from 2 to 10, for a typical system size (${\cal N}=60$) used in our computer simulations. As one can see, with a good approximation, $\mathfrak{x}(p)$ satisfies all the requirements stated in Proposition~\ref{AppDiffCalc} when ${\cal Q}$ is given a reasonable value. In our numerical simulations we take ${\cal Q}=10$.

\begin{figure}
\center
  \includegraphics[width=9cm]{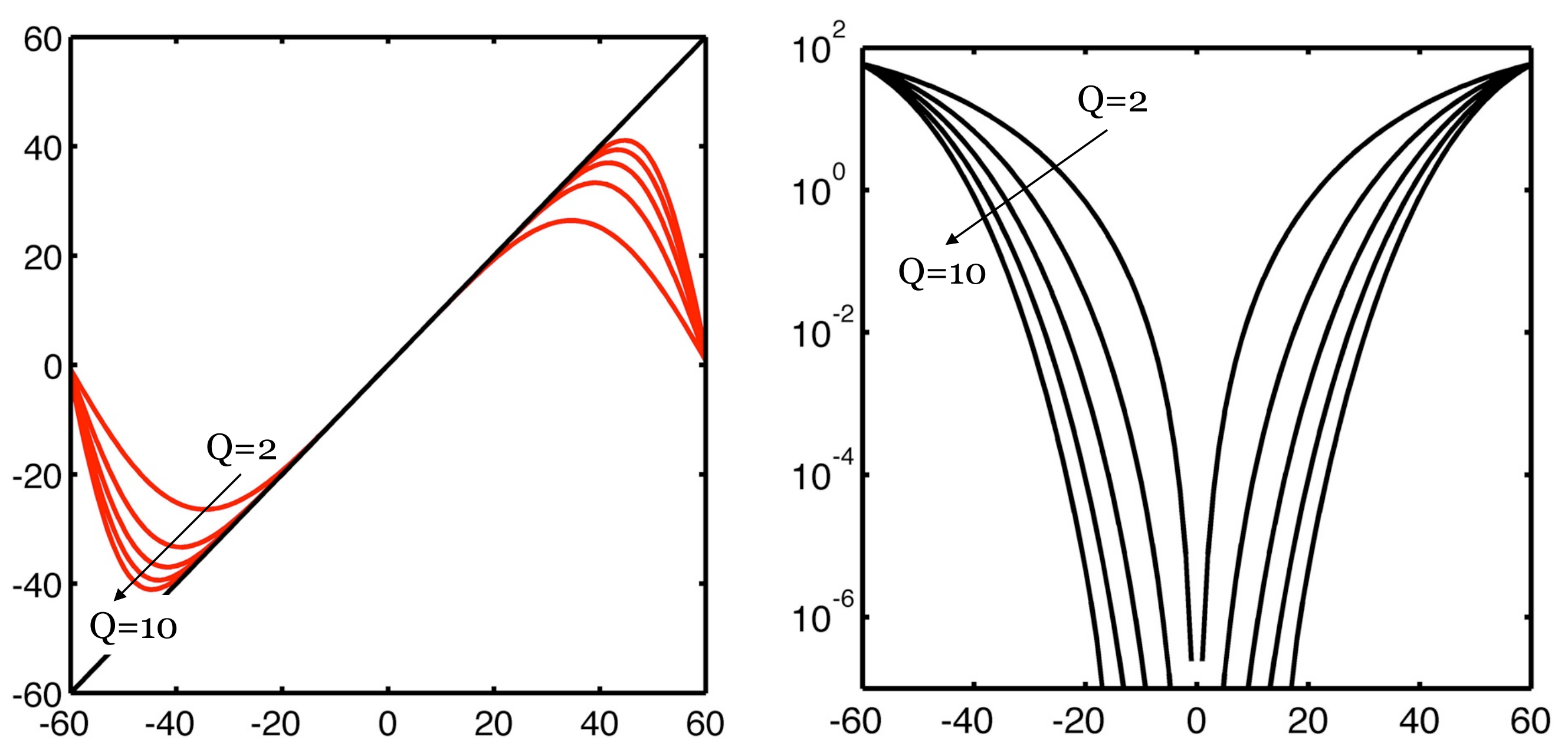}\\
  \caption{Left: Representation of the function $\mathfrak{x}(p)= i (2{\cal N}+1)\sum_{k=1}^{\cal Q} b_k (\lambda_k^{n_p}-\lambda_k^{-n_p})$, when ${\cal Q}$ is varied from 2 to 10 in increments of 2. Right: Plot of the difference $n_p-\mathfrak{x}(p)$ for the same values of ${\cal Q}$.}
 \label{ApproxDiff}
\end{figure}

\subsection{\underline{Comparative estimates for the correlation functions}}

\begin{lemma} Let $h$ be the Hamiltonian in $C^*(\Omega_{\mathrm{per}}^{\cal N} \times \mathbb{Z}^d,{\bm F})$ corresponding to the nearest-neighbor hopping Hamiltonian of Eq.~\ref{MainModel},  and $\Phi_j$ ($j=1, \ldots, N$) be a set of analytic functions  in the neighborhood of $\sigma(h)$ defined by $\mathrm{dist}(z,\sigma(h))\leq \kappa$. Let $f_j$ ($j=1,\ldots,N$) be  elements of  the $C^*$-algebra $C^*(\Omega_{\mathrm{per}}^{\cal N} \times \mathbb{Z}^d,{\bm F})$ of the form:
\begin{equation}
f_j = \partial_{{\bm \alpha}_j}\Phi_j(h), \ j=1,\ldots,N.
\end{equation} 
Then for any $0<\xi< \sinh^{-1}(\kappa/2d)$, there exists a finite and fully identifiable parameter $\mathfrak{B}_3(\xi,\{{\bm \alpha}\})$ such that:
\begin{equation}
\left |{\cal T}_{\mathrm{per}}\left(\prod_{j=1}^N f_j\right)-{\cal T}_{\mathbb{T}}\left(\prod_{j=1}^N  P f_j\right) \right | \leq  \mathfrak{B}_3(\xi,\{{\bm \alpha}\}) \left ( \prod_{j=1}^N \bar{\Phi}_j \right )e^{-\frac{\xi}{2}{\cal N}}.
\end{equation}
\end{lemma}

\begin{proof} In general:
\begin{equation}
{\cal T}_{\mathbb{T}}(P f)={\cal T}_{\mathrm{per}}(f),
\end{equation}
so we can write:
\begin{equation}
{\cal T}_{\mathrm{per}}\left (\prod_{j=1}^N f_j \right )-{\cal T}_{\mathbb{T}}\left (\prod_{j=1}^N  P f_j\right) = {\cal T}_{\mathbb{T}}\left(P\prod_{j=1}^N f_j -\prod_{j=1}^N  P f_j \right ).
\end{equation}
We see that the problems has been reduced to comparing $P\prod_{j=1}^N f_j$ and $\prod_{j=1}^N  P f_j$. We proceed as follows. Let $r_k$ denote the difference:
\begin{equation}
r_k=P\prod_{j=1}^k f_j -\left (P\prod_{j=1}^{k-1} f_j\right)*Pf_k.
\end{equation}
Then:
\begin{equation}\label{12}
P\prod_{j=1}^N f_j -\prod_{j=1}^N Pf_j =\sum_{k=1}^N r_k*\left(\prod_{j>k} Pf_j\right).
\end{equation}
We show in the following that each $r_k$'s are exponentially small. From definitions:
\begin{equation}
r_k(\omega,{\bm p})=\sum_{{\bm m}\in \mathbb{Z}^d} \left ( \prod_{j=1}^{k-1} f_j \right )(\omega,{\bm m}) f_k(\mathfrak{t}_{\bm m}^{-1}\omega,{\bm n}_{\bm p}-{\bm m}) e^{i \pi ({\bm n}_{\bm p}\cdot {\bm F} \cdot {\bm m})} 
- \sum_{{\bm q}\in \mathbb{T}_d} \left ( \prod_{j=1}^{k-1} f_j \right )(\omega,{\bm n}_{\bm q}) f_k(\mathfrak{t}_{{\bm n}_{\bm q}}^{-1}\omega,{\bm n}_{\mathfrak{r}_{\bm q}^{-1}{\bm p}}) e^{i \pi ({\bm n}_{\bm p}\cdot {\bm F} \cdot {\bm n}_{\bm q})}.
\end{equation}
Since ${\bm n}_{\mathfrak{r}_{\bm q}^{-1}{\bm p}}={\bm n}_{\bm p}-{\bm n}_{\bm q}$ if ${\bm n}_{\bm p}-{\bm n}_{\bm q} \in \Lambda_{\cal N}$, many terms cancel in the above difference, and all that remains is:
\begin{equation}
\begin{split}
r_k(\omega,{\bm p})&=\sum\limits_{{\bm m} \in \mathbb{Z}^d,{\bm n}_{\bm p}-{\bm m}\in \Lambda_{\cal N}^c} \left ( \prod_{j=1}^{k-1} f_j \right )(\omega,{\bm m}) f_k(\mathfrak{t}_{\bm m}^{-1}\omega,{\bm n}_{\bm p}-{\bm m}) e^{i \pi ({\bm n}_{\bm p}\cdot {\bm F} \cdot {\bm m})} \medskip \\
&- \sum\limits_{{\bm q}\in \mathbb{T}_d,{\bm n}_{\bm p}-{\bm n}_{\bm q}\in \Lambda_{\cal N}^c} \left ( \prod_{j=1}^{k-1} f_j \right )(\omega,{\bm n}_{\bm q}) f_k(\mathfrak{t}_{{\bm n}_{\bm q}}^{-1}\omega,{\bm n}_{\mathfrak{r}_{\bm q}^{-1}{\bm p}}) e^{i \pi ({\bm n}_{\bm p}\cdot {\bm F} \cdot {\bm n}_{\bm q})}.
\end{split}
\end{equation}
We now make use of the following bounds:
\begin{equation}
\left | \left (\prod_{j=1}^{k-1} f_j \right )(\omega,{\bm n}) \right |\leq 
e^{-\xi |{\bm n}|} \prod_{j=1}^{k-1} \frac{ C_{h,{\bm \alpha_j}}(\xi) \bar{\Phi}_j}{ ( k- 2d\sinh \xi  )^{|{\bm \alpha_j}|+1}} 
\end{equation}
and
\begin{equation}
| f_k(\omega,{\bm n})|\leq \frac{ C_{h,{\bm \alpha}_k}(\xi) \bar{\Phi}_k}{ ( k- 2d\sinh \xi )^{|{\bm \alpha}_k|+1}} e^{-\xi |{\bm n}|},
\end{equation}
to write:
\begin{equation}
|r_k(\omega,{\bm p})|\leq \frac{\prod_{j=1}^k C_{h,{\bm \alpha_j}}(\xi) \bar{\Phi}_j}{ ( k- 2d\sinh \xi )^{\sum_{j=1}^k(|{\bm \alpha_j}|+1)}} 
\times \left [ \sum\limits_{{\bm m} \in \mathbb{Z}^d,{\bm n}_{\bm p}-{\bm m}\in \Lambda_{\cal N}^c}e^{-\xi (|{\bm m}|+|{\bm n}_{\bm p}-{\bm m}|} + \sum\limits_{{\bm q}\in \mathbb{T}_d,{\bm n}_{\bm p}-{\bm n}_{\bm q}\in \Lambda_{\cal N}^c} e^{-\xi (|{\bm n}_{\bm q}|+|{\bm n}_{\mathfrak{r}_q {\bm p}}|}\right ].
\end{equation}
We use the fact that the exponent $|{\bm m}|+|{\bm n}_{\bm p}-{\bm m}|$ is larger than ${\cal N}$, $|{\bm m}|$ or $|{\bm n}_{\bm p}|$, to write:
\begin{equation}
|{\bm m}|+|{\bm n}_{\bm p}-{\bm m}|>\frac{1}{2}{\cal N}+\frac{1}{4}|{\bm m}|+\frac{1}{4}|{\bm n}_{\bm p}|.
\end{equation}
 Similarly, $|{\bm n}_{\bm q}|+|{\bm n}_{\mathfrak{r}_q {\bm p}}|$ is larger than ${\cal N}$, $|{\bm n}_{\bm q}|$ or $|{\bm n}_{\bm p}|$, so
\begin{equation}
|{\bm n}_{\bm q}|+|{\bm n}_{\mathfrak{r}_q {\bm p}}|>\frac{1}{2}{\cal N}+\frac{1}{4}|{\bm n}_{\bm q}|+\frac{1}{4}|{\bm n}_{\bm p}|.
\end{equation} 
As such:
\begin{equation}
|r_k(\omega,{\bm p})|\leq \left ( \prod_{j=1}^k\frac{ C_{h,{\bm \alpha_j}}(\xi) \bar{\Phi}_j}{ ( k- 2d\sinh \xi )^{|{\bm \alpha_j}|+1}} \right ) 2 e^{-\frac{\xi}{2}{\cal N}} e^{-\frac{\xi}{4}|{\bm n}_{\bm p}|} \sum_{{\bm m}\in \mathbb{Z}^d} e^{-\frac{\xi}{4}|{\bm m}|}.
\end{equation}
We can now estimate the norm of $r_k$ using the bound of Eq.~\ref{NormA}:
\begin{equation}
\|r_k\|\leq \frac{2A_d(\xi/4)^2 \prod_{j=1}^k  C_{h,{\alpha}_j}(\xi) \bar{\Phi}_j}{( k- 2d\sinh \xi )^{\sum_{i=1}^k(|{\bm \alpha_i}|+1)}} e^{-\frac{\xi}{2}{\cal N}}.
\end{equation}
Also:
\begin{equation}
\left \|Pf_j \right \| \leq    C_{h,{\alpha}_j}(0)\bar{\Phi}_j/\kappa^{|{\bm \alpha_j}|+1}.
\end{equation}
Finally, based on Eq.~\ref{12}, we can write:
\begin{equation}
\left |{\cal T}_{\mathbb{T}}\left (P\prod_{j=1}^N f_j -\prod_{j=1}^N Pf_j  \right ) \right |\leq \sum_{k=1}^N \|r_k\| \prod_{j>k} \|Pf_j\|,
\end{equation}
and conclude that:
\begin{equation}
\left |{\cal T}_{\mathrm{per}}\left ( \prod_{j=1}^N f_j \right ) - {\cal T}_{\mathbb{T}}\left (\prod_{j=1}^N P f_j \right ) \right | 
\end{equation}
is bounded by:
\begin{equation} 
2A_d(\xi/4)^2 \left ( \sum_{k=1}^N \prod_{j\leq k}\frac{ C_{h,{\alpha}_j}(\xi)}{( \kappa- 2d\sinh \xi )^{|{\bm \alpha_j}|+1}}  \prod_{j>k} \frac{C_{h,{\alpha}_j}(0)}{\kappa^{|{\bm \alpha_j}|+1}} \right ) \left ( \prod_{j=1}^N \bar{\Phi}_j  \right )  e^{-\frac{\xi}{2}{\cal N}}.
\end{equation}
The parameter $\mathfrak{B}_3(\xi,\{{\bm \alpha}\})$ can be explicitly read from above.\qed
\end{proof}

\begin{theorem}\label{CompEst2} Let $\tilde{h}$ be the Hamiltonian on the torus corresponding to the nearest-neighbor hopping Hamiltonian of Eq.~\ref{MainModel}, and $\Phi_j$ ($j=1, \ldots, N$), be a set of analytic functions  in the neighborhood of $\sigma(J{\tilde h})$ defined by $\mathrm{dist}(z,\sigma(J\tilde{h}))\leq \kappa$. Let $\tilde{f}_j$, $j=1,\ldots,N$, be elements of  the $C^*$-algebra $C^*(\Omega_{\cal N} \times \mathbb{T}_d,{\bm F})$ of the form:
\begin{equation}
\tilde{f}_j = \tilde{\partial}_{{\bm \alpha}_j}\Phi_j(\tilde{h}), \ j=1,\ldots,N,
\end{equation}
and let $f_j$, $j=1,\ldots,N$, be  elements of  the $C^*$-algebra $C^*(\Omega_{\mathrm{per}}^{\cal N} \times \mathbb{Z}^d,{\bm F})$ defined by:
\begin{equation}
f_j = \partial_{{\bm \alpha}_j}\Phi_j(J\tilde{h}), \ j=1,\ldots,N.
\end{equation} 
Then for any $0<\xi< \sinh^{-1}(\kappa/2d)$, there exists a finite and fully identifiable parameter $\mathfrak{B}_4(\xi,\{ {\bm \alpha}\})$ such that:
\begin{equation}
\left |{\cal T}_{\mathrm{per}}\left(\prod_{j=1}^N f_j\right)-{\cal T}_{\mathbb{T}}\left(\prod_{j=1}^N  \tilde{f}_j\right) \right | <  \mathfrak{B}_4(\xi,\{ {\bm \alpha} \}) \left (\prod_{j=1}^N \bar{\Phi}_j \right)e^{-\frac{\xi}{2}{\cal N}}.
\end{equation}
\end{theorem}

\begin{proof} We write:
\begin{equation}
{\cal T}_{\mathrm{per}}\left(\prod_{j=1}^N f_j\right)-{\cal T}_{\mathbb{T}}\left(\prod_{j=1}^N  \tilde{f}_j\right) 
= {\cal T}_{\mathrm{per}}\left(\prod_{j=1}^N f_j\right)-{\cal T}_{\mathbb{T}}\left(\prod_{j=1}^N P f_j\right)+{\cal T}_{\mathbb{T}}\left(\prod_{j=1}^N P f_j- \prod_{j=1}^N  \tilde{f}_j\right).
\end{equation}
The absolute value of the first difference, $|{\cal T}_{\mathrm{per}}(\prod_{j=1}^N f_j)-{\cal T}_{\mathbb{T}}(\prod_{j=1}^N P f_j)|$, was estimated in the previous Lemma, and the following upper bound:
\begin{equation}
\mathfrak{B}_3(\xi,\{ {\bm \alpha} \} ) \left ( \prod_{j=1}^N \bar{\Phi}_j \right )e^{-\frac{\xi}{2}{\cal N}}
\end{equation}
was established for it. The second difference can be expanded as:
\begin{equation}\label{345}
{\cal T}_{\mathbb{T}}\left (\prod_{j=1}^N P f_j - \prod_{j=1}^N  \tilde{f}_j \right )
=\sum\limits_{k=1}^N {\cal T}_{\mathbb{T}} \left ( \left (\prod_{j<k} P f_j \right )*(Pf_k - \tilde{f}_k)*\left ( \prod_{j>k}  \tilde{f}_j \right ) \right ),
\end{equation}
so we can write:
\begin{equation}
\left | {\cal T}_{\mathbb{T}}\left(\prod_{j=1}^N P f_j-\prod_{j=1}^N  \tilde{f}_j\right) \right | \leq \sum_{k=1}^N \prod_{j<k} \|Pf_j\| \ \|Pf_k - \tilde{f}_k\| \prod_{j>k}  \|\tilde{f}_j\|,
\end{equation}
and, using Theorem~\ref{DiffCalcTorus}, we can establish that this is  less than or equal to:
\begin{equation}
 \left ( \sum_{k=1}^N \mathfrak{B}_2(\xi,{\bm \alpha}_k){\cal N}^{|{\bm \alpha}_k|+d}\prod_{j \neq k}\frac{C_{h,{\bm \alpha}_j}(\xi)}{\kappa^{|{\bm \alpha}_j|+1}}  \right ) \left (\prod_{j=1}^N \bar{\Phi}_j \right)e^{-\frac{\xi}{\sqrt{2}}{\cal N}}.
\end{equation}
If $c_k(\xi)=\sup_{\cal N} {\cal N}^{|{\bm \alpha}_k|+1} e^{-\frac{\sqrt{2}-1}{2}\xi {\cal N}}$, then the absolute value of the difference of Eq.~\ref{345} is less than or equal than:
\begin{equation}
 \left ( \sum_{k=1}^N c_k(\xi)\mathfrak{B}_2(\xi,{\bm \alpha}_k)\prod_{j \neq k}\frac{C_{h,{\bm \alpha}_j}(\xi)}{\kappa^{|{\bm \alpha}_j|+1}} \right ) \left (\prod_{j=1}^N \bar{\Phi}_j \right)e^{-\frac{\xi}{2}{\cal N}},
\end{equation}
and the statement follows.\qed
\end{proof}

\section{The noncommutative Kubo formula on the discrete torus}

The noncommutative Kubo formula with a dissipation $\Gamma$ on the torus is a straightforward translation of the exact Kubo formula from Eq.~\ref{KuboFormula}. This translation can be achieved using the maps $P$ and $J$ defined in Eqs.~\ref{P} and \ref{J}, and the approximate differential calculus introduced in the previous section. It takes the form:
\begin{equation}\label{KuboFormulaTorus}
\boxed{
\tilde{\sigma}_{i j}({\cal N})= {\cal T}_{\mathbb{T}}\left ((\tilde{\partial}_i \tilde{h}) * (P \circ \Gamma \circ J+ {\cal L}_{\tilde{h}})^{-1} \tilde{\partial}_j \Phi_{\mathrm{FD}}(\tilde{h}) \right ),}
\end{equation}
where $\tilde{h}=(P \circ \mathfrak{P}_{\Lambda_{\cal N}})h$, $h$ being the Hamiltonian in the thermodynamic limit. We are now in the position to make two important statement about the above approximate Kubo. The first statement is that Eq.~\ref{KuboFormulaTorus} can be efficiently evaluated on a computer, to a point where, for example, it enabled us to gain unprecedented insight into the  transport properties of the Quantum Hall systems and topological insulators. The second statement is that, at least for the relaxation time approximation $\Gamma=1/\tau_r$, the errors decay exponentially fast as the size of the torus is increased. As such, the bulk transport coefficients can be computed with good accuracy on relatively small simulation boxes. 

\subsection{\underline{Error bounds}}

We will assume the relaxation time approximation, where $\Gamma$ is equal to $1/\tau_r$ times the identity map. The parameter $\tau_r < \infty $ is the so called relaxation time, an empirical parameter that is usually extracted from the experimental data. 

\begin{theorem} Let $\kappa$ be defined as: $\kappa=\min\{1/2\tau_r,\pi k T\}-\epsilon$, with $\epsilon$ an arbitrarily small but strictly positive constant. Then, for the nearest-neighbor hopping Hamiltonian of Eq.~\ref{MainModel}, there exists a finite and fully identifiable constant $\mathfrak{B}(\xi)$ such that:
\begin{equation}
|\sigma_{ij}-\tilde{\sigma}_{ij}({\cal N})|\leq \mathfrak{B}(\xi)e^{-\frac{1}{2}\xi {\cal N}},
\end{equation} 
for all $0<\xi< \sinh^{-1}(\kappa/2d)$.
\end{theorem}

\begin{proof} The statement will follow after we put together the comparative estimates derived in the previous sections. We will follow our general strategy and compare the exact noncommutative Kubo formula with an approximate Kubo formula from the periodic algebra, which will be subsequently compared with the approximate Kubo formula from the torus algebra. We start from the exact conductivity tensor:
\begin{equation}
\sigma_{j k}= {\cal T}\left ((\partial_j h) * (1/\tau_r + {\cal L}_h)^{-1}
\partial_k \Phi_{\mathrm{FD}}(h) \right ),
\end{equation}
and use Proposition~\ref{Laplace} in reverse to write:
\begin{equation}
\sigma_{j k}= \int_0^\infty dt e^{-t/\tau_r} {\cal T}\left ((\partial_j h) * u_h(t,0)
[\partial_k \Phi_{\mathrm{FD}}(h)] \right ).
\end{equation}
With the explicit representation of the time evolution given in Eq.~\ref{TE}, the above expression becomes:
\begin{equation}
\sigma_{jk}= \int_0^\infty dt e^{-t/\tau_r} {\cal T}\left ((\partial_j h) * e^{-ith}*
[\partial_k \Phi_{\mathrm{FD}}(h)] * e^{ith}\right )
\end{equation}
Now one can see that we can connect with the Theorem~\ref{ThR1}, which allows us to write the following comparative estimate on the integrand:
\begin{equation}\label{144}
\begin{split}
&\left | {\cal T}\left ((\partial_j h) * e^{-ith}*
[\partial_k \Phi_{\mathrm{FD}}(h)] * e^{ith}\right ) - {\cal T}_{\mathrm{per}}\left ( (\partial_j h_p) * e^{-ith_p}*
[\partial_k \Phi_{\mathrm{FD}}(h_p)] * e^{ith_p} \right ) \right | \medskip \\
&\ \ \ \leq \mathfrak{B}_1(\xi,\{{\bm \alpha}\}) \left (\prod_{j=1}^4 \bar{\Phi}_j \right ) {\cal N}^{-1} e^{-\frac{2\sqrt{2}}{3}\xi {\cal N}}.
\end{split}
\end{equation} 
where $h_p= \mathfrak{P}_{\Lambda_{\cal N}}h$. The $\Phi$'s can be easily identified as: $\Phi_1(z)=z$, $\Phi_{2,4}(z)=e^{\mp iz t}$ and $\Phi_3=\Phi_{\mathrm{FD}}(z)$. 

Furthermore, with $\tilde{h}$ defined as $\tilde{h}=(P \circ \mathfrak{P}_{\Lambda_{\cal N}})h$, one can see that $h_p=J\tilde{h}$. Hence, we  can use Theorem~\ref{CompEst2} to write:
\begin{equation}\label{133}
\begin{split}
&\left |{\cal T}_{\mathrm{per}}\left ( (\partial_j h_p) * e^{-ith_p}*
[\partial_k \Phi_{\mathrm{FD}}(h_p)] * e^{ith_p} \right )     - {\cal T}_{\mathbb{T}}\left ((\tilde{\partial}_j \tilde{h}) * e^{-it\tilde{h}}*
[\tilde{\partial}_k \Phi_{\mathrm{FD}}(\tilde{h})] * e^{it\tilde{h}}\right )\right | \medskip \\
&\ \ \ \leq \mathfrak{B}_4(\xi,\{ {\bm \alpha} \}) \left (\prod_{j=1}^4 \bar{\Phi}_j \right)e^{-\frac{\xi}{2}{\cal N}},
\end{split}
\end{equation} 
where ${\Phi}_j$'s are same functions as before. 

With $\kappa$ chosen as indicated in the text of the Theorem, the complex set defined by those $z$'s for which $\mathrm{dist}(z,\sigma(J\tilde{h}))\leq \kappa$ (which is a subset of the complex strip $|\mathrm{Im}(z)|\leq \kappa$) resides inside the analytic domain of the functions $\Phi_i(z)$'s. In fact the only function that needs to be checked is $\Phi_3$, the Fermi-Dirac function, which has poles at $E_{\bm F}+i(2n+1)\pi  k T$, $n\in \mathbb{Z}$. The functions $\Phi_{2,4}(z)$ depend on the parameter $t$ which is to be integrated from $0$ to $\infty$, with a weight factor $e^{-t/\tau_r}$. Since, $\bar{\Phi}_{2,4}(t)\leq  const. \ e^{\kappa t}$, the integration over $t$ remains finite when $\kappa$ is chosen less than $1/2\tau_r$. We then conclude:
\begin{equation}
\begin{split}
&\left |\int_0^\infty dt \ e^{-t/\tau_r} {\cal T}\left ((\partial_j h) * e^{-ith}*
[\partial_k \Phi_{\mathrm{FD}}(h)] * e^{ith}\right ) - \int_0^\infty dt \ e^{-t/\tau_r} {\cal T}_{\mathbb{T}}\left ( (\tilde{\partial}_j \tilde{h}) * e^{-it\tilde{h}}*
[\tilde{\partial}_k \Phi_{\mathrm{FD}}(\tilde{h})] * e^{it\tilde{h}} \right ) \right | \medskip \\
&\ \ \ \leq \frac{const.}{1/\tau_r - 2\kappa} \left [ \mathfrak{B}_1(\xi,\{{\bm \alpha}\})  {\cal N}^{-1} e^{-\frac{2\sqrt{2}}{3}\xi{\cal N}} + \mathfrak{B}_4(\xi,\{ {\bm \alpha} \}) e^{-\frac{1}{2}\xi{\cal N}} \right ],
\end{split}
\end{equation} 
and the statement follows.\qed
\end{proof}

\subsection{\underline{Numerical Implementation}}

The only issue remaining to be clarified is how to invert the map $1/\tau_r + {\cal L}_{\tilde{h}}$ (we will restrict the discussion to the relaxation time approximation). Formally, this does not present any fundamental difficulty because we can simply assume to work in the Hilbert space ${\cal H}$ defined by the scalar product:
\begin{equation}
(\tilde{f},\tilde{g})={\cal T}_{\mathbb{T}}\left (\tilde{f}^**\tilde{g} \right ),
\end{equation}
which was already discussed in our introductory sections. It is straightforward to demonstrate that ${\cal L}_{\tilde{h}}$ becomes a bounded normal operator (more precisely, $i{\cal L}_{\tilde{h}}$ is self-adjoint) so the inverse $(1/\tau + {\cal L}_{\tilde{h}})^{-1}$ can be computed using the ordinary spectral decomposition of ${\cal L}_{\tilde{h}}$. This is precisely the route we chose to follow here. In the following, we explain how to compute the spectral decomposition of ${\cal L}_{\tilde{h}}$. 

We first introduce an equivalence relation on $\Omega_{\cal N}$ by saying that $\omega \sim \omega'$ if there is a $q\in \mathbb{T}_d$ such that $\omega = r_{\bm q}^{-1}\omega'$. We will denote the equivalence classes by $[\omega]$. The factorization of $\Omega_{\cal N}$ in equivalence classes will be denoted by $\Omega_{\cal N}/\sim$, and the factorization of the probability measure $dP_{\mathbb{T}}(\omega)$ by $dP_{\mathbb{T}}[\omega]$ ($= dP_{\mathbb{T}}(\omega)/\sim$).

Using Proposition~\ref{TraceTorus}, the scalar product can be written as:
\begin{equation}
(\tilde{f},\tilde{g})=\frac{1}{|\mathbb{T}_d|} \int_{\Omega_{\cal N}} dP_{\mathbb{T}} (\omega) \ \mathrm{Tr}\left \{(\tilde{\pi}_\omega \tilde{f}^*) (\tilde{\pi}_\omega \tilde{g}) \right \},
\end{equation}
and due to the covariant property of the representation $\tilde{\pi}$, we observe that:
\begin{equation}\label{90}
\mathrm{Tr}\left \{(\tilde{\pi}_\omega \tilde{f}^*) (\tilde{\pi}_\omega \tilde{g})\right \}
\end{equation}
remains unchanged when $\omega$ is replaced by $\mathfrak{r}_{\bm q}\omega$, with ${\bm q}$ an arbitrary point of the torus. In other words, the trace in Eq.~{90} depends only on the equivalence class of $\omega$ so we can conclude:
\begin{equation}
(\tilde{f},\tilde{g})=\frac{1}{|\mathbb{T}_d|} \int_{\Omega_{\cal N}/\sim} dP_{\mathbb{T}} [\omega] \ \mathrm{Tr}\left \{(\tilde{\pi}_\omega \tilde{f}^*) (\tilde{\pi}_\omega \tilde{g}) \right \},
\end{equation}
where $\omega$ appearing inside the trace can be any element from the equivalence class $[\omega]$. Furthermore, since the $\tilde{\pi}_\omega \tilde{f}$ representation of an element involves only the values of $f$ on $[\omega]\times \mathbb{T}_d$, we can draw a more fundamental conclusion, namely, that the original Hilbert space ${\cal H}$ decomposes as a direct integral:
\begin{equation}
{\cal H} = \int^{\oplus}_{\Omega_{\cal N}/\sim} dP_{\mathbb{T}} [\omega] \ {\cal H}_{[\omega]}
\end{equation}
where the elements of the Hilbert space ${\cal H}_{[\omega]}$ are functions $\tilde{f}:[\omega]\times \mathbb{T}_d \rightarrow \mathbb{C}$ and the scalar product is:
\begin{equation}
(\tilde{f},\tilde{g})_{[\omega]}=\mathrm{Tr} \left \{(\tilde{\pi}_\omega \tilde{f}^*) (\tilde{\pi}_\omega \tilde{g}) \right \}.
\end{equation}
Here, again, the $\omega$ appearing inside the trace can be any element from the equivalence class $[\omega]$. The product and the star operations of the torus algebra remain well defined when restricted to ${\cal H}_{[\omega]}$. 

\begin{figure}
\center
  \includegraphics[width=6cm]{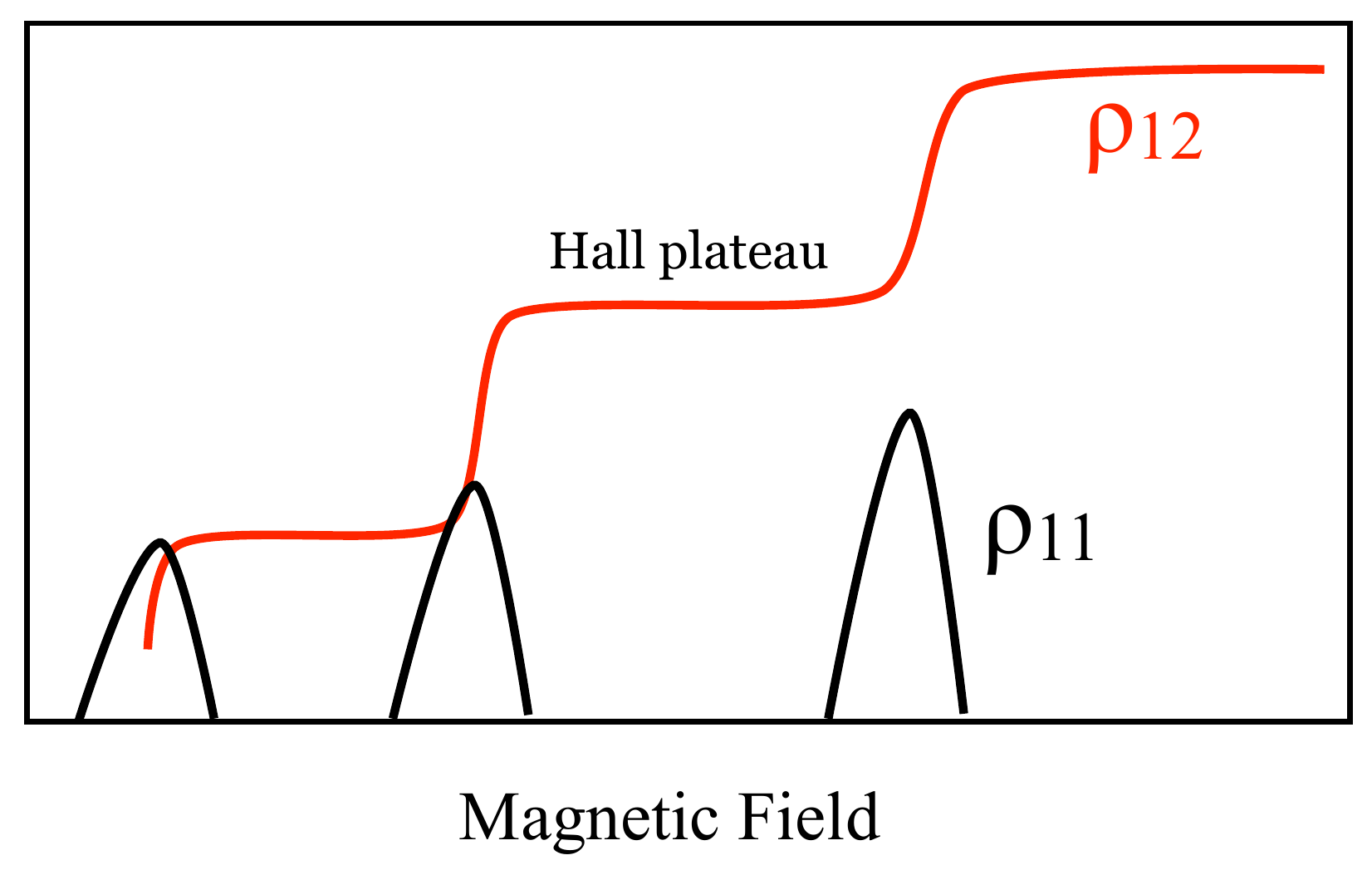}\\
  \caption{Typical output of a Hall measurement.}
 \label{ExpDiagram}
\end{figure}

Now, the operator ${\cal L}_{\tilde{h}}$ leaves invariant the Hilbert spaces ${\cal H}_{[\omega]}$ so
\begin{equation}
{\cal L}_{\tilde{h}} = \int^\oplus_{\Omega_{\cal N}/\sim} dP_{\mathbb{T}} [\omega] \ {\cal L}_{\tilde{h}}{[\omega]},
\end{equation}
where ${\cal L}_{\tilde{h}}{[\omega]}$ is simply the restriction of the map ${\cal L}_{\tilde{h}}$ to the space of functions defined over $[\omega]\times \mathbb{T}_d$. The spectral decomposition of ${\cal L}_{\tilde{h}}{[\omega]}$ can be achieved in the following way. We pick a representative $\omega$ from the class $[\omega]$ and instruct a computer to diagonalize the operator $\tilde{\pi}_\omega \tilde{h}$. Since the Hilbert space $\ell^2(\mathbb{T}_d)$ is finite, this can be achieved with the diagonalization routines from any standard linear algebra library. Let $\{\epsilon_a,\phi_a^{\omega}\}$ be the resulting eigensystem:
\begin{equation}
(\tilde{\pi}_\omega \tilde{h})\phi_a^\omega = \epsilon_a \phi_a^\omega, \ a=1,\ldots,|\mathbb{T}_d|.
\end{equation}
Now note that all $\omega$'s from $[\omega]$ can be uniquely written as $\mathfrak{r}_{\bm q}\omega$, with ${\bm q}$ a point of the torus. For ${\bm p}$ and ${\bm q}$ arbitrary points of the torus, we define:
\begin{equation}\label{57}
\tilde{f}_{ab}(\mathfrak{r}_{\bm q}\omega,{\bm p})=\phi_a^\omega ({\bm q})\overline{\phi_b^\omega(\mathfrak{r}_{\bm q}{\bm p})}e^{i \pi ({\bm n}_{\bm q}\cdot F \cdot {\bm n}_{\bm p})}, \ a,b=1,\ldots,|\mathbb{T}_d|.
\end{equation}
Eq.~\ref{57} defines genuine functions from $[\omega]\times \mathbb{T}_d$ to $\mathbb{C}$, which will be viewed as elements of ${\cal H}_{[\omega]}$. A straightforward calculation will show that:
\begin{enumerate}
\item The elements are orthonormal:
\begin{equation}
(\tilde{f}_{ab},\tilde{f}_{cd})_{[\omega]}=\delta_{ac}\delta_{bd}.
\end{equation}
\item The sequence $\{\tilde{f}_{ab}\}_{a,b=1,\ldots,|\mathbb{T}_d|}$ is complete, for if:
\begin{equation}
(\tilde{f},\tilde{f}_{ab})_{[\omega]}=0
\end{equation}
for all $a,b=1,\ldots,|\mathbb{T}_d|$, then $\tilde{f}$ is necessary the null element.
\item The elements are eigenvectors of ${\cal L}_{\tilde{h}}[\omega]$:
\begin{equation}
{\cal L}_{\tilde{h}}[\omega](\tilde{f}_{ab})=i(\epsilon_a-\epsilon_b) \tilde{f}_{ab}.
\end{equation}
\item The operator representation of $\{\tilde{f}_{ab}\}_{a,b=1,\ldots,|\mathbb{T}_d|}$ is:
\begin{equation}
\tilde{\pi}_\omega \tilde{f}_{ab}=|\phi_a^\omega\rangle \langle \phi_b^\omega|
\end{equation}
\end{enumerate}
Thus, we have accomplished the spectral decomposition of ${\cal L}_{\tilde{h}}[\omega]$ and we can conclude:
\begin{equation}
(1/\tau_r +{\cal L}_{\tilde{h}}[\omega])^{-1}(f)=\sum_{a,b=1}^{|\mathbb{T}_d|}\frac{\tilde{f}_{ab}(\tilde{f}_{ab},f)_{[\omega]}}{1/\tau_r + i(\epsilon_a-\epsilon_b)}.
\end{equation}
Given that $\tilde{\pi}_\omega \tilde{f}_{ab}=|\phi_a^\omega\rangle \langle \phi_b^\omega|$, we can now write the explicit expression for the conductivity tensor:
\begin{equation}\label{PracticalKubo}
\tilde{\sigma}_{jk} = - \int\limits_{\Omega_{\cal N}/\sim} dP_{\mathbb{T}}[\omega] \  \frac{1}{|\mathbb{T}_d|}\sum_{a,b=1}^{|\mathbb{T}_d|}\frac{\langle \phi_b^\omega|\tilde{\pi}_\omega(\tilde{\partial}_j\tilde{h})|\phi_a^\omega \rangle \langle \phi_a^\omega | \tilde{\pi}_\omega (\tilde{\partial}_k \Phi_{\mathrm{FD}}(\tilde{h})|\phi_b^\omega \rangle}{1/\tau_r + i(\epsilon_a-\epsilon_b)}.
\end{equation}
This exact formula was implemented on a computer and used in our numerical simulations.

We should comment here that computing the inverse $(1/\tau_r +{\cal L}_{\tilde{h}}[\omega])^{-1}$ is, after all, a classic linear algebra problem so it can be achieved in many different ways. All that is required is a clever implementation of the action of ${\cal L}_{\tilde{h}}[\omega]$ on the elements of ${\cal H}_{[\omega]}$. We are currently exploring various other paths to invert $1/\tau_r +{\cal L}_{\tilde{h}}[\omega]$, but for the present work we chose to use the spectral decomposition because this way we can understand and minimize the numerical errors. In fact, the computations of $\tilde{\sigma}_{jk}({\cal N})$ shown in the next section are done with the machine precision.

\section{Application: The Integer Quantum Hall Effect}

There is a tremendous amount of literature dedicated to the Integer Quantum Hall Effect (IQHE), so here we will be brisk and stay to the point. The reader can find a comprehensive discussion of the effect in the classic textbook put together by Prange and Girvin \cite{PrangeBook1987cu}. Most of the facts stated below can be found in this textbook. There have been several important developments after the publication of this book, namely the observation of IQHE in graphene \cite{Zhang:2005mu,Novoselov:2005rk,Novoselov:2006nl,Novoselov:2007fa}, in 2-dimensional Quantum spin-Hall systems \cite{Koenig:2008so} and at the surface of a 3-dimensional topological insulator \cite{ChengPRL2010gj,HanaguriPRB2010ch,XiongPhysicaE2012vb}.  It is particularly important to mention that IQHE in graphene was observed at room temperature \cite{Novoselov:2007fa}, hence making our finite temperature analysis even more relevant. An explicit lattice analysis of the IQHE in clean graphene can be found in Refs.~\cite{Bernevig:2007rr,Arai:2009na} and the effect of disorder (and interaction) at zero temperature has been analyzed in \cite{Sheng:2007pr}. Our goal here is to demonstrate what we can achieve with the new numerical procedure. The analysis of the IQHE itself, based on the new numerical results, is beyond the scope of this work, simply because that by itself could be an entirely separate research project.  

In short, IQHE is observed in the transport measurements on a 2-dimensional electron gas when subjected to a perpendicular magnetic field. Experimentally, the 2-dimensional electron gas can be obtained by trapping free-moving electrons at the interface between two semiconducting materials. A 2-dimensional electron gas is also forming at the metallic surface of a 3-dimensional topological insulator and it occurs naturally in the atomically thin graphene. The density of the 2-dimensional electron gas can be adjusted in a finely controlled manner by applying gate voltages. If the samples are of high quality and the temperature is low enough (except for graphene where IQHE can be observed at room temperatures in very high magnetic fields), the transport curves of the Hall resistivity $\rho_{12}$ (to be defined shortly) as function of the $B$-field (while holding the electron density constant), or as function of electron density (while holding $B$ constant) display a staircase pattern as schematically illustrated in Fig.~\ref{ExpDiagram}, with the horizontal plateaus being quantized to several digits of precision. In the same time, the diagonal resistivity $\rho_{11}$ is observed to vanish inside the Hall plateaus, and to have sharp peaks at the transitions between plateaus. The experimental data for IQHE is most often presented in terms of the resistivity tensor $\rho_{ij}=(\sigma^{-1})_{ij}$, rather than the conductivity tensor $\sigma_{ij}$, because the resistivity can be measured with much greater precision. Explicitly, for our isotropic model:
\begin{equation}
\begin{array}{c}
\rho_{11}=\sigma_{11}/(\sigma_{11}^2+\sigma_{12}^2) \medskip \\
\rho_{12}=\sigma_{12}/(\sigma_{11}^2+\sigma_{12}^2).
\end{array}
\end{equation} 
There is a special feature in 2-dimensions, in that the resistivity has same physical units as the resistance and conductivity as conductance.

\begin{figure}
\center
  \includegraphics[width=11cm]{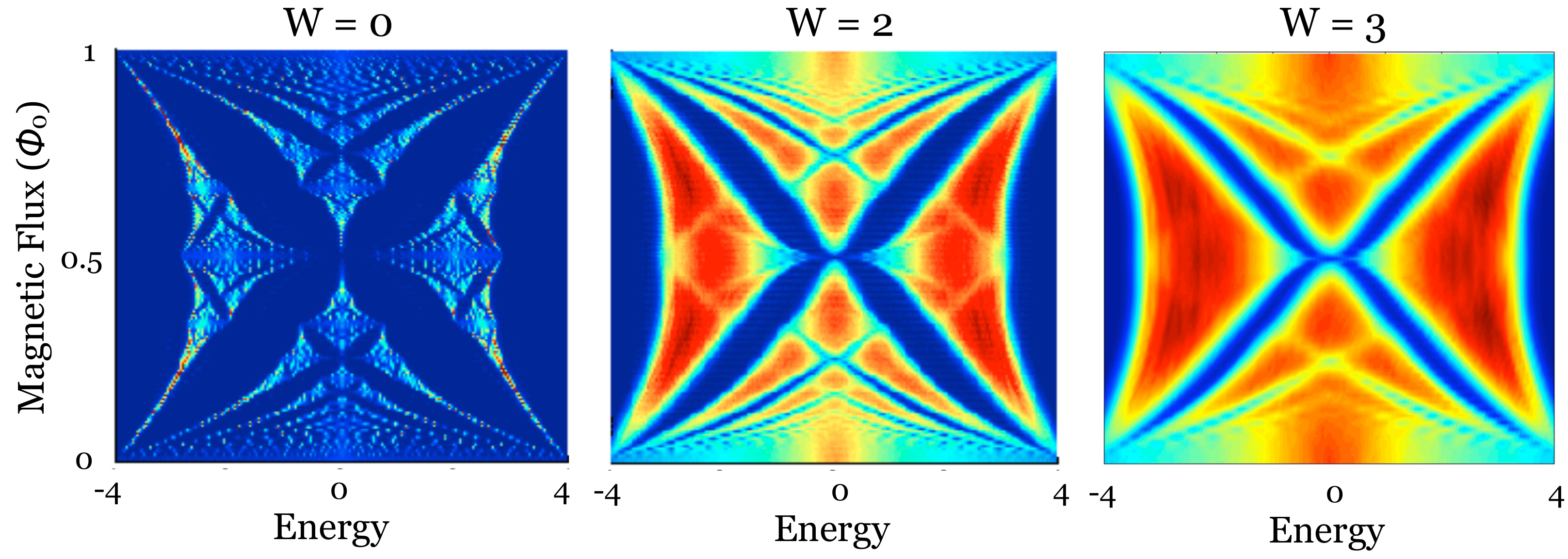}\\
  \caption{The density of states of our model Hamiltonian as function of energy and magnetic flux through the unit cell, for three degrees of disorder: W=0, 2 and 3. At W=0, one can recognize the Hofstadter butterfly.}
 \label{Spectrum}
\end{figure}

\begin{figure}
\center
  \includegraphics[width=11cm]{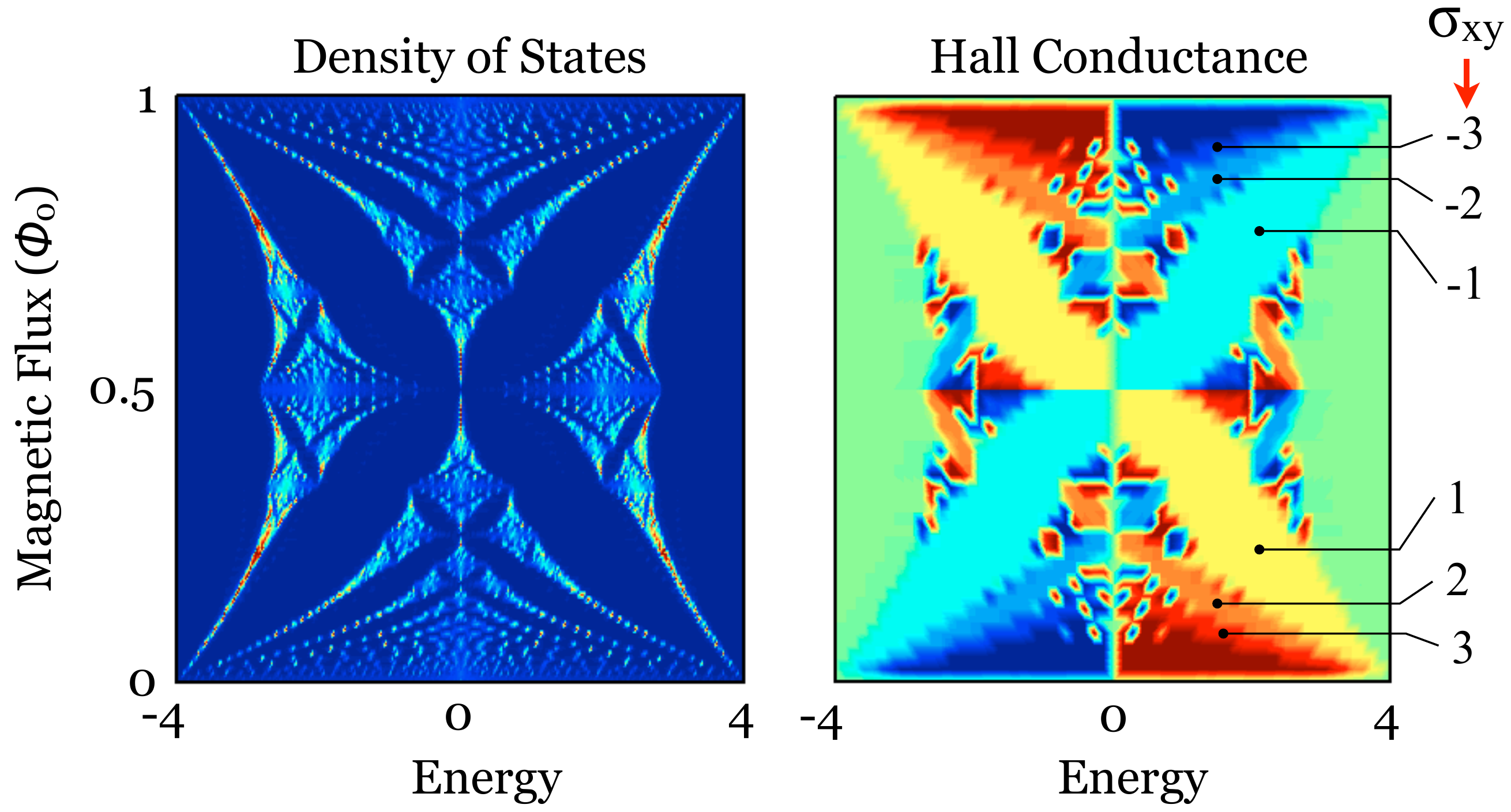}\\
  \caption{DOS (left) and color map of the Hall conductivity (right) for $W=0$. The regions of quantized Hall conductivity, which appear as well defined patches of same color and, are indicated at the right.}
 \label{HallMap1}
\end{figure}

In the absence of a periodic potential, the energy spectrum of the 2-dimensional electron gas (in a mean-filed, non-interacting approximation) consists of equally spaced Landau levels separated by spectral gaps. One can easily compute the Hall conductance $\sigma_{12}$ at $T=0$, assuming the lowest $n$ Landau levels as completely filled, in which case the result will be $\sigma_{12}=n\frac{e^2}{h}=nG_0$. While this can provide a clue about the origin of the quantized values, it cannot explain the existence of the plateaus. That is because the probability of having some Landau levels completely filled while the rest completely empty is practically zero. There is no way one can adjust the electron density with such precision! Also, within such a simplistic picture, it is impossible to explain the variation of the electron density when traversing a Hall plateau. Presently, the accepted view is that, in between the Landau levels, there are localized impurity states, in which case the Fermi level transitions smoothly (instead of jumping)  from one Landau level to another, when increasing the $B$-field or the electron density. The localized spectrum is thought to be generated by a disordered potential, inherently present even in high quality samples. Thus, the presence of disorder is crucial for the explanation of the IQHE. A complete, rigorous theory of the IQHE at zero temperature was derived in a series of papers by Bellissard and his collaborators. In particular, Ref.~\cite{BELLISSARD:1994xj} developed the noncommutative theory of transport that we followed here, together with a set of optimal conditions for the quantization and invariance of the Hall resistance in the presence of strong disorder.

The hypothesis of localized states was experimentally tested in several ways, and the conclusions were always positive. For example, it was observed that the width of the plateaus increases when the mobility of the electrons at zero field decreases. However, too much disorder can lead to the disappearance of the plateaus. Other interesting and relevant experimental observations is that the shape of the curves at the beginning or at the end of a Hall plateau is not always smooth as illustrated in Fig.~\ref{ExpDiagram}. Dips and sometime even oscillations can show up in the experimental curves. There is also a very interesting dependence on temperature of the transport coefficients. One standard observation is the gradual appearance and sharpening of the Hall plateaus as the temperature is lowered. The direct resistance displays an activation behavior which can be understood from the simple shape of the Fermi-Dirac distribution, but when the temperature becomes extremely low a different behavior sets in, determined by the type of disorder and dissipation present in the system. 

\begin{figure}
\center
  \includegraphics[width=10cm]{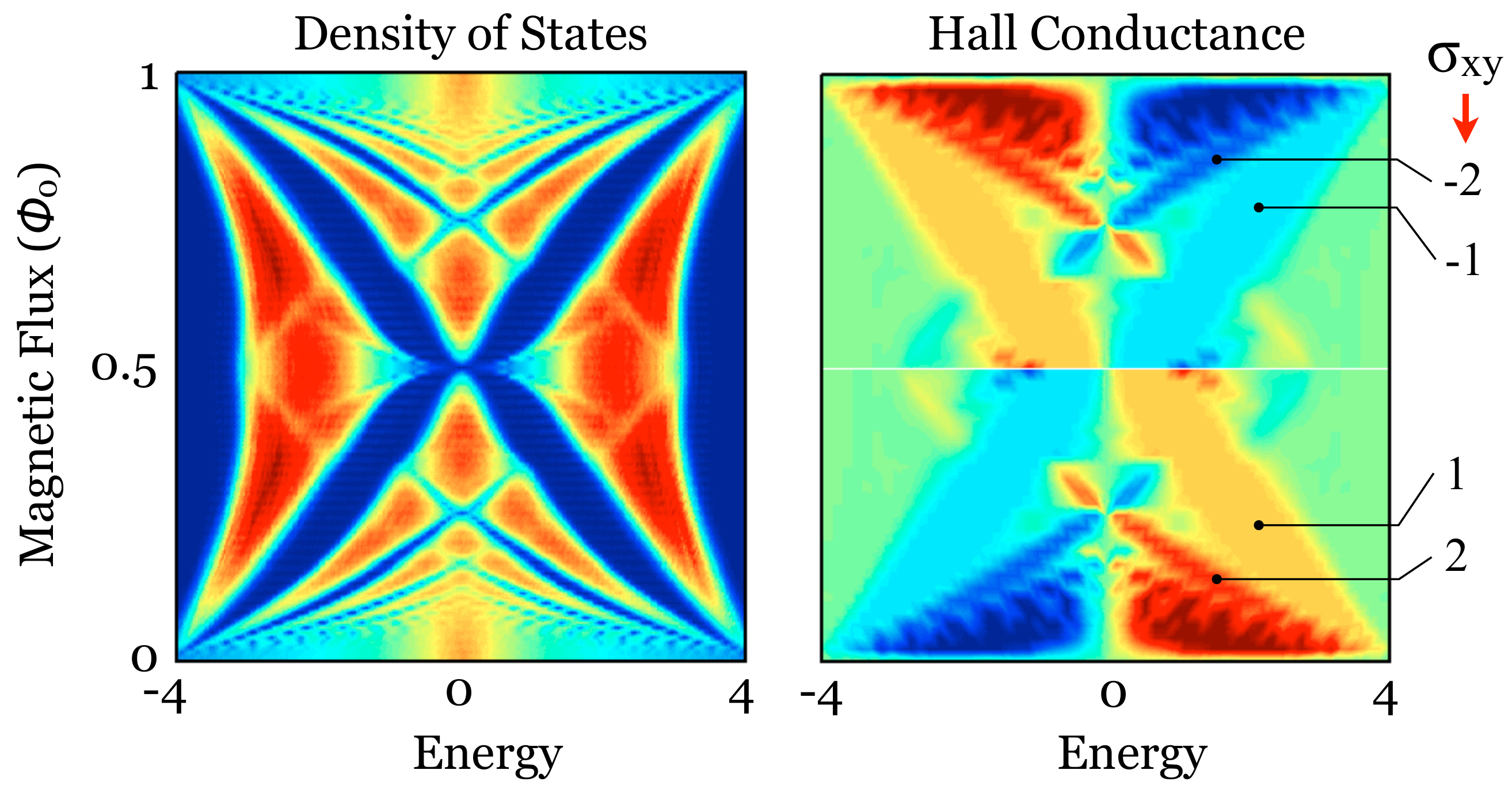}\\
  \caption{DOS (left) and color map of the Hall conductivity (right) for $W=2$. The regions of quantized Hall conductivity, which appear as well defined patches of same color and, are indicated at the right.}
  \label{HallMap2}
\end{figure}

There are many published simulations of the IQHE in the presence of disorder \cite{CzychollSolStComm1988ds,CzychollZPhysB1988jd,MandalPhysB1998as,TanJPhysCondMatt1994rt,ShengPRL1997cy,RochePRB1999te,YangPRB1999de,ShengPRB2001se,SteffenPRB2004bn,ZhouPRB2004ft,MelintePRL2004by,KoshinoPRB2006gh,Sheng:2007pr,MaitiPhysLettA2012gh,Dutta2012cv}. However, many of these simulations are restricted to just the diagonal component of the conductivity tensor. The computation of $\sigma_{12}$ or $\rho_{12}$ is more difficult and, as such, simulations of  $\sigma_{12}$ or $\rho_{12}$ are scarce. Almost all of the existing simulations present $\sigma_{12}$ as function of Fermi energy at fixed magnetic field. Such plots, of course, show wide and well quantized Hall plateaus when the Fermi energy crosses the gaps between the Landau levels.  In experiments, however, what is controlled is the electron density and not the Fermi level. Since the variation of the electron density is very small when the Fermi level crosses the gap between the Landau levels, the wide and well quantized Hall plateaus would appear much narrower if $\sigma_{12}$ was plotted as function of electron density. In fact, it is not clear at all if any of the Hall plateaus observed in these numerical studies would remain visible if the data was plotted as function of electron density. Furthermore, we were not able to find any record of a simulation which analyzes the conductivity or resistivity tensor as function of magnetic field at fixed electron density.   

\begin{table}[htdp]\scriptsize
\caption{The numerical values of $\sigma_{11}$ at $kT=1/\tau_r=0.1$, $W=\phi =0$ and various Fermi energies, obtained with the new algorithm for increasing lattice sizes. The last column displays the value of $\sigma_{11}$ computed with machine precision via Eq.~\ref{Bloch}.}
\begin{center}
\rowcolors{1}{white}{lightgray}
\begin{tabular}{|c|c|c|c|c|c|}
\hline
$E_F$ & $80\times 80$ & $100 \times 100$ & $120 \times 120$ &  $140 \times 140$ & Exact \\
\hline
  0.0   &4.0339628247 & 4.0339630615  & 4.0339630708  &   4.0339630712  & 4.0339630712 \\   
 -0.4   & 3.9394154619    & 3.9394154735   & 3.9394154621 &  3.9394154624   & 3.9394154624     \\
 -0.8   & 3.7040304262   & 3.7040301193   & 3.7040301310 &  3.7040301307   & 3.7040301307     \\
  -1.3   & 3.3684805414     & 3.3684801617   & 3.3684801517 &  3.3684801516  & 3.3684801516     \\
  -1.7    &  2.9522720814   & 2.9522713926 & 2.9522714007 & 2.9522714009   & 2.9522714009     \\
  -2.2   &  2.4678006935     & 2.4678005269   & 2.4678005093 &  2.4678005104  & 2.4678005104     \\
  -2.6   & 1.9239335953    & 1.9239338070  & 1.9239338090 & 1.9239338089   & 1.9239338089     \\
  -3.1   & 1.3274333126   & 1.3274333067  & 1.3274333084 &  1.3274333085   & 1.3274333085     \\
  -3.5    & 0.6854442914    & 0.6854442923 & 0.6854442923 & 0.6854442923   & 0.6854442923     \\
  -4.0   & 0.1086465150  & 0.1086465150  & 0.1086465150 &  0.1086465150 & 0.1086465150 \\
  \hline
\end{tabular}
\end{center}
\label{TableT0p1}
\end{table}%

\begin{table}[htdp]\scriptsize
\caption{The numerical values of $\sigma_{11}$ at $kT=1/\tau_r=0.025$, $W=\phi=0$ and various Fermi energies, obtained with the new algorithm for increasing lattice sizes. The last column displays the value of $\sigma_{11}$ computed with machine precision via Eq.~\ref{Bloch}. The numbers in parentheses represent the relative errors, computed as $|\sigma_{11}(\mathrm{approx})-\sigma_{11}(\mathrm{exact})|/ \sigma_{11}(\mathrm{exact})$.}
\begin{center}
\rowcolors{1}{white}{lightgray}
\begin{tabular}{|c|c|c|c|c|c|c|}
\hline
$E_F$ & $100 \times 100$ & $120 \times 120$ & $140 \times 140$ & $160 \times 160 $ & Exact \\
\hline
  0.0                                     & 16.204612406 (3.7e-5) & 16.204972260 (1.5e-5) &     16.205119575 (6.6e-6) & 16.205181208 (2.8e-6) & 16.205227112\\   
 -$\nicefrac{4}{9}$      & 15.800901211 (8.2e-5) & 15.799191631 (2.5e-5) &     15.799439976 (9.7e-6) & 15.799694877 (6.3e-6) & 15.799593904     \\
 -2 $\times \ \nicefrac{4}{9}$    & 14.845640534 (8.3e-5) &14.847746111 (5.8e-5) &     14.846643808 (1.5e-5) & 14.846886090 (6.1e-7) & 14.846876954     \\
  -3 $\times \ \nicefrac{4}{9}$    & 13.501282967 (1.7e-4) & 13.498925093 (1.3e-6) &    13.498688837 (1.6e-5) & 13.498843113 (4.6e-6) & 13.498906270     \\
  -4 $\times \ \nicefrac{4}{9}$    & 11.830160006 (1.1e-5) & 11.830983739 (5.8e-5) &    11.830916078 (5.2e-5) & 11.830155334 (1.1e-5) & 11.830294392     \\
  -5 $\times \ \nicefrac{4}{9}$    & 9.8929640489 (3.3e-4) & 9.8886679730 (1.0e-4) &  9.8899094634 (2.2e-5)  & 9.8896881580 (4.7e-7) & 9.8896834893     \\
  -6 $\times \ \nicefrac{4}{9}$     & 7.7119398206 (2.4e-5) & 7.7130531856 (1.1e-4) &  7.7127298903 (7.7e-5)  & 7.7123144240 (2.3e-5) & 7.7121326232     \\
  -7 $\times \ \nicefrac{4}{9}$    & 5.3242604491  (2.4e-5) & 5.3236541115 (1.3e-4) &  5.3241123103  (5.2e-5)  & 5.3243169737 (1.3e-5) & 5.3243912832     \\
  -8 $\times \ \nicefrac{4}{9}$  & 2.7469871332 (2.4e-4) & 2.7475442292 (4.0e-5) &  2.7476748910 (7.4e-6) & 2.7476680440 (4.9e-6) & 2.7476545340     \\
  -9 $\times \ \nicefrac{4}{9}$  & 0.1099066152 (2.8e-9) &0.1099066156 (0e-10) &  0.1099066156 (0e-10) & 0.1099066156 (0e-10) & 0.1099066156 \\
  \hline
\end{tabular}
\end{center}
\label{TableT0p025}
\end{table}%

\begin{table}[h]\scriptsize
\caption{The numerical values of $\sigma_{11}$ at $kT=1/\tau_r=0.01$, $W= \phi =0$ and various Fermi energies, obtained with the new algorithm for increasing lattice sizes. The last column displays the value of $\sigma_{11}$ computed with machine precision via Eq.~\ref{Bloch}. The numbers in parentheses represent the relative errors, computed as $|\sigma_{11}(\mathrm{approx})-\sigma_{11}(\mathrm{exact})|/ \sigma_{11}(\mathrm{exact})$.}
\begin{center}
\rowcolors{1}{white}{lightgray}
\begin{tabular}{|c|c|c|c|c|c|}
\hline
$E_F$  & $100 \times 100$ & $120 \times 120$ & $140 \times 140 $ & $160 \times 160 $ & Exact \\
\hline
  0.0                                                     & 40.520686410 (1.2e-4) & 40.522558677 (7.5e-5) &  40.523641888 (4.8e-5) & 40.524308038 (3.2e-5) & 40.525626855 \\    
 -1 $\times \ \nicefrac{4}{9}$      & 39.517833421 (3.3e-4) & 39.497150004 (1.9e-4)   & 39.499629103 (1.2e-4) & 39.508996214 (1.0e-4) & 39.504691905     \\
 -2 $\times \ \nicefrac{4}{9}$      & 37.110642808 (2.9e-4) & 37.138820375 (4.6e-4) &     37.114579007 (1.8e-4) & 37.120997917 (1.2e-5) & 37.121462109     \\
  -3 $\times \ \nicefrac{4}{9}$    & 33.777005514 (7.7e-4)   & 33.748897618 (5.4e-5) &    33.742729517 (2.3e-4) & 33.747792038 (8.7e-5) & 33.750748672     \\
  -4 $\times \ \nicefrac{4}{9}$    & 29.585203369 (2.1e-4) & 29.596597708 (6.0e-4) &    29.596944121 (6.1e-4) & 29.573808787 (1.6e-4) & 29.578697042     \\
  -5 $\times \ \nicefrac{4}{9}$         & 24.787127553 (2.4e-3)   & 24.713893823 (5.2e-4) &  24.729175598 (9.6e-5)  & 24.724104241 (1.0e-4) & 24.726791044     \\
  -6 $\times \ \nicefrac{4}{9}$      & 19.265355553 (8.9e-4)  & 19.304253858 (1.1e-3) &  19.305176839 (1.1e-3)  & 19.288293494 (2.9e-4) & 19.282623910     \\
  -7 $\times \ \nicefrac{4}{9}$     & 13.320411822 (5.5e-4)  & 13.288432535 (1.8e-3) &  13.302525433 (7.8e-4)   & 13.305757131 (5.4e-4) & 13.313040334     \\
  -8 $\times \ \nicefrac{4}{9}$       & 6.8441390250 (3.9e-3) & 6.8590106166 (1.7e-3) &  6.8736981172 (3.9e-4)  & 6.8765917049 (8.1e-4) & 6.8710107299     \\
  -9 $\times \ \nicefrac{4}{9}$    & 0.1101512994 (2.3e-5)  & 0.1101540276 (1.7e-6) & 0.1101538290 (7.9e-8)  & 0.1101538378 (0e-10) & 0.1101538378  \\
  \hline
\end{tabular}
\end{center}
\label{TableT0p01}
\end{table}%

We now present what we were able to achieve (so far) with the formalism presented in the previous sections for the model Hamiltonian of Eq.~\ref{MainModel}. As we already mentioned, at the algebra level, this Hamiltonian is generated by $h$ which takes the values $h(\omega,{\bm n})=1$ if $|{\bm n}|=1$, $h(\omega,{\bm 0})=W \omega_{\bm 0}$, and zero in rest. The calculations will be carried out on a finite lattice which is wrapped into a torus. The size of the finite lattice will be indicated as $Nr\times Nr$, where $Nr$ is the total number of nodes in one direction. In the old notation, $Nr$ would be given by $2{\cal N}+1$ but here we relax the requirement that $Nr$ be an odd number. In two dimensions ($d=2$), the matrix ${\bm F}$ is simply $F_{i,j}=\epsilon_{ij3}\phi$, where $\phi$ is the magnetic flux through the repeating cell, expressed in units of flux quantum $\phi_0=h/e$. The Hamiltonians at fluxes $\phi$ and $\phi+1$ are unitarily equivalent so we can restrict the values of $\phi$ to the interval $[0,1]$. Since the calculations are carried out on the torus, the magnetic flux will always be quantized in our calculations according to the rule given in Eq.~\ref{Q}. As such, the larger the torus the more magnetic field values we can sample.

\begin{figure}
\center
  \includegraphics[width=10cm]{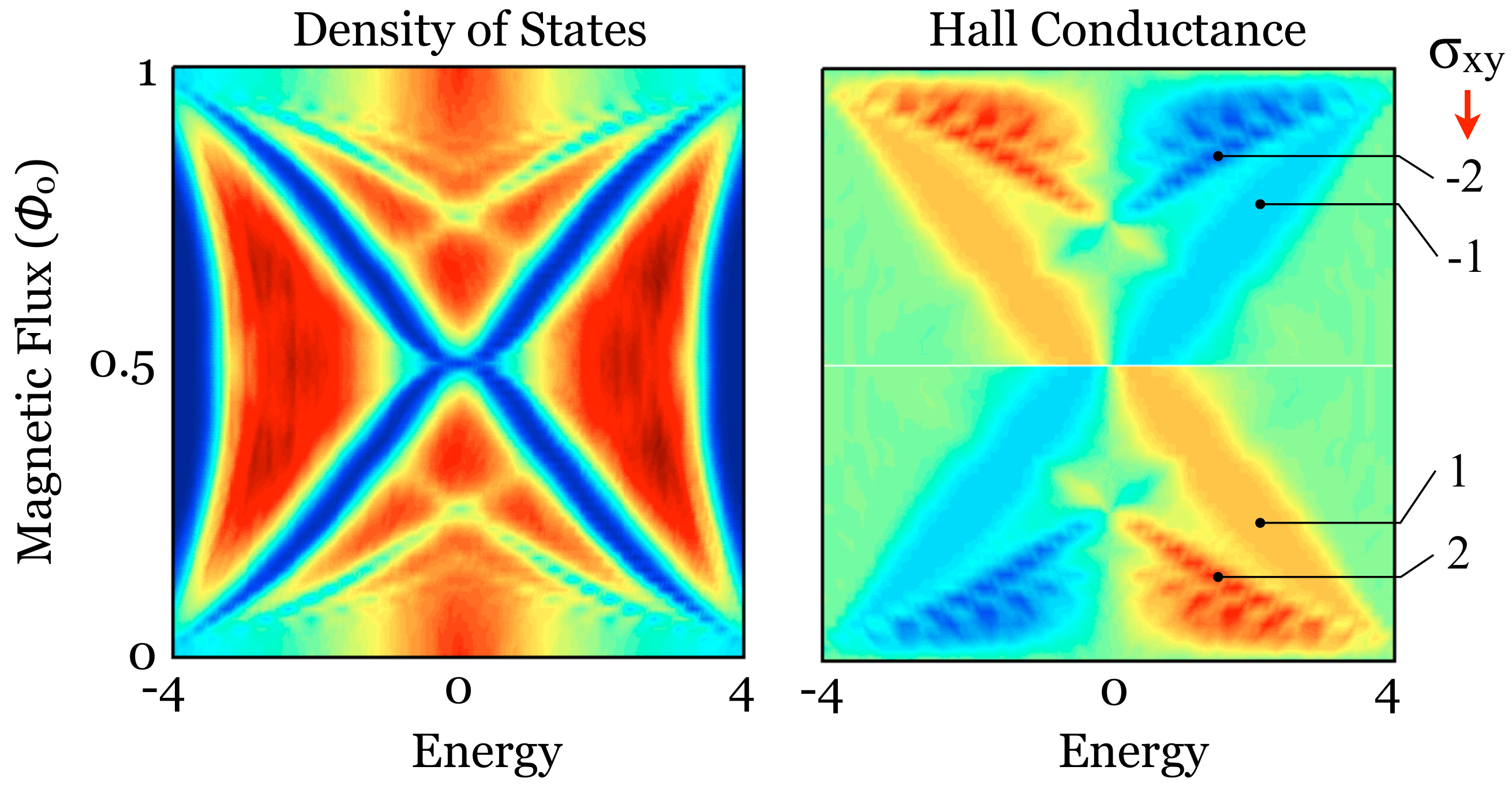}\\
  \caption{DOS (left) and color map of the Hall conductivity (right) for $W=3$. The regions of quantized Hall conductivity, which appear as well defined patches of same color and, are indicated at the right.}
 \label{HallMap3}
\end{figure}

We start with some qualitative analysis. The spectrum of $h$, obtained on a $120\times 120$ lattice, is shown in Fig.~\ref{Spectrum} for three values of $W$: 0, 2 and 3. The spectrum is represented via the smoothed density of states (DOS), defined as:
\begin{equation}
\begin{array}{c}
D(\epsilon)=\frac{1}{\pi}{\cal T}\left (\mathrm{Im}(h-\epsilon-i\delta)^{-1}\right ),
\end{array}
\end{equation}
where $\delta$ was given the small value $0.01$. The DOS is plotted as an intensity map in the ($\epsilon$,$\phi$) plane. The plot at $W=0$ reproduces the famous Hofstadter butterfly \cite{HofstadterPRB1976km}. This can serve as a reminder that the simple Landau level picture becomes more complicated when a periodic potential is present. The Hofstadter spectrum has a fractal structure with the patterns seen in the first panel of Fig.~\ref{Spectrum} repeating over and over when zooming to finer and finer scales. At weak magnetic fields, and especially towards the edges of the energy spectrum, one can identify well defined Landau bands (a more careful scrutiny will reveal that these bands are made of many sub-bands). If the Fermi level is somehow fixed in between these bands, then one expects quantized values of $\sigma_{12}$ in units of $G_0$. The DOS at $W=2$ and $W=3$ look like a blurred version of the DOS at $W=0$. The blur originates from the impurity states generated inside the empty spaces between the Landau bands by the disordered potential. Maps of the Hall conductivity are reported in Figs.~\ref{HallMap1}, \ref{HallMap2} and \ref{HallMap3} for $W=0$, 2, and 3 cases, respectively. These calculations were performed on a $100 \times 100$ lattice with $kT=1/\tau_r=0.01$. A single disorder configuration was used, since the fluctuations due to the disorder will not have any visible effects for the intensity maps of Figs.~\ref{Spectrum}, \ref{HallMap1}, \ref{HallMap2} and \ref{HallMap3}. The Fermi energy was varied over the entire energy spectrum, which was sampled at 60 equally spaced points. In all three cases, the maps display regions where $\sigma_{12}$ takes the expected quantized values. In the middle of these regions, the quantization occurs with five digits of precision or better. For the cases with disorder, it is important to notice that this quantization occurs beyond the regions empty of spectrum, that is, in the regions where $E_F$ is emersed in the localized spectrum. The Hall conductivity maps from Figs.~\ref{HallMap1}, \ref{HallMap2} and \ref{HallMap3} give a useful panoramic view of the Hall conductivity but, even though we see sharp quantized values, the results are actually far from demonstrating the IQHE in our model. That is because in the real experiments, there is no control over the Fermi level. What is being controlled is the density of electrons and the Fermi level adjusts itself correspondingly. What it is known for sure is that the Fermi level will never stay in the regions empty of spectrum. As such, to actually demonstrate the occurrence of the IQHE in our model, we need to map the transport coefficients as function of density while holding the magnetic flux fixed, or as function of magnetic flux while holding the electron density fixed, allowing the Fermi level to adjust itself to the appropriate value. Revealing the Hall plateaus in these new conditions is a much more difficult task. At last, it is interesting to note the similarity of the Hall conductivity map at $W=0$ in Fig.~\ref{HallMap1} and the Chern number map computed in Ref~\cite{OsadchyJMP2001vb}. There are expected differences between the two results because the map in Fig.~\ref{HallMap1} was computed at a finite temperature while the Chern number map gives the Hall conductivity at zero temperature.

\begin{figure}
\center
  \includegraphics[width=11cm]{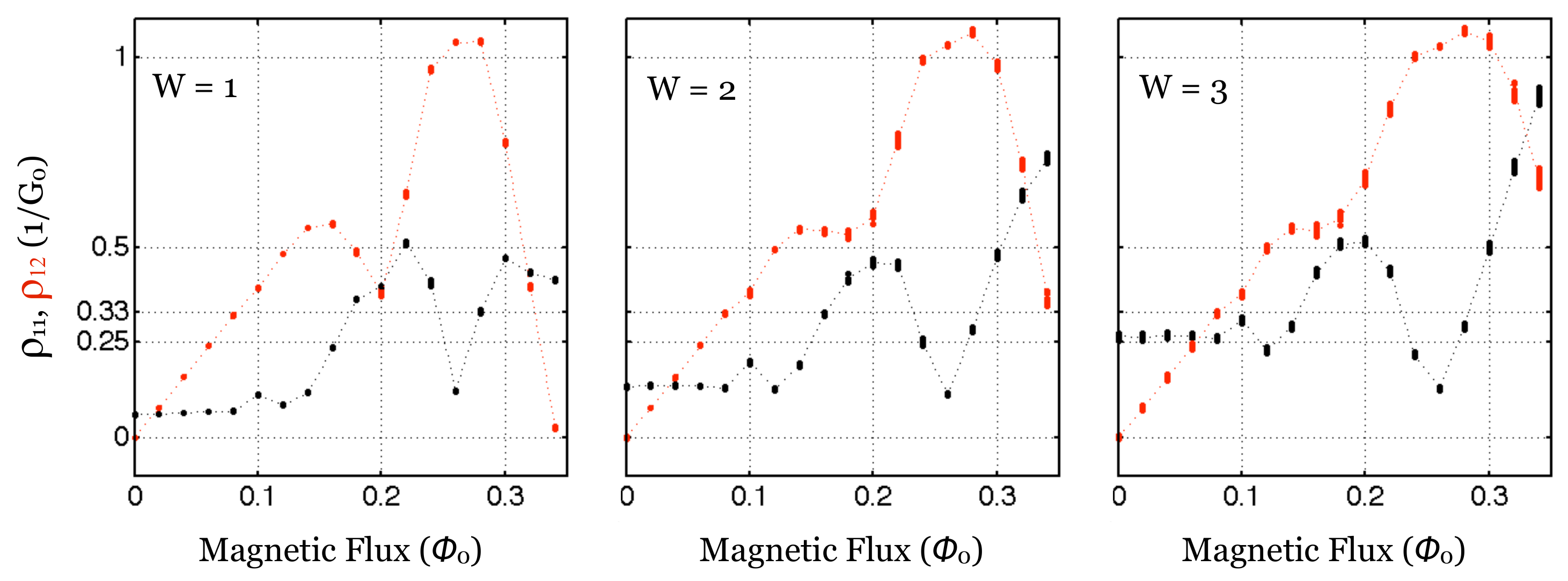}\\
  \caption{Diagonal and Hall resistivities computed with Eq.~\ref{RPracticalKubo} for 10 random disorder configurations, at disorder strengths $W=1$, 2 and 3. The simulation was performed on a $100 \times 100$ lattice with $kT=1/\tau_r=0.1$. The electron density was fixed at $n_e=0.25$. The standard deviation was observed to be less than 2\%.}
 \label{Average}
\end{figure}

At this point we start our quantitative analysis. First, we present a simple but direct test of the numerical algorithm. In the absence of disorder and magnetic field, the Kubo formula takes the extremely simple form:
\begin{equation}\label{Bloch}
\sigma_{11}=\frac{1}{2 \pi kT/\tau_r}\int_{-\pi}^\pi dk_1 \int_{-\pi}^\pi dk_2 \frac{\sin^2 k_1}{\cosh^2\left[\frac{1}{2}(E_{\bm k}-E_F)/kT\right]},
\end{equation} 
where $E_{\bm k}$ is the band energy of the model: $E_{\bm k}=2\cos k_1 + 2 \cos k_2$. This formula can be computed with the machine precision and we will call such a computation ``exact." In Tables~\ref{TableT0p1}, \ref{TableT0p025} and \ref{TableT0p01}, we report a comparison between the exact value of $\sigma_{11}$ and the values obtained with the new algorithm on lattices of varying sizes.  In Table~\ref{TableT0p1}, we fixed the temperature and dissipation at relatively large values: $kT=1/\tau_r=0.1$. Four lattice sizes were considered: $80\times 80$, $100 \times 100$, $120\times 120$ and $140 \times 140$. The Fermi level was varied from the bottom to the middle of the energy spectrum. By examining the columns of Table~\ref{TableT0p1}, we see that, already for the smallest lattice size, the output from the new algorithm displays 7 digits of precision for all $E_F$'s. As predicted, the degree of precisions increases rapidly with the lattice size and for the largest size we see 11 digits of precision for all $E_F$'s. This indicates that all the lattice sizes appearing in Table~\ref{TableT0p1} are in the asymptotic regime where the exponentially-fast convergence occurs. In Table~\ref{TableT0p025}, the temperature and dissipation were fixed at lower values: $kT=1/\tau_r=0.025$. Four lattice sizes were considered: $100\times 100$, $120 \times 120$, $140\times 140$ and $160 \times 160$. In this table too, we see an improvement of the accuracy as the size of the system is increased, though not as dramatic as in Table~\ref{TableT0p1}. Looking at the relative errors, we see that, on average (with respect to $E_F$), the accuracy was improved by more than one order of magnitude when the lattice size was increased from $100 \times 100$ to $160 \times 160$. For the largest lattice size, there are more than 5 digits of precision. Based on the monotonic reduction of the relative errors from one column to the other (with one or two exceptions), we believe that the asymptotic regime where the exponentially fast convergence occurs has been already reached in Table~\ref{TableT0p025}. In Table~\ref{TableT0p01}, the temperature and dissipation were fixed at even lower values: $kT=1/\tau_r=0.01$. Here we still see a weak improvement of the accuracy with the size of the system. The precision here is about 4 digits for the largest two lattices. Based on the information contained in these tables, we expect the algorithm to be well converged on the $120 \times 120$ or $140 \times 140$ lattices when $kT$ and $1/\tau_r$ are equal to or larger than 0.025, and we expect reasonable outputs even for  $kT=1/\tau_r=0.01$. We will always verify these statements explicitly, by repeating the calculations for increasing lattice sizes.

\begin{figure}
\center
  \includegraphics[width=11cm]{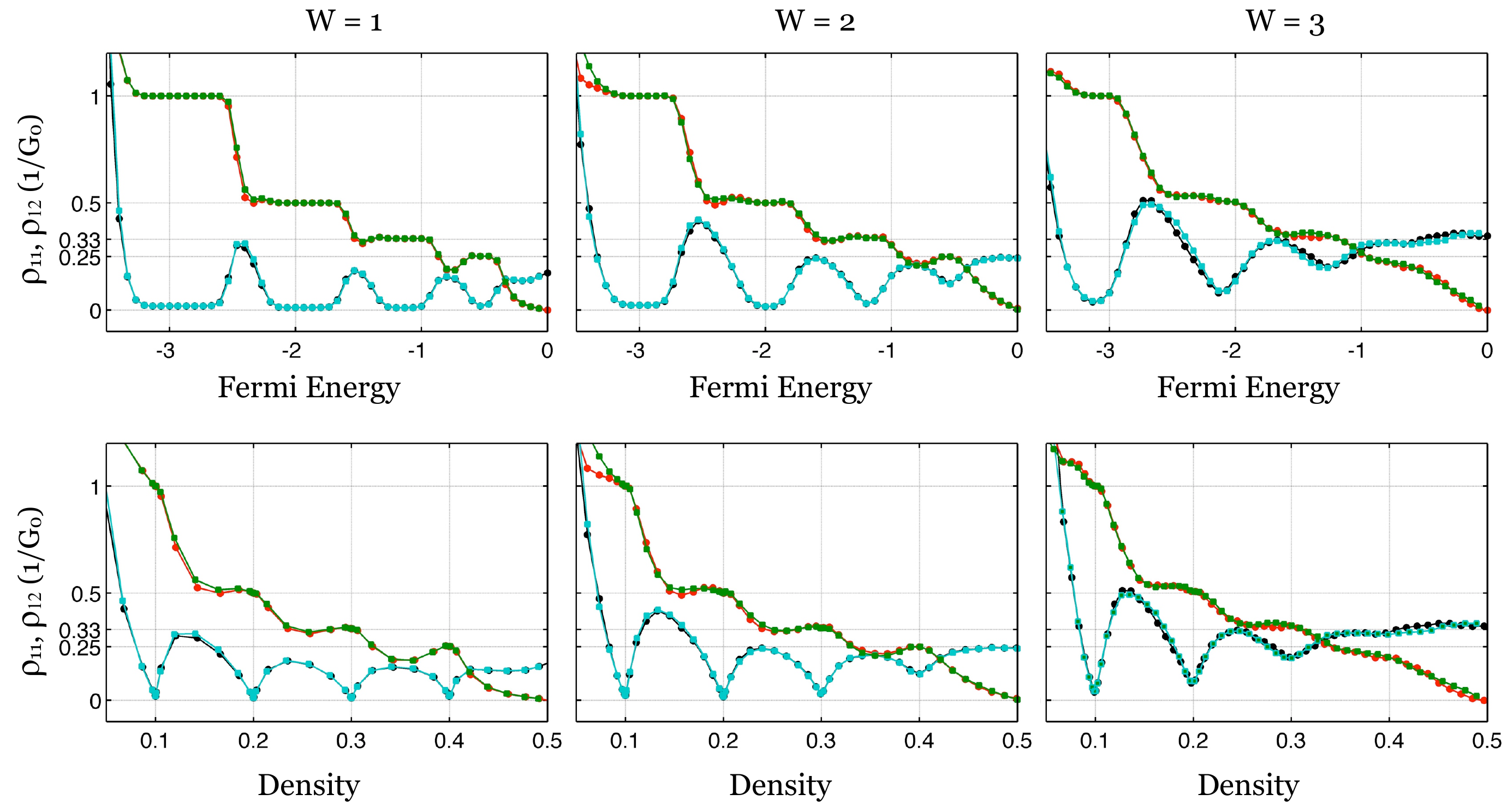}
  \caption{First row (Second row): The diagonal and the Hall resistivities as function of Fermi energy (density) at fixed magnetic flux $\phi=0.1\phi_0$, for $kT=1/\tau_r=0.025$ and disorder strengths $W=1$, 2, 3. Each panel compares the data obtained on the $100 \times 100$ lattice (circles) and on the $120 \times 120$ lattice (squares).}
 \label{EfDenComp100vs120BF0p1T0p025}
\end{figure}

Next we discuss the issue of averaging over disorder. The integrand in Eq.~\ref{PracticalKubo}:
\begin{equation}\label{RPracticalKubo}
 \  \frac{1}{|\mathbb{T}_d|}\sum_{a,b=1}^{|\mathbb{T}_d|}\frac{\langle \phi_b^\omega|\tilde{\pi}_\omega(\tilde{\partial}\tilde{h})|\phi_a^\omega \rangle \langle \phi_a^\omega | \tilde{\pi}_\omega (\tilde{\partial} \Phi_{\mathrm{FD}}(\tilde{h})|\phi_b^\omega \rangle}{1/\tau + i(\epsilon_a-\epsilon_b)},
\end{equation}
 is self-averaging, meaning it converges to its disorder-average as the size of the torus is taken larger and larger. This is an important observation because it practically solves the issue of disorder average. For the sizes considered in this study, we can demonstrate that indeed, the integrand is practically independent of the particular disorder configuration, so we can eliminate the disorder-average step entirely. Fig.~\ref{Average} presents a computation of $\rho_{11}$ and of $\rho_{12}$ as functions of the magnetic flux for 10 randomly generated disorder configurations, and for three disorder strengths, $W=1$, 2 and 3. The calculation was performed on a $100 \times 100$ lattice at fixed electron density $n_e=0.25$, and at $kT=1/\tau_r=0.1$. The standard deviation of the fluctuations due to disorder is less than 2\% for the data reported in Fig.~\ref{Average}. This calculation confirms that, starting from this lattice size, the fluctuations of the resistivities from one disorder configuration to another are very small, at least for this temperature and dissipation strength. From now on, the transport coefficients will be computed with just one disorder configuration, using the formula from Eq.~\ref{RPracticalKubo}, and we will visually estimate the fluctuations due to the disorder by examining the data for different lattice sizes.  
 
 \begin{figure}
\center
  \includegraphics[width=11cm]{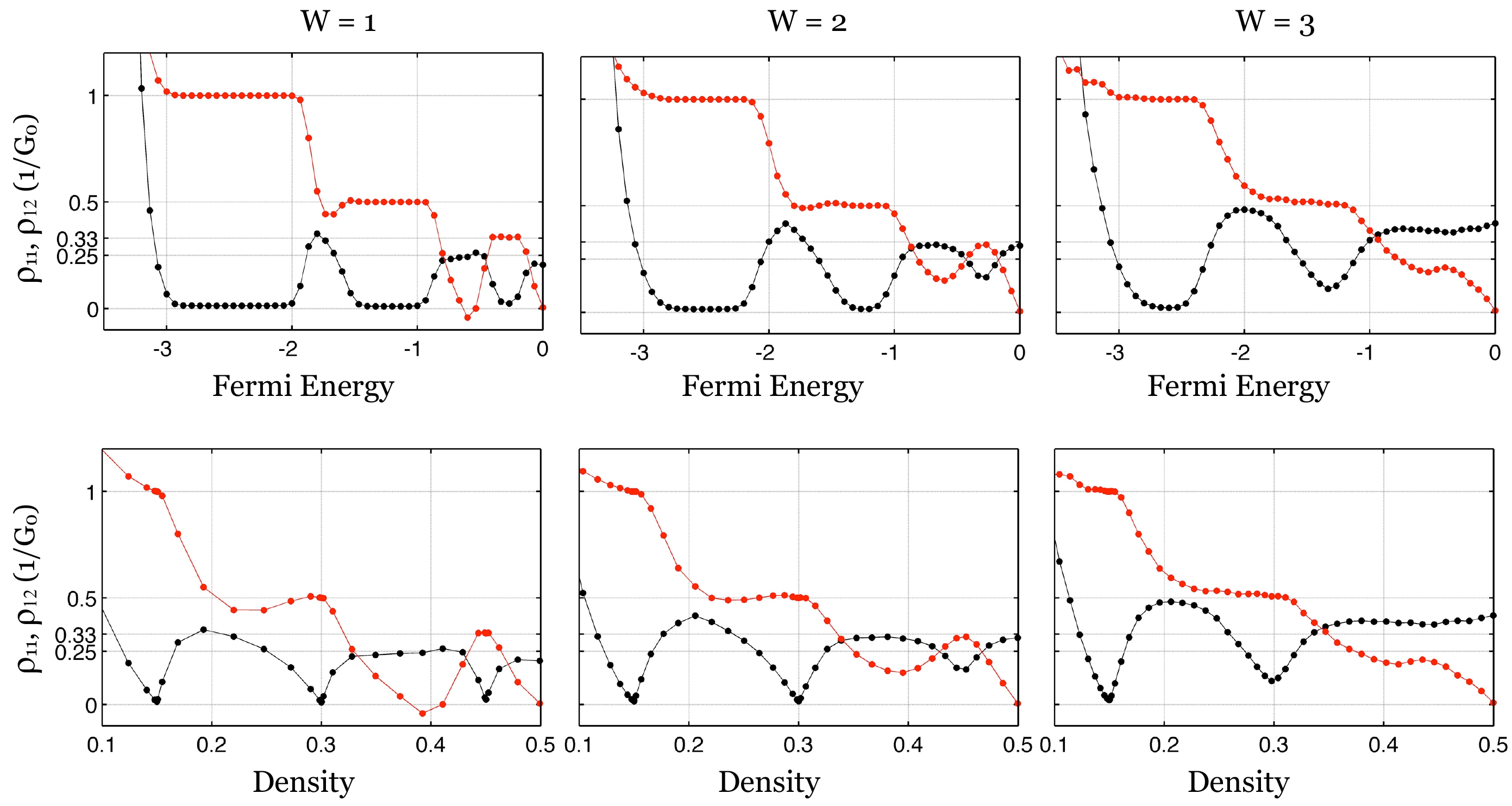}
  \caption{First row (Second row): The diagonal and the Hall resistivities as function of Fermi energy (density) at fixed magnetic flux $\phi=0.15\phi_0$, for $kT=1/\tau_r=0.025$ and disorder strengths $W=1$, 2, 3. For clarity, we show only the data obtained on the $120 \times 120$ lattice.}
 \label{EfDenComp100vs120BF0p15T0p025}
\end{figure}

We now present the actual simulation of the IQHE. The first row of Fig.~\ref{EfDenComp100vs120BF0p1T0p025} reports $\rho_{11}$ and $\rho_{12}$ as function of Fermi energy, with the magnetic flux hold fixed at $0.1\phi_0$ and for $kT=1/\tau_r=0.025$. The calculations were performed on the $100 \times 100$ lattice and then repeated on the $120 \times 120$ lattice, for three disorder strengths: $W=1$, 2 and 3. The most important observation here is that the data corresponding to the two lattice sizes overlap almost perfectly for all three disorder strengths, indicating a very good convergence of the results. Also, there are no visible fluctuations due to disorder. For $W=1$ and $W=2$ we can see four well defined and quantized Hall plateaus. The diagonal resistivity takes very small values in the region Hall plateau regions. For $W=3$, only the first two Hall plateaus display good quantization. The same data on the resistivity tensor is plotted in the second row of Fig.~\ref{EfDenComp100vs120BF0p1T0p025} as function of density $ n_e={\cal T}\left(1/(1+e^{(h-E_F)/kT})\right)$. Note that $n_e$ varies from 0 (no electrons) to 1 (one electron per lattice site). As we already pointed out, the quantized Hall plateaus are drastically reduced in this representation. The disorder seems to slightly widen the Hall plateaus, especially when one compares $W=1$ and $W=2$ plots, but we do not see the anticipated dramatic effects of the disorder. Calculations on even larger lattices (see the following discussion) makes us confident that the calculations shown in Fig.~\ref{EfDenComp100vs120BF0p1T0p025} are converged. Additional calculations will also show that there are no visible changes as the temperature and dissipation are lowered, so the narrow plateaus seen in our simulations must be a characteristic of the model at the particular value of the magnetic flux considered here. Looking back at Figs.~\ref{HallMap2} and \ref{HallMap3}, one can see an increase of the spreading of the Landau bands with the disorder when the magnetic flux approaches $0.5 \phi_0$, so we have repeated the calculations by fixing the magnetic flux at the larger value $\phi=0.15\phi_0$, in a hope to improve outcome.  The results are shown in Fig.~\ref{EfDenComp100vs120BF0p15T0p025} and, indeed, here we can see the first and second quantized Hall plateaus forming when the disorder is increased, even when the data is plotted as function of electron density. Still, we think there is a lot of work to be done until we will see a quantitative similarity with the experimental data. To our knowledge, this has been never achieved.

\begin{figure}
\center
  \includegraphics[width=10cm]{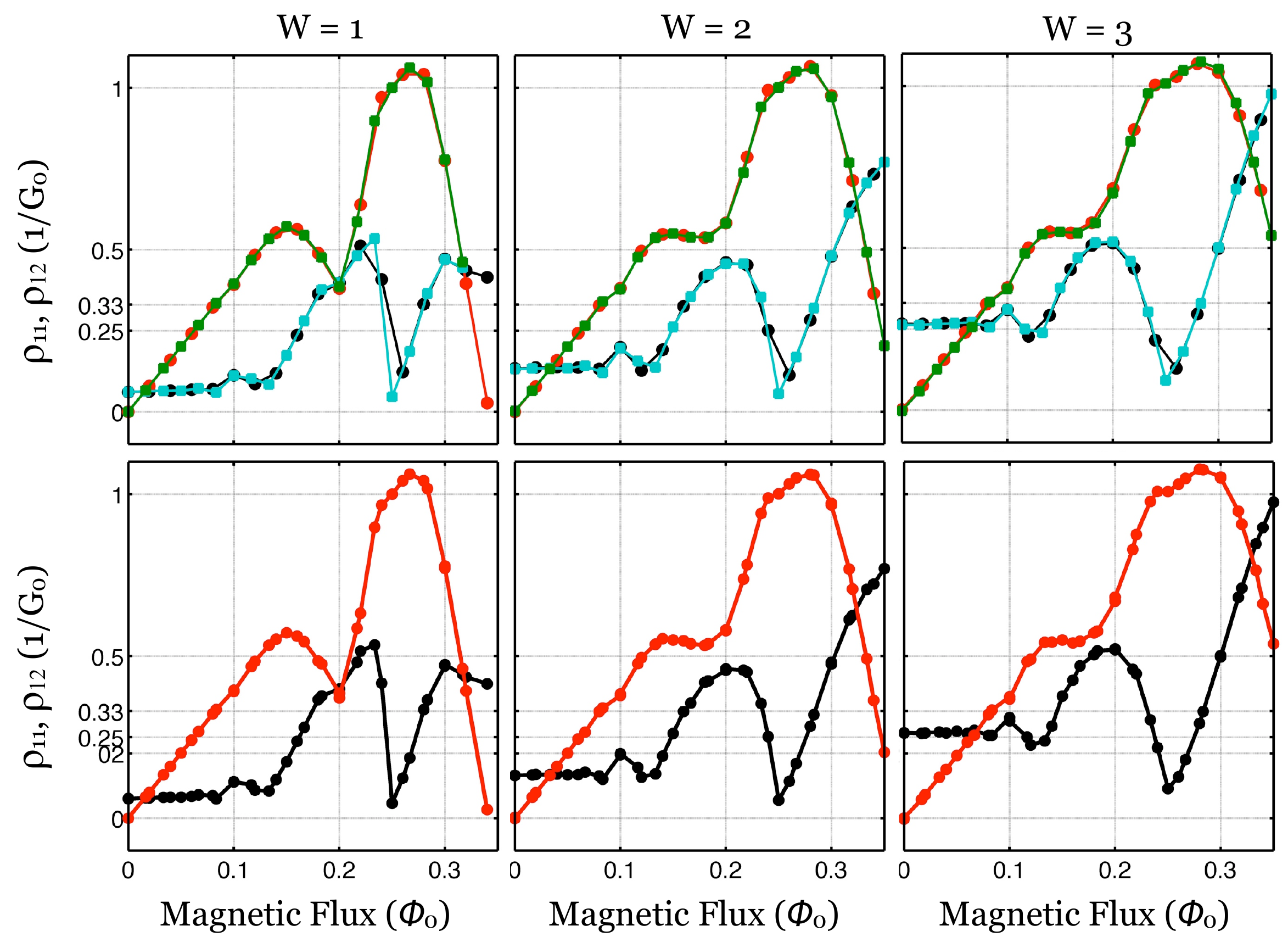}\\
  \caption{The diagonal and the Hall resistivities as function of magnetic flux, at $kT=1/\tau_r=0.1$, $n_e=0.25$ and disorder strengths $W=1$, 2, 3. The first row shows a comparison between the data obtained on the $100 \times 100$ lattice (circles) and on the $120 \times 120$ lattice (squares). The two sets of data shown in the first row were joined together and plotted in the second row.}
 \label{CompT0p1}
\end{figure}

Next, we fix the electron density and vary the magnetic flux. Fig.~\ref{CompT0p1} reports $\rho_{11}$ and $\rho_{12}$ as functions of the magnetic flux, with the electron density hold fixed at 0.25 and for $kT=1/\tau_r=0.1$. The calculations were performed on the $100 \times 100$ lattice and then it was repeated on the $120 \times 120$ lattice, for three disorder strengths: $W=1$, 2 and 3. The results are compared in the first row of Fig.~\ref{CompT0p1}. We see again a very good overlap between the data corresponding to the two lattice sizes, for all three disorder strengths, and there are no visible fluctuations due to disorder. For each $W$ values, we can then gather together the data obtained on the two lattices together and practically double the magnetic flux sampling values. The resulting data is plotted in the second row of Fig.~\ref{CompT0p1}, and we believe these graphs are quite accurate representation of the resistivity tensor of the model system, at this temperature and degree of dissipation. The features seen in these graphs, especially those for $W=2$ and 3, resemble qualitatively the experimental IQHE curves at higher temperatures, with the Hall plateaus just about to form and with the diagonal resistivity dipping in the plateau regions.

\begin{figure}
\center
  \includegraphics[width=10cm]{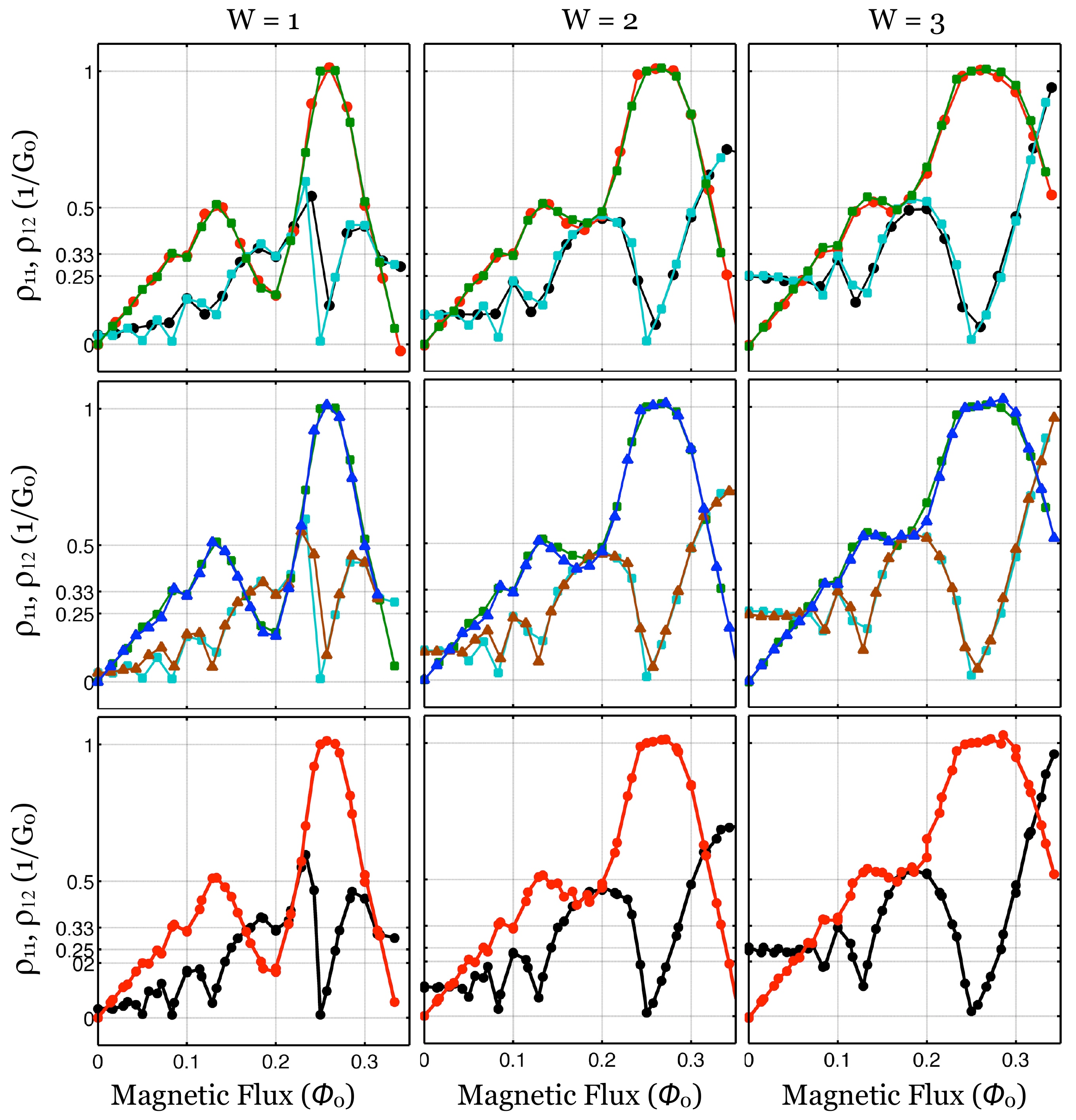}\\
  \caption{The diagonal and the Hall resistivities as function of magnetic flux, at $kT=1/\tau_r=0.025$, $n_e=0.25$ and disorder strengths $W=1$, 2, 3. The first row shows a comparison between the data obtained on the $100 \times 100$ lattice (circles) and on the $120 \times 120$ lattice (squares). The second row shows a comparison between the data obtained on the $120 \times 120$ lattice (squares) and on the $140 \times 140$ lattice (triangles).The two sets of data shown in the second row were joined together and plotted again in the third row.}
 \label{CompT0p025}
\end{figure}

In Fig.~\ref{CompT0p025}, the temperature and dissipation were reduced to $kT=1/\tau_r=0.025$. All the other parameters were held at the same values as above. The first row presents a comparison between the data obtained on the $100 \times 100$ and $120 \times 120$ lattices, and the second row presents a comparison between the data obtained on the $120 \times 120$ and $140 \times 140$ lattices. The overlap between the two sets of data is already very good in the first row, and it is almost perfect in the second row. The visible discrepancies are due to different sampling of the magnetic flux. We join together the data from the $120 \times 120$ and $140 \times 140$ lattices in order to increase the magnetic flux sampling values, and the plot of the resulting data is shown in the third row of Fig~\ref{CompT0p025}. Examining the Hall resistivity, we see the first Hall plateau to be very well defined and relatively flat, at all three disorder strengths. The first plateau is seen to widen out as the disorder strength is increased. The calculations also resolved higher order Hall plateaus. For $W=1$ and 2, we distinguish as many as five Hall plateaus. For all these five Hall plateaus (excepting the 4th one), the diagonal resistivity dips to very low values, as expected. The fifth and fourth Hall plateaus are washed out when disorder is increased to $W=3$.

\begin{figure}
\center
  \includegraphics[width=10cm]{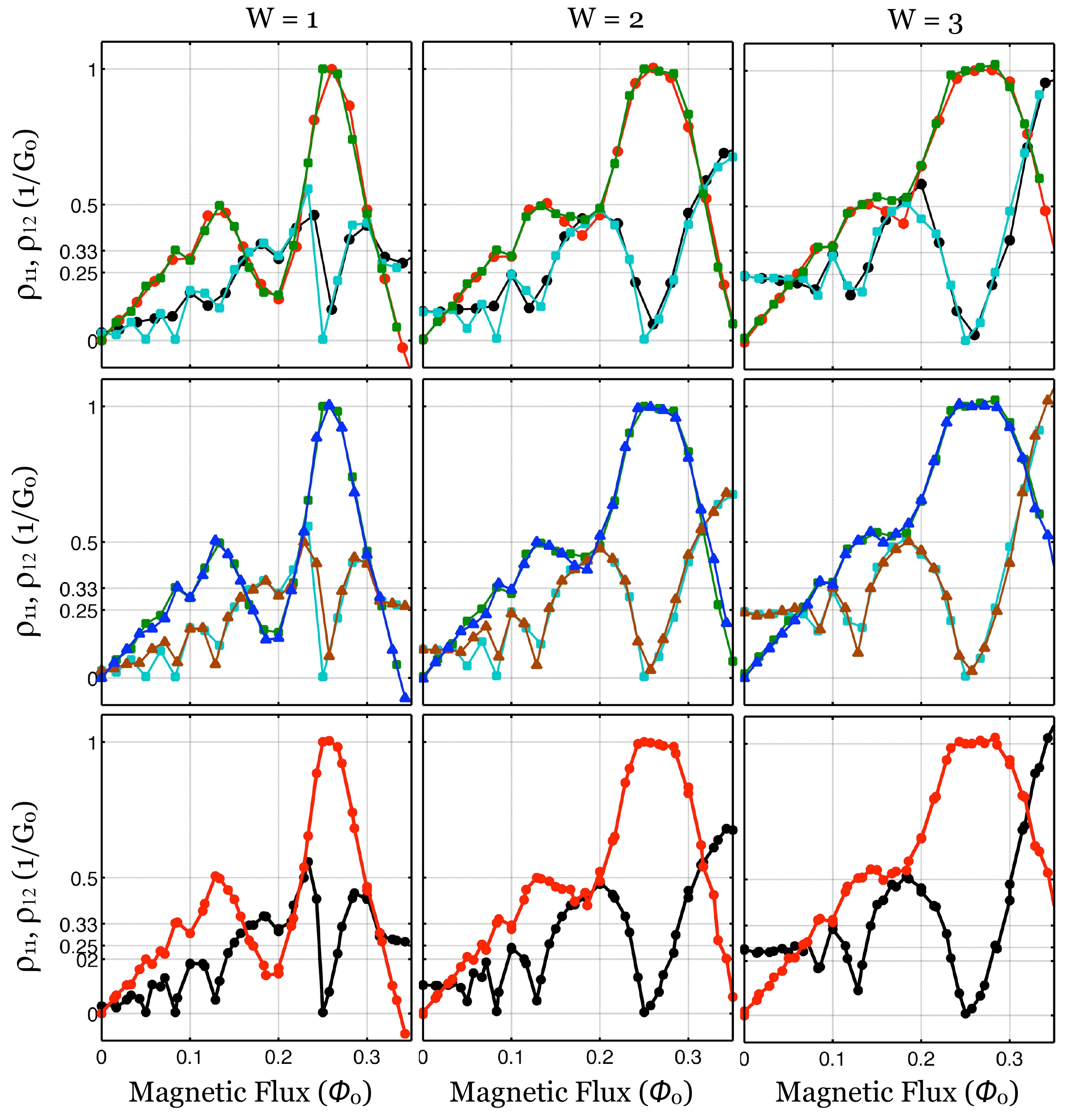}\\
  \caption{The diagonal and the Hall resistivities as function of magnetic flux, at $kT=1/\tau_r=0.001$ and disorder strengths $W=1$, 2, 3. The first row shows a comparison between the data obtained on the $100 \times 100$ lattice (circles) and on the $120 \times 120$ lattice (squares). The second row shows a comparison between the data obtained on the $120 \times 120$ lattice (squares) and on the $140 \times 140$ lattice (triangles).The two sets of data shown in the second row were joined together and plotted again in the third row.}
 \label{CompT0p01}
\end{figure}

Finally, in Fig.~\ref{CompT0p01}, the temperature and the dissipation were reduced even more to $kT=1/\tau_r=0.01$. The layout of this figure is the same as for Fig.~\ref{CompT0p025} and there are practically no major changes from the previous case, except that the first Hall plateau became sharper. In fact, a direct comparison shows a good quantitative similarity between the data at $kT=1/\tau_r=0.025$ and $kT=1/\tau_r=0.01$, when the graphs were visually inspected.    

\section{Conclusions}

The framework of the $C^*$-algebras and noncommutative calculus enabled us to develop a ``canonical" finite-volume approximation to the exact noncommutative Kubo formula with dissipation. This approximate formula is amenable on a computer. Within the same framework and for the relaxation time approximation for the dissipation effects, we were able to establish rigorous error bounds for this approximate formula, and to demonstrate that the errors vanish exponentially fast in the thermodynamic limit. Given the general upper bounds derived here for generic correlation functions, we believe that this last conclusion also applies to situations beyond the relaxation time approximation. The $C^*$-algebra formalism was absolutely instrumental for both aspects of our analysis, namely, for the derivation of the approximate Kubo formula and for the analysis of the errors. We believe that the formalism of $C^*$-algebras will also lead to better numerical implementations of the finite-volume approximate Kubo formula. 

Based on the abstract theoretical ideas, we devised a numerical algorithm for the computation of the conductivity tensor for disordered systems under magnetic fields. The main computational advantage of the algorithm is that it does not require twisted boundary conditions and integration over a Brillouin torus in order to achieve exponentially fast convergence for the thermodynamic limit. The time saved by not having to repeat the calculations for many different twisted boundary conditions was fully re-invested to increase the size of system, leading to a substantial boost of the simulation box over the existing simulations.   

We presented an application to the Integer Quantum Hall Effect. Working with a lattice model of a disordered 2-dimensional electron gas in perpendicular magnetic fields, we were able to compute the resistivity tensor on lattices containing as many as $140 \times 140$ sites, placing the calculations quite close to their thermodynamic limit for the temperatures and the degree of dissipation considered in our study. We computed the resistivity tensor while varying the Fermi energy (or equivalently the electron density) while holding the magnetic field constant, and as function of the magnetic flux while holding the electron density fixed. We were able to resolve as many as five Hall plateaus, with the first plateau appearing sharp and well quantized as the temperature and dissipation is reduced.  We could also witness how the first and second plateaus become wider and wider as the disorder strength is was being increased. To our knowledge, this is the first time such simulation results have been reported.

\ack
This research was supported by the U.S. NSF grants DMS-1066045 and DMR-1056168.

\bibliographystyle{spmpsci}

\end{document}